\documentclass[11pt]{article}%
\usepackage{setspace}
\usepackage{amssymb}
\usepackage{enumerate}
\usepackage{graphicx}
\usepackage{amsmath}
\usepackage{mathrsfs}
\usepackage{amsfonts}
\usepackage{mathpazo}
\usepackage{mathtools}
\usepackage{subcaption}
\usepackage{pdflscape}

\usepackage{harvard}
\usepackage{comment}
\usepackage{url}
\usepackage{physics}
 \usepackage{autobreak}
 \usepackage[draft]{todonotes}
 \usepackage{booktabs}
 \usepackage{float}   

 \usepackage{enumitem}         
\setlist[enumerate]{nosep, topsep=2pt, parsep=2pt, partopsep=2pt}
\setlist[itemize]{nosep, topsep=2pt, parsep=2pt, partopsep=2pt}

\usepackage[multiple]{footmisc}
\usepackage[all,2cell,ps]{xy}
\usepackage{titlesec}
\usepackage{tikz}
\usepackage{xr}
\usepackage{titlesec}

\usepackage{etoolbox}

\AtBeginEnvironment{proposition}{%
  \setlength{\abovedisplayskip}{6pt}%
  \setlength{\belowdisplayskip}{6pt}%
  \setlength{\abovedisplayshortskip}{6pt}%
  \setlength{\belowdisplayshortskip}{6pt}%
}

\titleformat*{\section}{\normalsize\bfseries}
\titleformat*{\subsection}{\normalsize\bfseries}
\titleformat*{\subsubsection}{\normalsize\bfseries}
\titleformat*{\subparagraph}{\normalsize\itshape}

\newcommand\blfootnote[1]{%
	\begingroup
	\renewcommand\thefootnote{}\footnote{#1}%
	\addtocounter{footnote}{-1}%
	\endgroup
}

\titleformat{\paragraph}
{\normalfont\normalsize\bfseries}{\theparagraph}{1em}{}
\titlespacing*{\paragraph}{0pt}{3ex plus 1ex minus .2ex}{1.5ex plus .2ex}
\titlespacing*{\section}{0pt}{1.7ex plus .2ex minus .2ex}{0.6ex plus .1ex}
\titlespacing*{\subsection}{0pt}{1.1ex plus .2ex minus .2ex}{0.6ex plus .1ex}

\newtheorem{defn}{Definition}
\newtheorem{thm}{Theorem}
\newtheorem{prop}{Proposition}

\newtheorem{cor}{Corollary}

\newtheorem{lem}{Lemma}
\newtheorem{rem}{Remark}
\newcommand{\edge}[1]{\ar@{-}[#1]}

\graphicspath{ {"Graphics/"} }

\includecomment{comment}

\def\l{\ensuremath\left}
\def\r{\ensuremath\right}

\allowdisplaybreaks

\begin{document}
\title{Tax Salience: How Requiring Transparency\\Affects the Price of Equality}

\author{Ashley Craig\\
\textit{Australian National University}\\
$ $\\
Itai Sher \\
\textit{University of Massachusetts Amherst}}
\maketitle
\begin{abstract}
\noindent Less-salient taxes can ease the classic trade-off between equality and efficiency by making people respond less to taxation. But deliberately obscuring taxes may be viewed as dishonest. This creates a three-way trade-off between equality, efficiency, and honesty. We analyze this trade-off in a simple setting with a linear income tax. We define and characterize the morally efficient frontier, trading off utilitarian welfare against honesty or transparency.  Complete honesty is Pareto inefficient but not morally inefficient.  More generally, increasing honesty reduces utilitarian welfare.   When utilitarian welfare is decomposed into equality and efficiency, the cost of honesty falls most robustly on equality: Higher salience always reduces equality, while the effect on efficiency is ambiguous. This asymmetry arises because salience increases the \textit{price of equality}, which is the efficiency cost of a marginal increase in equality.  Our approach could be applied to other settings in which utilitarian and procedural or deontological values conflict.

\end{abstract}

\singlespacing
\blfootnote{\scriptsize \hskip -20pt \textbf{Contact}: Ashley Craig (ashley.craig@anu.edu.au); Itai Sher (isher@umass.edu).}
\blfootnote{\scriptsize \hskip -20pt \textbf{Declarations of interest}: None.}
\blfootnote{\scriptsize \hskip -20pt \textbf{Thanks}: We thank Joshua Abel, Felix Bierbrauer, Maya Eden, Antoine Ferey, Jacob Goldin, Wojciech Kopczuk, Benjamin Lockwood, Timothy Moore, Paolo Piacquadio, Joel Slemrod, Richard Tol,  Damian Vergara, Matthew Weinzierl, and seminar participants at the Australian National University, University of Massachusetts Amherst, University of Rochester, University of Connecticut, University at Buffalo, Purdue University, Normative Aspects of
Economic Policy workshop at the University of Zurich, Frontiers of Economics and Philosophy at the Paris School of Economics, Australian Public Economics Exchange, National Tax Association, and Normative Economics and Economic Policy online seminar for valuable comments that have improved the paper. Artificial intelligence tools (ChatGPT and Claude) were used to assist with mathematical derivations, drafting and coding. The authors developed the model and the main arguments, constructed proofs, conceived and organized the analysis, and verified all results.}

\onehalfspacing
\clearpage

\section{Introduction}

The classic trade-off in taxation is between equality and efficiency. Taxes can be used for redistribution, but they also have incentive effects which make that redistribution costly. This trade-off can be eased if the government can rely on taxes that are less salient, meaning that taxpayers partly ignore or do not pay full attention to them \citeaffixed{LiebmanZeckhauser2004,ChettyLooneyKroft2009,ChettyFriedmanSaez2013,AbelerJager2015,morrison2023rules,gebbiaGrayHancuchOrganTaxRateMisperception}{see e.g.,}. People will respond less to non-salient taxes, which reduces the efficiency costs of taxation.

Based on this logic, numerous authors have observed that non-salient or non-transparent taxes can increase social welfare. \citeasnoun{Galle2009HiddenTaxes} highlights that hidden taxes can mitigate behavioral responses of high-income individuals, allowing for more redistribution with smaller  efficiency costs. \citeasnoun{gamage2011three} actively argue for ``non-salient'' taxes to reduce the distortionary costs of taxation. \citeasnoun{craig2024tax} show that obfuscating incentives created by the tax system can raise utilitarian welfare. \citeasnoun{Goldin2015} shows how welfare can be increased by combining salient subsidization with non-salient taxation.\footnote{In related work, \citeasnoun{agersnap2024} study the potential for tax complexity to be used by governments to screen more or less responsive taxpayers if more knowledgeable taxpayers are also more responsive.}

However, misleading taxpayers to pursue social goals may be viewed as objectionable.  And many may regard deliberate obfuscation of taxes as a form of deception.  Moreover, exploiting non-salience may be seen as violating norms that government policy ought to be transparent to citizens. This is reflected in the \citeasnoun{irsmission} mission statement, which requires that it help people understand their tax responsibilities and enforce the law with integrity. If one takes such moral considerations seriously, then taxation involves not just the standard two-way trade-off between efficiency and equality, but rather a three-way trade-off between equality, efficiency and honesty. In this paper, we explore that trade-off formally.\footnote{\citeasnoun{Goldin2015} is clearly cognizant of the deontological and political economy hazards of policies that harness non-salience, but these are not in the formal analysis. Similarly, \citeasnoun{Galle2009HiddenTaxes} discusses political risks of deliberately reducing tax salience. \citeasnoun{gamage2011three} and \citeasnoun{Goldin2012} also discuss such concerns conceptually.} We use simple economics to shed light on how an honesty or transparency requirement affects welfare and systematically shifts the balance between equality and efficiency.

We study a simple setting with a linear income tax and a constant labor supply elasticity. Salience is modeled as the degree to which agents under-perceive the tax rate.  One interpretation of the model is that there is a choice between a salient income tax and a less-salient consumption tax, with the overall level of salience determined by the mix of the two taxes. There are also other ways governments could choose to make taxes more or less salient. For example, phasing out transfers may be less salient than explicit taxation. Alternatively, simplifying the tax system may make it easier for households to optimally respond to it.

Our paper builds a non-welfarist component on top of an otherwise standard economic framework. We consider two ethical objectives: (1) utilitarian social welfare and (2) honesty or transparency, and assume that taxes that are \textit{deliberately} less salient are less honest or transparent.  Rather than specifying a trade-off between the two, we define notions of \textit{moral dominance} and \textit{moral efficiency}, analogous to Pareto dominance and efficiency, but applied to these two objectives rather than to individual preferences. 
An advantage of our approach in comparison to generalized social welfare weights \cite{SaezStantcheva2016} is that procedural values like honesty cannot easily be captured by reweighting individual benefits. Such weights are also generally inconsistent with maximization of a well-defined social objective \cite{Sher2024GSMWW}.

Our first step is to show that utilitarian social welfare is decreasing in salience at the optimal tax rate, which implies a trade-off between welfare and honesty (Theorem \ref{salience vs welfare proposition}).  Moreover, starting from a fully salient tax system and holding the tax rate fixed, reducing salience has no first-order effect on utility through behavior, but increases labor supply and thus the tax rebate received by all agents, resulting in a Pareto improvement (Proposition \ref{prop:pareto}).\footnote{Relatedly, in a model of optimal commodity taxes with a representative agent, Goldin \citeyear{Goldin2012,Goldin2015} uses envelope logic to show that full salience is never welfare maximizing.  Goldin also shows that zero salience never maximizes welfare (except in a knife-edge case). In contrast, in our model, zero salience would be welfare maximizing if it were achievable in practice.}

We next show that salience shifts the balance between equality and efficiency.  To study this, we use a standard decomposition, due to \citeasnoun{atkinson1970measurement} of the form:
\begin{align}\label{equality-efficiency}
\textup{Utilitarian Welfare} =   \textup{Equality} \times \textup{Efficiency}.
\end{align}
An advantage of this approach is that it appeals to the notion of equality that is inherent in the utilitarian social welfare function.

The fact that, at the optimal tax rate,  increasing tax salience lowers utilitarian welfare implies that doing so must have a cost to efficiency, equality, or both. One might expect higher salience to reduce efficiency, because more salient taxes effectively raise the elasticity of taxable income. However, the sign of the impact of salience on efficiency at the optimum turns out to be ambiguous. Instead, we show that increasing salience always increases inequality---a result we prove formally in Theorem \ref{equality theorem}. In this sense, the robust conflict is actually between honesty or transparency and the desire to reduce inequality.

A key concept that we introduce to obtain these stark results is the \textit{price of equality}.  Figure \ref{stylized}(a) presents the choice of equality ($E$) and efficiency ($\mu$). For any level of salience, there is a frontier of combinations of $E$ and $\mu$ that can be achieved by varying the marginal tax rate; this is illustrated by the concave red line.  The price of equality at any point on the frontier is the cost in terms of efficiency that must be paid to attain an additional unit of equality. It corresponds to the slope of the blue line.  We establish that \textit{the price of equality rises with salience}.  

With this result in hand, the effects of salience on equality and efficiency can be explained using the simple economics of income and substitution effects.  Figure \ref{stylized}(b) shows the impacts of lowering tax salience. The green curve is the equality-efficiency frontier at a high level of salience and the red curve is the frontier at a low level. The multiplicative decomposition of utilitarian welfare (equation~(\ref{equality-efficiency})) implies convex indifference curves, drawn in black. The highest feasible welfare is obtained at the point where the equality-efficiency frontier is tangent to the highest attainable indifference curve.  When salience falls, the equality-efficiency frontier tilts and shifts upward, leading to both substitution and income effects. The substitution effect is shown by the ``compensated" change from the original green frontier to the orange frontier and corresponds to the movement from optimum A to optimum B. The substitution effect leads to more equality and less efficiency when salience falls, due to the decline in the price of equality. The income effect  corresponds to a vertical upward shift in the equality-efficiency frontier, which pushes toward more of both equality and efficiency, as they are both ``normal" goods.

\singlespacing
\begin{figure}[h!]
\begin{center}
\vspace{-1em}
\caption{Effect of Reducing Transparency \label{stylized}} 
\vspace{-0.6em}
\subfloat[Price of Equality]
{\includegraphics[
  width=0.49\textwidth,
  trim=0 0 0 10pt
]{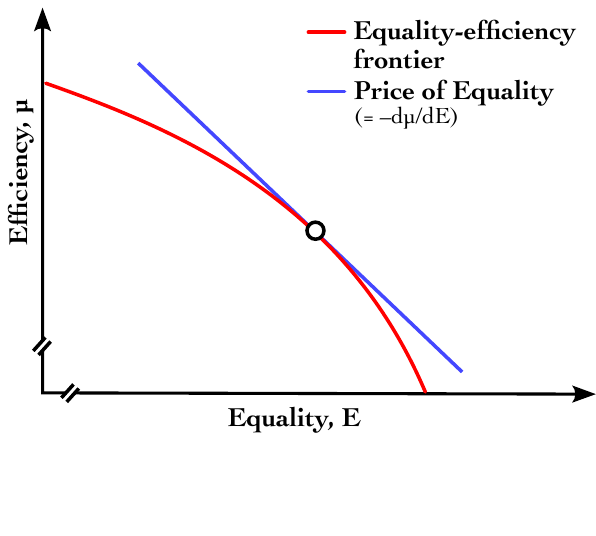}} \hspace{0.05em}
\subfloat[Income and Substitution Effects] 
{\includegraphics[
  width=0.49\textwidth,
  trim=0 0 0 10pt
]{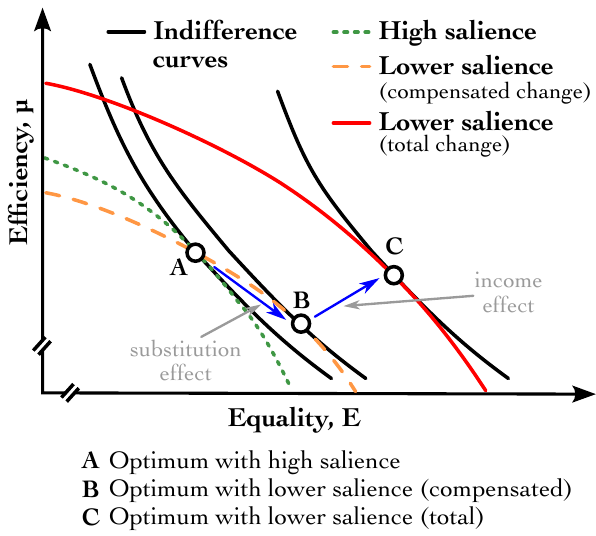}}
\end{center}
\vspace{-0.8em}
{\footnotesize\textit{Note.} This figure illustrates the utilitarian optimal tax problem formulated in terms of equality and efficiency at different levels of salience. Panel (a) displays the price of equality.  Panel (b) presents the effect of lowering salience on the optimum, with the overall impact divided into substitution and income effects.}
\end{figure}
\onehalfspacing

 The final optimum is at C. Equality is unambiguously higher with lower salience, while the impact on efficiency is ambiguous.  The reason for this asymmetry is that, with equality, the income and substitution effects move in the same direction and reinforce, whereas with efficiency they move in opposite directions and offset. This explains why the fundamental conflict is between transparency and equality, and not necessarily efficiency.

 The normative approach of this paper could be applied more broadly to modeling trade-offs between utilitarian and procedural values in economic settings. Typically, welfare economics is concerned with efficiency and aggregate welfare, with little attention to deontological criteria, such as rights and responsibilities. Honesty is one such criterion. There are many situations in which being less than fully honest or transparent might be useful from the standpoint of welfare, and yet people may still object to dishonesty or obfuscation. There are also other deontological objections that people may have to welfare-maximizing policies, attaching to paternalism, liberty, rights, or discrimination, where a similar approach may be appropriate.
 
\subsection{Related Literature} 
In addition to the literature on tax salience cited above, our paper contributes  to the broader literature on optimal policy in the presence of taxpayer biases \cite{Gerritsen2016b,TaubinskyReesJones2018,ReesJonesTaubinsky2020,FarhiGabaix2020,MooreSlemrod2021}, including how biases depend on the tax system \cite{TaubinskyReesJones2018,MooreSlemrod2021}. \citeasnoun{BoccanfusoFerey2024} study a setting in which misperceptions of the tax system self-correct.  \citeasnoun{craig2024tax} present a general model of optimal taxation when individuals have imperfect understanding of the tax system. Our analysis focuses on the trade-off between transparency and welfare, as well as on how salience affects efficiency and equality.

The contrast between utilitarianism and deontology is central to philosophical ethics; and normative criteria with a deontological component have drawn some attention from economists \cite{Sen1970,GaertnerPattanaikSuzumura1992,KanburPirttilaaTuomala2006,ZamirMedina2010,FleurbaeyManiquet2018,Sugden2018,Roemer2019}.  

There is also a literature on multi-criterion optimization at the boundary of computer science and economics, studying trade-offs between criteria such as welfare, simplicity, accuracy, and fairness in algorithmic settings \cite{KleinbergMullainathan2019,hu2020fair,wei2022fairness}. \citeasnoun{liang2025algorithm}
study the trade-off between accuracy and group-based fairness by considering the fairness-accuracy frontier.  They find that policies on the frontier may be Pareto inefficient, which is analogous to our Proposition \ref{prop:pareto} in their setting. In a different literature, \citeasnoun{AcemogluEtAl2021} study the frontier of economic losses and excess deaths in the context of COVID lockdowns, rather than invoking a specific value of life.\footnote{\citeasnoun{AshworthSpeaking} emphasizes this point in his discussion of \citeasnoun{AcemogluEtAl2021}.} \citeasnoun{AshworthSpeaking} develops an approach for economic advice in the context of multiple objectives that emphasizes dominance and efficient frontiers when considering conflicting goals. The literature cited above is similar in spirit to our approach of explicitly studying trade-offs between distinct social objectives without specifying how those objectives should be weighted.

\vspace{-0.1em}
\section{A model of tax misperception}\label{sec: model}

We assume a continuum of agents $i$ uniformly distributed on the interval $\l[0,1\r]$. Agents have quasilinear utility over consumption $c_i$ and pretax income $z_i \geq 0$:
\vspace{-0.4em}\begin{equation*}\vspace{-0.2em}
U_i\l(c_i,z_i\r) = u\l(c_i - v_i\l(z_i\r)\r).
\end{equation*}
The function $v_i\l(z_i\r)$ measures the cost of earning income $z_i$.  Where we do not impose specific functional forms, we assume that the relevant functions are continuous and sufficiently smooth so that all derivatives used in the analysis are well-defined and continuous. Utility is cardinal and interpersonally comparable.  The curvature of the outer utility function $u\l(\cdot\r)$ encodes the degree of inequality aversion. We assume that $u\l(\cdot\r)$ is strictly increasing and strictly concave.  The function $u\l(\cdot\r)$ does not affect the agent's decision problem but it does determine their \textbf{welfare weight}:
\vspace{-0.3em}\begin{equation*}\vspace{-0.6em}
g_i\l(c_i,z_i\r) = \pdv{c_i}U_i\l(c_i,z_i\r)=u'\l(c_i-v_i\l(z_i\r)\r).
\end{equation*}
It is also useful to define the average welfare weight $\bar{g}=\int_0^1 g_i \dd i$.

\vspace{-0.1em}
\subsection{Taxes and salience}\label{sec: taxes and salience}

The government sets a marginal income tax rate $\tau$. The per-capita revenue collected is $\tau \bar{z}$, where $\bar{z}=\int_0^1 z_i \dd i$, and is rebated  lump-sum to all agents.  Thus, agent $i$'s consumption is $c_i = \l(1-\tau\r)z_i + \tau \bar{z}$, and $i$'s utility is $u \l( z_i\l(1-\tau\r) -v_i\l(z_i\r) + \tau \bar{z}\r)$.

There is a salience parameter $s$ in $\l(0,1\r]$, which captures the degree to which agents under-perceive taxes.  Thus, agents perceive marginal tax rate $\tilde{\tau}=s \tau$, and they select income $z_i$ to maximize $z_i\l(1-\tilde{\tau}\r) - v_i\l(z_i\r)$.  Because utility is quasilinear, the anticipated rebate does not affect the choice of $z_i$.  Decisions are determined by the (mis-)\textit{perceived} tax rate $\tilde{\tau}$, and these decisions combine with the \textit{actual} tax rate $\tau$ to determine consumption and utility. 

If taxes were fully salient, then the first-order condition for $z_i$ would be $1-\tau - v'_i\l(z_i\r) =0$.  But with imperfect salience, it becomes: $1-\tilde{\tau} - v'_i\l(z_i\r) =0$, which is equivalent to:
\vspace{-0.1em}\begin{equation}\vspace{-0.1em}\label{first order condition internality}
1- \tau  - v_i'\l(z_i\r) =\tilde{\tau}-\tau.
\end{equation}
Note that $\tilde{\tau}-\tau< 0$.  This implies that agents choose $z_i$ past the peak of the objective they would maximize if they fully perceived taxes, working more than is privately optimal. 

We focus on $s$ in $\l(0,1\r]$ because that is the region in which there is a trade-off between honesty and utilitarian welfare; when $s> 1$, moving $s$ towards 1 improves accuracy of beliefs and increases welfare. We also exclude the degenerate (and unrealistic) case when $s=0$ because that simplifies some formulations below.

In what follows, we assume that the cost of earning income $v_i\l(z_i\r)$ takes the form
\vspace{-0.1em}\begin{equation}\vspace{-0.1em}\label{form cost of earning income}
v_i\l(z_i\r) = \frac{\l(\frac{z_i}{w_i}\r)^{1+\frac{1}{\varepsilon}}}{1+\frac{1}{\varepsilon}},
\end{equation}
 where $\varepsilon >0$ and where $w_i$ is $i$'s wage so that $\ell_i=\frac{z_i}{w_i}$ is $i$'s labor supply.  We assume that the distribution of wages is exogenously given and that the function $i \mapsto w_i$ is strictly increasing in $i$, with $w_i > 0$, for all $i$. Agents are heterogeneous only in their wages. 

It then follows from the agent's optimization problem that: 
\vspace{-0.1em}\begin{equation}\vspace{-0.1em}\label{zi expression}
z_i =w_i^{\varepsilon + 1}\l(1-\tilde{\tau}\r)^{\varepsilon}.
\end{equation}
So income is increasing in the wage, decreasing in the tax rate $\tau$, and decreasing in salience $s$.

 \subsection{Behavioral responses}\label{sec: behavioral responses}

With the above functional form, $\varepsilon$ is the elasticity of income with respect to the \textit{perceived} retention rate, $1-\tilde{\tau}$, for all agents and at all income levels. It is convenient to record the derivatives of $z_i$ and $\bar{z}$ with respect to $1-\tilde{\tau}$ as follows:
$\dv{z_i}{\l(1-\tilde{\tau}\r)} = \frac{\varepsilon z_i  }{1-\tilde{\tau}}$ and $\dv{\bar{z}}{\l(1-\tilde{\tau}\r)} = \frac{\varepsilon \bar{z}  }{1-\tilde{\tau}}$. We can also derive the responses of $z_i$ and $\bar{z}$ to the \textit{actual} retention rate. Using $1-\tilde{\tau}=1-s\tau$ and the chain rule, 
\vspace{-0.1em}\begin{align}\vspace{-0.1em}\label{tax derivatives}\pdv{z_i}{\l(1-\tau\r)} =& \frac{s\varepsilon z_i}{1-s\tau}, &\pdv{\bar{z}}{\l(1-\tau\r)} = \frac{s\varepsilon \bar{z}}{1-s\tau}.
\end{align}
This implies elasticities:
\vspace{-0.4em}\begin{equation*}\vspace{-0.6em}
\pdv{z_i}{\l(1-\tau\r)} \frac{1-\tau}{z_i} = \pdv{\bar{z}}{\l(1-\tau\r)}\frac{1-\tau}{\bar{z}} =  \underbrace{\frac{s\l(1-\tau\r)}{1-s\tau}}_{ \alpha\l(\tau,s\r)}\varepsilon. \end{equation*}
For any fixed $\tau$, $\alpha\l(\tau,s\r)$ increases from $0$ to $1$ as $s$ increases from $0$ to $1$.  It follows that, for any fixed $\tau$, \textit{reducing salience is behaviorally equivalent to reducing the elasticity of income with respect to the retention rate}.\footnote{In this sense, changing tax salience is similar to choosing the elasticity of taxable income, $\varepsilon$. \citeasnoun{SlemrodKopczuk2002} focus on the use of costly tools such as enforcement to manipulate $\varepsilon$, which---unlike in our paper---does not involve an internality. They study the choice of $\varepsilon$, depending on the degree of inequality aversion of the planner.  In our setting, the choice of effective elasticity, through the salience level and tax rate, would ultimately depend on how one trades off transparency and utilitarian welfare.} 
 
Finally, the responses of $z_i$ and $\bar{z}$ to salience are given by:
\begin{align}\label{dvzi}
\pdv{z_i}{s} = \dv{z_i}{\l(1-\tilde{\tau}\r)}\pdv{\l(1-\tilde{\tau}\r)}{s} &= - \frac{ \varepsilon \tau z_i}{1-s\tau}&\pdv{\bar{z}}{s} &= - \frac{ \varepsilon \tau \bar{z}}{1-s\tau}.
\end{align}
The rate at which income decreases in response to an increase in $s$ or $\tau$, holding the other fixed, is proportional to income $z_i$. 

\subsection{Interpretation and scope of the model}\label{alternative interpretation section}

We do not take a stance on the extent to which the government can alter tax salience, or how it does so. Our analysis is relevant if the government can marginally influence the salience level. One  interpretation is that our model represents a choice of how much to rely on a salient income tax, or a less-salient consumption tax.\footnote{In practice, there are differences between commodity and income taxes that are not captured here \cite{auerbach2019taxequivalences}. For example, they differ in their timing, which has implications for their incidence in a dynamic model. Reduced salience is one potential reason among others why a government might choose to lean on consumption taxes.} It is well-documented that individuals pay less-than-complete attention to commodity taxes when making their purchase decisions \cite{ChettyLooneyKroft2009,goldin2013smoke,feldman2018raising,TaubinskyReesJones2018,kroft2024salience}. It is less clear whether commodity taxes are also less salient than income taxes when people choose their labor supply, but experimental evidence suggests this is the case \cite{BlumkinRuffleGanun2012EER}.\footnote{It is hard to assess whether people correctly choose their labor supply in response to the overall level of commodity taxes. While not a labor supply response, \citeasnoun{rozema2017taxing} also show that households with smokers in them are sufficiently forward-looking to increase their take-up of food stamps when cigarette taxes increase.}

We show in Appendix \ref{app: alternative interpretation section} that this choice over a mix of two taxes is equivalent to a choice of $s$ and $\tau$. Specifically, take a pair $\l(\tau,s\r)$ in our model.  The government can achieve precisely the same incentives and allocations by choosing the right mix of a salient income tax $\tau^L$ and an exogenously less-salient consumption tax $\tau^C$. If we require both $\tau^L$ and $\tau^C$ to be nonnegative, this translation works provided the target salience is not too low relative to the salience $s^C$ of $\tau^C$.\footnote{The exact threshold depends on the target effective tax rate $\tau$ in addition to $s^C$.}  

Governments can manipulate salience in other ways as well. Phasing out a transfer received by households may be less salient than imposing a marginal tax rate, even if both produce equivalent budget sets. Thus, choosing different combinations of taxes and phase-out rates can alter salience while retaining the same true overall effective tax rate.\footnote{An example of misperception about transfers is documented by \citeasnoun{ChettyFriedmanSaez2013}, who show that many households in the United States appear to be unaware of incentives created by the Earned Income Tax Credit (EITC). \citeasnoun{craig2024tax} show that this creates an asymmetric incentive: in the phase-in region, where income is subsidized, informing households improves welfare; in the phase-out region, where households face a higher marginal tax rate as the EITC is phased out, there is an incentive to leave households uninformed. This would be an example of the welfare-honesty trade-off we study.} Tax complexity could also undermine understanding of incentives.  \citeasnoun{MCCAFFERY2003230} show that disaggregation---splitting a tax into separate components---affects perceptions of fairness and tax burden. If the same failure to aggregate across components affects tax salience, then choosing a more complex, disaggregated, tax system could be used to deliberately reduce tax salience.

One might expect that non-salience is temporary, and taxpayers would eventually learn to correctly perceive the tax schedule they face.\footnote{\citeasnoun{BoccanfusoFerey2024} study a setting in which, when misperceptions are exploited, they adjust toward the truth. A utilitarian government would want to commit to limiting the exploitation of misperceptions.  This is a complementary reason to avoid deception that is not present in our model.}  However, the evidence discussed above suggests that tax misperceptions are persistent and people do not fully adjust. Moreover, the tax system is complex and changing, which can interfere with adjustment. Our analysis can be interpreted as capturing either short-term misperceptions before learning has occurred, or residual equilibrium misperceptions that persist after learning.

In some cases, non-salience is intended to influence a political decision such as voting rather than an economic decision such as labor supply \cite{Schenk2011,gamage2011three}. For example, some countries deliberately choose not to index their tax brackets, because this allows tax rates to rise over time in real terms without any politically unpalatable active choice to increase them.  This type of deliberate political obfuscation has both similarities to and differences from non-salience used to increase efficiency. We focus on the latter.

Relatedly, when people speak of a ``non-salient" tax they may mean merely that agents are unaware of the tax. This lack of awareness need not affect economic decisions if agents nevertheless perceive the correct effective prices.\footnote{A commodity tax may be non-salient in the sense that consumers are unaware of the tax, but if the posted price is tax-inclusive, it will not affect their purchase decisions relative to the situation in which consumers were explicitly aware of the tax.} We are concerned instead with cases in which non-salience causes taxpayers to misperceive their after-tax wage and hence changes labor supply.

\section{Ethical criteria}\label{ethical criteria section}
We evaluate tax policy using two criteria:
\begin{enumerate}
\item utilitarian welfare; and 
\item honesty or transparency.
\end{enumerate}

First we consider utilitarian welfare.  A state is a pair $\l(\tau,s\r)$, where $\tau \in \l[0,1\r]$ is the marginal tax rate and $s \in \l(0,1\r]$ is the salience level.  We assume for simplicity that the salience level $s$ is the same for all agents. Utilitarian social welfare is given by:
\vspace{-0.3em}\begin{equation*}\vspace{-0.3em}
W\l(\tau,s\r) = \int_0^1 u \l( z_i\l(1-\tau\r) -v_i\l(z_i\r) + \tau \bar{z}\r) \dd i,
\end{equation*}
where income $z_i$ is chosen to respond optimally to the perceived tax rate $s\tau$.

The second value we consider is honesty or transparency, measured by $s$.  A lack of salience corresponds to dishonesty insofar as the government purposefully takes actions with the intent of reducing salience. If the government adopted a non-salient tax but had no reason to believe that it was non-salient, the lack of salience would simply be a behavioral bias, and would not imply any dishonesty on the part of the government. As we are interested in the policy decision to reduce salience, here we assume that changes in salience result from deliberate decisions.\footnote{We do not need to assume that every level of salience is achievable.  We are interested in the desirability of marginal changes.  We could, for example, assume that salience were only alterable within a small window $\l[s_0,s_1\r]$ and our main results would not change.}

Honesty and transparency are closely related but distinct values.  Dishonesty is defined as purposefully moving someone's beliefs away from the truth; non-transparency is defined as obscuring the facts and making them difficult to perceive.\footnote{One can distinguish transparency in the political process from transparency to taxpayers.  \citeasnoun{Schenk2011} advocates the former usage whereas our paper appeals to the latter. Even if, in principle, a citizen can look up the law, laws designed to obscure information from citizens can be problematic on grounds of transparency.}  Governments may be thought to have duties both to be honest and transparent.  Deliberate non-salience may be viewed as flouting either value.  For our purposes, it is not essential to finely specify the underlying value and our analysis is consistent with either interpretation.

\begin{defn}
Say that state $\l(\tau,s\r)$ \textbf{morally dominates} state $\l(\tau',s'\r)$ if either:
\begin{enumerate}
\item $s \geq s'$ and $W\l(\tau,s\r) > W\l(\tau',s'\r)$, or 
\item  $W\l(\tau,s\r) \geq W\l(\tau',s'\r)$ and $s > s'$.
\end{enumerate}
A state is \textbf{morally efficient} if there does not exist another state that morally dominates it.  
\end{defn} 

It is possible that one state
 $\l(\tau,s\r)$ Pareto dominates another state $\l(\tau',s'\r)$, 
 but $\l(\tau,s\r)$ does \textit{not} morally dominate $\l(\tau',s'\r)$ (since $s' > s$).  Indeed, we will present such a situation below (see Remark~\ref{policy interpretation remark}).  In such a situation, every individual is better off in $\l(\tau,s\r)$ but $\l(\tau,s\r)$ is less honest, so $\l(\tau,s\r)$ is not necessarily judged as better.  Whether it is judged as better depends on how one trades off utilitarian welfare and honesty/transparency.  

Many people regard transparency and honesty as important in tax policy and elsewhere.  There are two classes of potential reasons---\textit{deontological} and \textit{farsighted utilitarian}. Here we outline these reasons without endorsing them. 

\smallskip
 \noindent\textbf{Deontological reasons.} People should understand the laws and policies they face; it is bad to knowingly exploit people's misunderstandings, especially by the government.  These reasons appeal to ethical factors that are often absent in economic analysis: that intentions matter and that it is wrong to deceive; that transparency or honesty may be values in their own right; and that there may be special democratic duties of a government to be transparent and honest to its citizens. These are intricate values that we capture in a blunt way via the imperative to maximize $s$.    

\smallskip
 \noindent \textbf{Farsighted utilitarian reasons.} Acts such as exploiting non-salient taxes may erode trust in government and have farther reaching consequences.  Thus deception or policies perceived as manipulative may be bad for instrumental as well as intrinsic reasons.
 \smallskip
 
Under the farsighted interpretation, the conflict between honesty and welfare arises only because the standard optimal tax objective omits long-run welfare consequences; it is not a fundamental conflict.  Under the deontological interpretation, the conflict is genuine: a Pareto gain achieved through deceptive or nontransparent means may not be worth it.  The farsighted and deontological interpretations are not mutually exclusive; both types of reasons may be at play.  

One might be tempted to reduce deontological moral concerns to a matter of ethical preferences---that is, to incorporate a preference for honesty into agents' utility. However, violating a moral principle is conceptually distinct from not satisfying someone's preference that the principle be satisfied. Moreover, treating moral principles as preferences raises controversial and unresolved questions about whether ethical preferences should count towards individual welfare \cite{Sen1977RationalFools,Dworkin1990DoubleCounting,Milgrom1993Sympathy,KaplowShavell2002Fairness,Sher2020PerspectiveBased}.

People may differ in how they weigh welfare against honesty or transparency---from refusing to exploit salience regardless of welfare gains, to dismissing honesty concerns entirely, to trading the two off.  We do not take a stand.  Our framework accommodates any of these attitudes, allowing us to discuss the trade-offs without advocating for a specific point on the moral frontier.

\vspace{-0.1em}
\section{The trade-off between honesty and welfare}

We now characterize the morally efficient frontier.  

\vspace{-0.1em}
\begin{defn}\vspace{-0.1em}
For any $s \in \l(0,1\r]$, say that $\tau^*$ is $s$\textbf{-optimal} if $\tau^*$ maximizes $W\l(\tau,s\r)$.
\end{defn}
We denote the $s$-optimal marginal tax rate by $\tau\l(s\r)$; its uniqueness is shown in the Appendix (Proposition \ref{prop: uniqueness}). A formula characterizing the $s$-optimal tax rate is given in equation (\ref{useful optimal tax formula}) in Section \ref{sec: optimal tax rates} below, where we also discuss its properties further.

\vspace{-0.1em}\begin{prop}\vspace{-0.1em}\label{one way moral efficiency proposition}
If $\l(\tau,s\r)$ is morally efficient, then $\tau$ is $s$-optimal.
\end{prop}
\textit{Proof.}  If $\tau$ were not $s$-optimal, then it would be possible to increase utilitarian social welfare without lowering $s$ by modifying $\tau$. $\square$

We next ask whether the converse holds---whether $s$-optimality of the tax rate is sufficient for moral efficiency. This would be the case only if utilitarian welfare and salience are in strict opposition, so that any increase in honesty necessarily reduces welfare.  Otherwise, starting at salience level $s$ and an $s$-optimal tax rate, we might be able to increase salience and welfare simultaneously.  Theorem \ref{salience vs welfare proposition} below shows that this is impossible and that salience and welfare are indeed diametrically opposed.

Define $i$'s \textbf{consumption equivalent} to be  
\vspace{-0.3em}\begin{equation}\vspace{-0.2em}\label{money-metric utility}
\hat{c}_i = z_i\l(1-\tau\r) - v_i\l(z_i\r) + \tau \bar{z}. 
\end{equation}
Equivalently, $\hat{c}_i$ is the level of consumption that, combined with zero labor, would give $i$ the same utility as the bundle  $\l(c_i,z_i\r)$, where $c_i =\l(1-\tau\r)z_i + \tau \bar{z}$. Since $u\l(\cdot\r)$ is strictly increasing, $\hat{c}_i$ and $u\l(\hat{c}_i\r)$ provide ordinally equivalent representations of $i$'s utility.

Using (\ref{first order condition internality}) and (\ref{dvzi}), we have:

\vspace{-2em}
\begin{align}\label{ci derivative s}
\begin{split}
\pdv{\hat{c}_i}{s}  
= \frac{\varepsilon \tau^2}{1-s \tau}\l[z_i  \l(1-s\r)- \bar{z} \r].
\end{split}\vspace{-0.1em}
\end{align} We start with the case of perfect salience. When $s =1$,  $\pdv{\hat{c}_i}{s} =- \frac{\varepsilon \tau^2}{1-s \tau}\bar{z} < 0$, for all  $i$.  At full salience, agents are optimizing subject to the true tax rate, so reducing salience has no first-order effect on the utility of agents through the induced change in their own behavior.  However, when salience is reduced, everyone works a little more, leading to  a larger rebate, which benefits everyone.  As the effect on the rebate is the same for everyone, at $s=1$, $\pdv{\hat{c}_i}{s}$ does not depend on $i$.  Compactness and smoothness properties now imply Proposition \ref{prop:pareto}.\footnote{The set $\l[0,1\r]$ of agents $i$ is compact, and, for each $\tau$, $\hat{c}_i(\tau,s)$ is continuously differentiable in $(i,s)$. Since $\l.\pdv{\hat{c}_i}{s}\r|_{s=1}<0$ for all $i$, it follows that, for all $\tau$, there exists $s^* \in \l(0,1\r)$ such that $\hat{c}_i$ is decreasing in $s$ for all $i$ and all $s \in \l(s^*,1\r]$. The smallest possible such $s^*$---which turns out to be independent of $\tau$---is identified below.}
\vspace{-0.1em}\begin{prop}\vspace{-0.1em}\label{prop:pareto}
 When $\tau >0$ and salience is close to $1$, reducing salience slightly while holding the tax rate fixed is a \textit{Pareto improvement.}
 \end{prop}
 \vspace{-0.9em}\begin{rem}\vspace{-0.1em}\label{policy interpretation remark} \textbf{\textup{(Policy interpretation)}}
Suppose the government has two tax instruments which would be equivalent under full salience, one of which is fully salient and the other less than fully salient (see Section \ref{alternative interpretation section} for a discussion of such cases). Then starting from a situation where only the fully salient tax is used, it is Pareto improving to reduce the salient tax a little and increase the non-salient tax so that the effective marginal tax rate remains the same.
 \end{rem}
As discussed in Section \ref{ethical criteria section}, the existence of a Pareto improvement does not necessarily imply that we should reduce salience, as doing so would conflict with the honesty criterion. 

Reducing salience (at a fixed tax rate) leads to a Pareto improvement, not only at $s=1$, but as long as $s > s^*:=\frac{z_1-\bar{z}}{z_1} = 1- (\l.\int_0^1 w_i^{1+\varepsilon} \dd i\r/w_1^{1+\varepsilon})$, a threshold that is independent of the tax rate.

When $s$ is below $s^*$, the effect of a decrease in salience on agent $i$'s welfare is negative if and only if $z_i > \frac{\bar{z}}{\l(1-s\r)}$, that is, if $i$'s income is sufficiently above the average. In this case, lowering salience creates a conflict between higher- and lower-income agents: for higher income agents, the welfare cost of misoptimization is larger and they are made worse off, while lower income agents are made better off on net from the increased rebate.  This raises the question of whether aggregate utilitarian welfare rises or falls as salience declines. We address this question next.

Let $\hat{c}_i\l(\tau,s\r)$ be $i$'s equivalent consumption when the tax rate is $\tau$ and the salience level is $s$.  Say that a marginal tax rate $\tau$ is \textbf{order-preserving} at $s$ if $\hat{c}_i\l(\tau,s\r)$ is strictly increasing in $i$. As $w_i$ is increasing in $i$, this means that higher wage agents are better off. Misperception of the tax rate has a larger effect on the utility of high wage workers, and could in principle cause utility to decline with the wage. The following lemma (proved in Appendix \ref{app: lemma 1}) shows that this cannot happen when the tax rate is chosen optimally.

\vspace{-0.1em}\begin{lem}\vspace{-0.1em}\label{increasing normal lemma}
The $s$-optimal tax rate must be order-preserving at $s$.\end{lem}  
Define $W\l(s\r)= W\l(\tau\l(s\r),s\r)$, where $\tau\l(s\r)$ is the $s$-optimal tax rate. Thus $W\l(s\r)$ is utilitarian social welfare when the marginal tax rate is chosen optimally for the level of salience.  

\vspace{-0.1em}\begin{thm}\vspace{-0.1em}\label{salience vs welfare proposition}
Utilitarian social welfare is decreasing in salience if the marginal tax rate is chosen optimally in response to $s$. Formally, $W'\l(s\r)<0\:$ for all $s$.   
\end{thm}
\textit{Proof.}  Observe that $W\l(\tau,s\r)= \int u\l(\hat{c}_i\l(\tau,s\r)\r) \dd i$.  Recall that $g_i = \pdv{U_i}{c_i}$ is $i$'s welfare weight and that $\bar{g}= \int_0^1 g_i \dd i$ is the average welfare weight. Let $\tau$ be any positive order-preserving tax rate at $s$. Then: 

\vspace{-1.7em}
\begin{align}\label{Welfare salience inequality}
\begin{split}
\pdv{W}{s} =& \int g_i \pdv{\hat{c}_i}{s} \dd i
=  \frac{\varepsilon \tau^2}{1-s \tau}\int g_i \l[z_i  \l(1-s\r)- \bar{z} \r] \dd i\\
= &\frac{\varepsilon \tau^2}{1-s \tau} \l[\l(1-s\r) \int g_i z_i \dd i -\bar{g}\bar{z}\r]\\
\leq & \frac{\varepsilon \tau^2}{1-s \tau} \l[\l(1-s\r) \bar{g}\bar{z}-\bar{g}\bar{z}\r]=-\frac{\varepsilon \tau^2}{1-s \tau}\bar{g} \bar{z} s < 0,
\end{split}
\end{align}
where the second equality follows from (\ref{ci derivative s}), the first inequality follows from the fact that $\tau$ is order-preserving at $s$ and the Chebyshev integral inequality since when $\hat{c}_i$ is increasing in $i$, $g_i$ is decreasing in $i$.  Noting that, by Proposition \ref{prop: uniqueness} (in the Appendix) and Lemma \ref{increasing normal lemma}, $\tau\l(s\r)$ is positive and order-preserving at $s$, it follows from the envelope theorem and (\ref{Welfare salience inequality}) that $W'\l(s\r) =\l.\pdv{s}W\l(\tau,s\r)\r|_{\tau=\tau\l(s\r)} < 0$. $\square$

The proof of the theorem also shows that welfare is decreasing in salience holding \textit{fixed} any positive and order-preserving tax rate.

Theorem \ref{salience vs welfare proposition} implies a strict opposition between salience and utilitarian social welfare: any increase in salience reduces social welfare both when the tax rate adjusts optimally and when it is held fixed. It follows that moral efficiency coincides with $s$-optimality, strengthening Proposition~\ref{one way moral efficiency proposition} from a one-way implication to a biconditional that characterizes moral efficiency.
 
\begin{cor}\label{moral efficiency corollary}
$\l(\tau,s\r)$ is morally efficient if and only if $\tau$ is $s$-optimal.
\end{cor}

\section{Optimal tax rates}\label{sec: optimal tax rates}

We now characterize the $s$-optimal marginal tax rates that determine the morally efficient frontier.   At a fixed level of salience $s$, the utilitarian optimization problem is to choose $\tau$ to maximize:
\begin{align}\label{utilitarian social welfare objective}
W=\int_i u\l(z_i\l(1-\tau\r)-v_i\l(z_i\r)+\bar{z}\tau\r) \dd i,
\end{align}
given that each agent $i$ chooses $z_i$ optimally in response to perceived tax rate $s \tau$. The first order condition is:
\vspace{-0.3em}\begin{equation}\label{social welfare first order condition}
\pdv{W}{\tau} = \underbrace{\int g_i \l(\bar{z}-z_i\r)\dd i}_{\textup{redistribution}}\underbrace{+ \l(1-s\r)
\frac{s \tau \varepsilon}{1-s\tau }\int g_i z_i \dd i}_{\textup{internality}} \underbrace{- \frac{s\tau   \varepsilon }{1-s\tau }  \bar{g}\bar{z}}_{\textup{fiscal externality}} = 0.
\end{equation}
The first-order condition decomposes the effect of raising taxes into three components:
\begin{enumerate}[label=(\roman*), leftmargin=*, nosep]
\item \emph{redistribution}: the welfare gain from shifting income from high to low-wage agents, holding behavior fixed;
\item \emph{internality}: the corrective effect on misoptimization due to under-perception of taxes; and 
\item \emph{fiscal externality}: the revenue loss due to behavioral responses.
\end{enumerate}
Using equation (\ref{social welfare first order condition}) and appealing to Corollary \ref{moral efficiency corollary}, we obtain the following proposition. 

\begin{prop}\label{optimal tax formula proposition}
If $\l(\tau,s\r)$ is morally efficient, then it satisfies: 
\begin{equation}\label{useful optimal tax formula}
\frac{s\tau \varepsilon}{1-s\tau} = \frac{1-\bar{h}}{1-\l(1-s\r)\bar{h}},
\end{equation}
where: $\bar{h}= \frac{\int g_i z_i \dd i}{\bar{g} \bar{z}}$; and $g_i, z_i, \bar{g}, \bar{z}$ are the values of income and welfare weights induced by $\l(\tau,s\r)$.  
\end{prop}
The standard optimal tax formula with perfect salience is $\frac{\tau}{1-\tau}=\frac{1-\bar{h}}{\varepsilon}$. This is the special case of equation (\ref{useful optimal tax formula}) when $s=1$, in which case the first-order condition equates the redistributive benefit from a marginal increase in the tax rate $\bar{g}\bar{z} \times \l(1-\bar{h}\r)$ to the fiscal externality generated by behavioral responses $\bar{g} \bar{z} \times \frac{\tau \varepsilon}{1- \tau}$. With imperfect salience, the redistributive benefit is instead equated to the fiscal externality net of the corrective benefit associated with the internality $\bar{g}\bar{z}\times\frac{s\tau\varepsilon}{1-s\tau}-\bar{g}\bar{z}\times(1-s)\bar{h}\frac{s\tau\varepsilon}{1-s\tau}$. The corrective benefit comes from bringing labor supply closer to what the agent would choose without misperception. Both the fiscal externality and the corrective benefit depend on the perceived tax rate $s\tau$, which is why the perceived tax rate occupies the role of the statutory tax rate in the standard formula. Because the corrective benefit offsets part of the fiscal externality, the fiscal externality exceeds the redistributive benefit at the optimum.

\section{Equality-efficiency decomposition}\label{Sec: Equality-efficiency decomposition}
\subsection{Derivation of the decomposition}\label{sec: decomposition derivation}
Theorem \ref{salience vs welfare proposition} showed that increasing salience reduces utilitarian social welfare.  As we now show, utilitarian social welfare can be decomposed into equality and efficiency.  This implies that salience always has a cost to either efficiency or equality (or both).  We shall explore which it is. 

Following a tradition associated with \citeasnoun{atkinson1970measurement}, we use the equality measure that is inherent in the social welfare function.\footnote{The notions of equally distributed equivalent and the corresponding equality measure of the form (\ref{equality definition}) were independently presented by \citeasnoun{Kolm1969}.}  Given any state $\l(\tau,s\r)$, we define the  \textbf{equally distributed equivalent} $\xi\l(\tau,s\r)$ by:
\begin{equation*}
\xi\l(\tau, s\r)= u^{-1}\l(W\l(\tau, s\r)\r).
\end{equation*}
$\xi\l(\tau,s\r)$ is the level of consumption $\hat{c}$ that, if all agents were to consume $\hat{c}$ with zero labor supply, would yield the same utilitarian welfare as setting marginal tax rate $\tau$ under salience level $s$.\footnote{Under the isoelastic specification (\ref{isoelastic form}) introduced below, scale invariance, a standard requirement for relative measures of equality, pins down zero as the reference level of labor: it is the unique nonnegative level for which our equality measure is scale invariant. Even if we are willing to sacrifice scale invariance by selecting a nonzero reference level, our main results remain robust, subject to a condition on the wage distribution when $\rho > 1$, namely, that the elasticity of the cumulative wage distribution $w f_w\l(w\r)/F_w\l(w\r)$ is nonincreasing in $w$. This is satisfied by lognormal and Pareto distributions, as well as the Pareto-lognormal specification used by \citeasnoun{heathcote2021optimal}.  See  Appendix \ref{reference labor appendix}.} This equally distributed level of consumption is always positive if the tax rate is order-preserving, because this ensures that $\hat{c}_i >0, \forall i$ (see Proposition \ref{prop: positivity} in the Appendix). And, as Lemma \ref{increasing normal lemma} shows, the optimal tax rate is always order-preserving. Finally, since the function $u$ is strictly increasing, $\xi\l(\tau, s\r)$ and $W\l(\tau, s\r)$ are ordinally equivalent social welfare functions. 

Next, define the \textbf{efficiency} of the pair $\l(\tau,s\r)$ by:
\vspace{-0.2em}\begin{equation*}\vspace{-0.2em}
\mu\l(\tau,s\r)= \int_0^1 \hat{c}_i\l(\tau,s\r) \dd i.
\end{equation*}
Efficiency is measured by the average equivalent consumption of all agents in $\l(\tau,s\r)$, without correcting for diminishing marginal utility via $u\l(\cdot\r)$.  

The concavity of $u$ implies that $\xi\l(\tau,s\r) \leq \mu\l(\tau,s\r)$: If equivalent consumption is equally distributed, less of it is necessary to attain the same value of utilitarian welfare. 
As in \citeasnoun{atkinson1970measurement}, 
define \textbf{equality} at $\l(\tau,s\r)$ as:
\vspace{-0em}\begin{equation}\vspace{-0em}\label{equality definition}
E\l(\tau,s\r)= \frac{\xi\l(\tau,s\r)}{\mu\l(\tau,s\r)}.
\end{equation} 
This equality measure gives the proportion of total equivalent consumption required to achieve the current level of social welfare with equal distribution.  The complementary inequality index, $I=1-E$, can be thought of as telling us what proportion of equivalent consumption we could discard and achieve the same welfare by distributing what remains equally. Under perfect equality $\xi=\mu$, so that $E=1$; at the other extreme, when almost all money-metric utility is concentrated near the top, a much smaller aggregate level of equivalent consumption, distributed equally, would suffice to generate the same social welfare, so that $E$ would be small.

It follows from the above definitions that utilitarian welfare can be decomposed into equality and efficiency:
\vspace{-0.3em}\begin{equation}\label{welfare decomposition}
\underbrace{\xi\l(\tau,s\r)}_{\textup{utilitarian welfare}}= \underbrace{E\l(\tau,s\r)}_{\textup{equality}} \times\underbrace{\mu\l(\tau,s\r)}_{\textup{efficiency}}.  
\end{equation} 
This is the Atkinson decomposition. 

Why should one use $E=\xi/\mu$ rather than some other equality measure? Efficiency $\mu$ can be thought of as the size of the pie.  For our purposes, we want an equality measure that acts as a multiplicative discount factor that captures how unequal distribution reduces the social value of that pie.  $E = \xi/\mu$ is the \textit{unique} equality measure such that (i) $E \times \mu$ ranks allocations of equivalent consumption $\l(\hat{c}_i\r)_{i \in \l[0,1\r]}$ in the same way as the utilitarian social welfare function $W$ and (ii) $E=1$ for all constant allocations, giving the same equivalent consumption to all agents.  A formal statement is given in  Appendix \ref{appendix: characterization atkinson}.

For a fixed salience level $s$, the utilitarian optimal marginal tax rate $\tau$ is chosen to trade off efficiency against equality.  A higher marginal tax rate reduces efficiency but, via the induced rebate, increases equality.  This is formally equivalent to the problem of choosing the marginal tax rate $\tau$ to maximize $W$ in (\ref{utilitarian social welfare objective}) above.    

To proceed with the analysis, we again follow \citeasnoun{atkinson1970measurement} by specializing the outer utility function to an isoelastic form:
 \begin{equation}\label{isoelastic form} u\l(\hat{c}\r)= \frac{\hat{c}^{1-\rho}}{1-\rho},\;\;\; \rho > 0,\end{equation}
with $u\l(\hat{c}\r)=\log\l(\hat{c}\r),$ when $\rho=1$. Given this specification, equality as defined by (\ref{equality definition}) reduces to the standard scale-invariant Atkinson formula: \begin{equation}\label{Atkinson inequality index}
 E = \frac{1}{\mu} \l(\int_0^1 \hat{c}^{1-\rho}_i \dd i\r)^{\frac{1}{1-\rho}}.
 \end{equation}
The parameter $\rho$ measures inequality aversion, with larger $\rho$ indicating more concern for equality.

Given the functional form (\ref{isoelastic form}), the decomposition yields intuitive comparative statics with respect to inequality aversion, $\rho$.  At any fixed level of salience, the optimal tax rate is increasing in inequality aversion. Efficiency at the optimum is decreasing in inequality aversion. The effect on equality requires care because the equality measure itself depends on $\rho$. Nonetheless, we obtain a clear result. As $\rho$ increases, let the utilitarian optimal allocation vary accordingly. Equality then rises with $\rho$ when it is evaluated according to any fixed Atkinson measure of the form (\ref{Atkinson inequality index}).  See Appendix \ref{appendix: comparative statics of rho} for a precise statement.

\subsection{Partial effects of salience and taxes}

Ultimately, we are interested in how salience affects efficiency and equality \textit{at the optimal marginal tax rate}.  As a preliminary to that analysis, we first consider the partial effects: how each of salience and the marginal tax rate affects equality and efficiency, holding the other fixed.  

\begin{table}[htbp!]
\centering
\caption{Partial derivatives of efficiency and equality with respect to salience and the tax rate}
\label{tab:partials}
\begin{tabular}{lcc}
\toprule
 & \textbf{Efficiency $\mu$} & \textbf{Equality $E$} \\
\midrule
\textbf{Salience $s$} & $\partial \mu / \partial s < 0$ & $\partial E / \partial s < 0$ \\
\addlinespace[0.4em]  
\textbf{Tax rate $\tau$} & $\partial \mu / \partial \tau < 0$ & $\partial E / \partial \tau > 0$ \\
\bottomrule
\end{tabular}
\end{table}

Holding the tax rate fixed, increasing salience reduces both efficiency and equality.  The negative effect on efficiency stems from a reduction in labor effort, which moves us away from the efficient level. The reduction in equality arises in part because more salience reduces government revenue, which shrinks the equalizing lump-sum transfer, and in part because the resulting reduction in the internality favors higher income agents.  Unlike salience, changing the tax rate has opposite effects on efficiency and equality.  The trade-off when selecting marginal tax rates is precisely between efficiency and equality: Raising the marginal tax rate increases equality at the expense of efficiency.

We now formally state the conditions under which the inequalities in the table hold.  Recall that a marginal tax rate $\tau$ is order-preserving at $s$ if it makes higher wage types better off.  Say that a marginal tax rate is \textbf{revenue efficient} at $s$ if it is less than or equal to the revenue maximizing tax rate at $s$. We call a tax rate \textbf{admissible} at $s$ if it is order-preserving and revenue efficient at $s$.  Any $s$-optimal tax rate must be admissible at $s$.\footnote{The result that the optimal tax rate must be revenue efficient is standard---see Proposition \ref{revenue efficiency Proposition} in the Appendix---and Lemma \ref{increasing normal lemma} shows that it must be order-preserving as well.}

\begin{prop}\label{prop: partials}Assume that $\tau$ is admissible at $s$ and that $\tau >0$. Then the signs of the partial derivatives reported in Table \ref{tab:partials} hold.  Specifically, efficiency decreases in both salience and the tax rate, while equality decreases in salience and increases in the tax rate.\footnote{Different aspects of admissibility are required for different results: the salience effects require order-preservation and the tax rate effects require revenue efficiency.}
\end{prop}
\textit{Proof Sketch.} For efficiency, equation \ref{ci derivative s} implies that $\pdv{s} \mu\l(\tau,s\r) = \int \pdv{\hat{c}_i}{s} \dd i = \frac{ \tau^2\varepsilon}{1-s \tau} \int \l[z_i  \l(1-s\r)- \bar{z} \r] \dd i = - \frac{s \tau^2 \varepsilon \bar{z}}{1-s\tau} < 0.$  A similar calculation shows that $\pdv{\tau} \mu\l(\tau,s\r) =- \frac{  s^2 \tau \varepsilon\bar{z}}{1-s \tau} < 0$. For equality, the results follow from the behavior of the growth rates $\pdv{\hat{c}_i}{s}/\hat{c}_i$ and $\pdv{\hat{c}_i}{\tau}/\hat{c}_i$  which are respectively increasing and decreasing in $i$ (i.e., as we move to higher types). See Appendix \ref{proof of prop 4} for the formal argument.

\section{The effect of salience on the equality-efficiency frontier and the price of equality}\label{Sec: The price of equality}

We now turn to our main question: How does a change in salience affect efficiency and equality when the marginal tax rate is chosen optimally in response to the salience level? 

Define $E\l(s\r) = E\l(\tau\l(s\r),s\r)$ and $\mu\l(s\r)=\mu\l(\tau\l(s\r),s\r)$.  That is, $E\l(s\r)$ and $\mu\l(s\r)$ are, respectively, the levels of equality and efficiency that result at salience level $s$ when the marginal tax rate is chosen to be $s$-optimal.

Our main result is as follows:

\begin{thm}\label{equality theorem}
Equality is decreasing in salience when the marginal tax rate is chosen to be $s$-optimal.  Formally, $E'\l(s\r) < 0\:$ for all $s$.
\end{thm}  

In contrast, the effect of salience on efficiency is ambiguous when the marginal tax rate is chosen to be $s$-optimal.  For some specifications of the parameters and values of $s$, $\mu'\l(s\r) < 0$ and, for others, $\mu'\l(s\r) > 0$. Section \ref{efficiency ambiguous} discusses this further, and Appendix \ref{app:simulations} illustrates the ambiguity of the sign of $\mu'\l(s\r)$ with numerical simulations.

To prove Theorem \ref{equality theorem}, we begin by recasting the welfare maximization problem so that the decision variable is equality rather than the marginal tax rate.  Under this formulation, salience determines an equality-efficiency frontier and welfare maximization corresponds to choosing the optimal point on that frontier.  Changes in salience shift and tilt the frontier, and the resulting changes in equality and efficiency can be analyzed using standard income-substitution logic.  For equality, both the substitution and income effects push in the same direction: higher salience raises the ``price of equality" and contracts the feasible set. For efficiency, by contrast, the effects work in opposite directions.  This explains why the effect of salience on equality can be signed, while the effect on efficiency is ambiguous.  The following sections develop this logic in detail.

\subsection{A change of variables}\label{sec: change of variables}
Recalling that $\xi\l(\tau,s\r)$ is ordinally equivalent to $W\l(\tau,s\r)$, the problem of finding the $s$-optimal marginal tax rate can be written as:

\vspace{0.6em}
\noindent \textbf{Welfare Maximization ($\tau$ version) { }} Choose $\tau$ to maximize $\xi\l(\tau,s\r)=E\l(\tau,s\r) \times \mu\l(\tau,s\r)$. 
\vspace{0.6em}

\noindent     
Because $\pdv{E}{\tau} > 0$ for admissible tax rates (see Table \ref{tab:partials}), for any \textbf{admissible equality level} $\check{E}$ at $s$---that is, one compatible with a tax rate that is admissible at $s$---there is a unique tax rate $\check{\tau}\l(\check{E},s\r)$, admissible at $s$, that brings about $\check{E}$. Formally, $\check{\tau}\l(\check{E},s\r)$ is defined by: 
\begin{align}\label{definition of tilde tau}
E\l(\check{\tau}\l(\check{E},s\r),s\r) = \check{E}.
\end{align}
Define $\check{\mu}\l(\check{E},s\r)= \mu\l(\check{\tau}\l(\check{E},s\r),s\r)$ to be the corresponding efficiency level.

Since, at any $s$, there is a one-to-one correspondence  between admissible tax rates and admissible equality levels, we can equivalently reformulate the welfare maximization problem so that equality rather than the marginal tax rate is the decision variable. 

\vspace{0.6em}
\noindent \textbf{Welfare Maximization ($E$ version) { }} Choose $E$ to maximize $\check{\xi}\l(E,s\r):=E \times \check{\mu}\l(E,s\r)$.  
\vspace{0.6em}

\noindent Fixing $s$, we can think of this as the problem of choosing the welfare-optimal point on the \textbf{equality-efficiency frontier} $\l\{\l(E,\check{\mu}\l(E,s\r)\r): E \textup{ is admissible at }s\r\}$, as illustrated in Figure \ref{fig:eefrontier}.\footnote{\citeasnoun{Kolm1969} conceptualized public policy problems in general in a similar way, as trading off equality (which he referred to as ``justice") against efficiency.  See his Figure 7 on p. 187, which resembles our Figure \ref{fig:eefrontier}.}

\vspace{-1.6em}
\singlespacing
\begin{figure}[H]
    \begin{center}
    \includegraphics[width=0.5\textwidth, trim=0 45 0 0, clip]{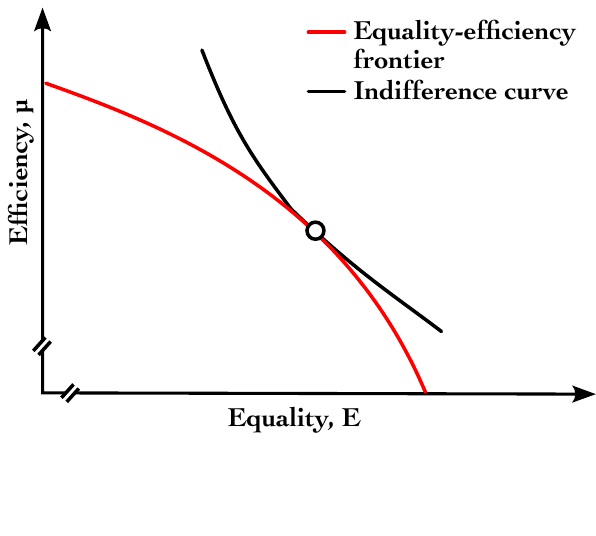}
    \vspace{-0.5em}
    \caption{Equality–efficiency frontier}
    \label{fig:eefrontier}
    \end{center}
    \vspace{-0.8em}
{\footnotesize\textit{Note.} This figure shows the welfare-maximizing choice of equality and efficiency. The optimum is the white dot, which corresponds to the highest-welfare combination of equality and efficiency that is feasible. This is where the equality-efficiency frontier is tangent to the utilitarian indifference curve.}
\end{figure}
\onehalfspacing
 Since $\tau =0$ cannot be an $s$-optimal tax rate (see Proposition \ref{prop: uniqueness} in the Appendix), the discussion below restricts attention to admissible equality levels for which $\check{\tau}\l(E,s\r)> 0$.

\subsection{Decomposing the effect of salience into substitution and income effects}\label{Sec: Decomposing}

When salience changes, the equality–efficiency frontier is altered in two ways. First, the level of the feasible set shifts, corresponding to a change in the level of efficiency that can be attained for a given level of equality. Second, the slope of the frontier changes, altering the marginal efficiency cost of additional equality. As in standard consumer theory, the change in optimal equality decomposes into a substitution effect (response to the change in the slope) and an income effect (response to the shift in the frontier). The two effects are illustrated in Figure \ref{fig:substitution_income}.

\vspace{-1.7em}
\singlespacing
\begin{figure}[H]
    \begin{center}
    \includegraphics[width=0.45\textwidth]{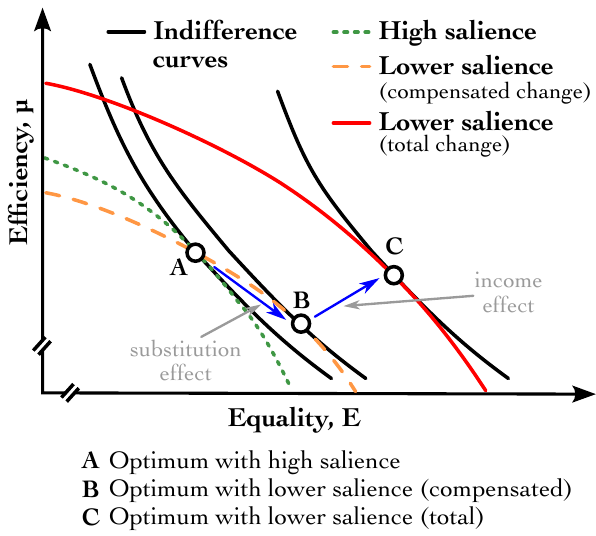}
    \vspace{-0.5em}
    \caption{Substitution and income effects}
    \label{fig:substitution_income}
    \end{center}
    \vspace{-0.8em}
    {\footnotesize\textit{Note.} This figure shows how the welfare-maximizing choice of efficiency and equality changes when salience is reduced. The overall impact is divided into the substitution effect (\textbf{A}$\rightarrow$\textbf{B}) and the income effect (\textbf{B}$\rightarrow$\textbf{C}).}
\end{figure}
\onehalfspacing

\vspace{-0.5em}
To formalize the decomposition, we introduce a relaxed problem with a slack variable $\delta$ that shifts the frontier as shown in Figure \ref{fig:shift_frontier}.

\vspace{0.5em}
\noindent \textbf{Welfare Maximization (relaxed version) { }} Choose $E$ and $\mu$ to maximize $\xi= E \times \mu$ subject to $\mu \leq \check{\mu}\l(E,s\r)+\delta$.
\vspace{0.5em}

\vspace{-1.8em}
\singlespacing
\begin{figure}[H]
    \begin{center}
    \includegraphics[width=0.45\textwidth]{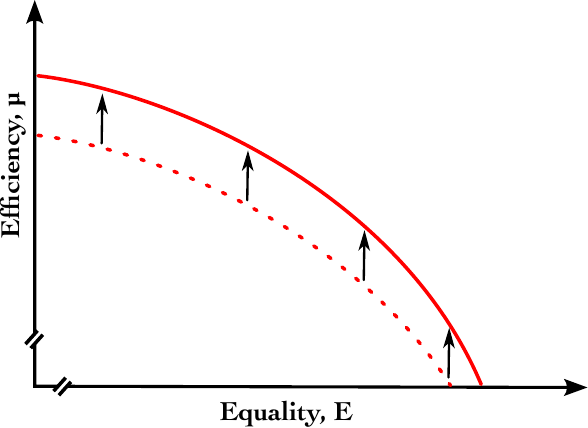}
    \vspace{-0.5em}
    \caption{Shifting the equality-efficiency frontier}
    \label{fig:shift_frontier}
    \end{center}
    \vspace{-0.8em}
    {\footnotesize\textit{Note.} This figure shows the ``relaxed'' problem which shifts the feasible set up vertically so that more efficiency is possible for any given level of equality.}
\end{figure}
\onehalfspacing

\noindent  Let $E^{\mathrm{rel}}\l(s,\delta\r)$ be the optimal level of equality for the relaxed problem.

 Now consider an initial salience level $s_0$ and a new salience level $s_1$. Define:
\begin{align}\label{exact compensation}
\delta\l(s_0,s_1\r):= \check{\mu}\l(E\l(s_0\r),s_0\r) - \check{\mu}\l(E\l(s_0\r),s_1\r)
\end{align}
to be the perturbation making the $s_0$-optimal equality-efficiency pair $\l(E\l(s_0\r),\mu\l(s_0\r)\r)$ just feasible in the relaxed problem at $s_1$.  Define $E^{\mathrm{c}}\l(s_0,s_1\r)$, which we refer to as the \textbf{compensated demand for equality}, to be the optimal level of equality at salience level $s_1$, when the welfare maximization problem is relaxed so as to make the $s_0$-optimum just feasible.  Formally:
\begin{align}\label{compensated definition}
E^{\mathrm{c}}\l(s_0,s_1\r) := E^{\mathrm{rel}}\l(s_1,\delta\l(s_0,s_1\r)\r).
\end{align} 
The following lemma follows easily from our definitions, and formally captures the substitution and income effects.
\begin{lem}\label{lem:decomposition}Let $E\l(s\r)$ be the optimal level of equality at $s$.  Then the change in optimal equality can be decomposed as:
$\displaystyle \l. \dv{E}{s}\r|_{s =s_0}=\underbrace{\l.\pdv{E^\mathrm{c}}{s_1}\r|_{s_1=s_0}}_{\text{substitution effect}} +\underbrace{\l.\pdv{E^{\mathrm{rel}}}{\delta}\r|_{\substack{s = s_0 \\ \delta = 0}}\times \l.\pdv{\check{\mu}}{s}\r|_{\substack{E= E\l(s_0\r) \\ s=s_0}}}_{\text{income effect}}.$
\end{lem}
To understand the income effect, note that, by (\ref{exact compensation}),  $\l.\pdv{\check{\mu}}{s}\r|_{\substack{E= E\l(s_0\r) \\ s=s_0}}=-\l.\pdv{\delta}{s_1}\r|_{s_1=s_0}$. 

Next, we sign these effects.

\subsection{Signing the income effect for equality}\label{sec:income effect} 

The income effect is $\displaystyle \pdv{E^{\mathrm{rel}}}{\delta} \times \pdv{\check{\mu}}{s}$.  Since the objective $E \times \mu$ is of the Cobb-Douglas form, equality acts as a normal good and hence expanding the frontier leads to more equality: $\displaystyle \l.\pdv{E^{\mathrm{rel}}}{\delta}\r|_{\delta =0} >0$ (see Appendix \ref{app: monotonicity appendix}).

 Next we sign $\pdv{\check{\mu}}{s}\l(E,s\r)$, which determines whether an increase in salience causes the frontier to expand or contract. Observe that if we increase salience $s$ and hold the marginal tax rate $\tau$ fixed, both efficiency and equality must decline (see Table \ref{tab:partials}).  To hold equality fixed, we must correspondingly increase $\tau$ as we increase $s$ (again Table \ref{tab:partials}); this would bring efficiency down further (Table \ref{tab:partials}).  So, to hold equality fixed when we increase $s$, we must reduce efficiency.  Hence $\pdv{\check{\mu}}{s} < 0$, which implies that the income effect for equality is negative, so the income effect moves equality in the opposite direction of salience.

\subsection{The price of equality}

The substitution effect captures how the optimal level of equality responds to a change in the slope of the equality–efficiency frontier.  Define the \textbf{price of equality} to be
\begin{align}\label{definition of pE}
p_E\l(E,s\r) := -\pdv{\check{\mu}}{E}\l(E,s\r).
\end{align}
The price of equality is the negative of the slope of the equality-efficiency frontier for salience level $s$ at the point $\l(E, \check{\mu}\l(E,s\r)\r)$ as shown in Figure \ref{fig:price_equality}.

The price of equality tells us how many units of efficiency, measured in terms of consumption equivalents, one would have to give up to attain an extra unit of equality.  The following lemma relates the substitution effect to the change in the price of equality. 

\vspace{-1.5em}
\singlespacing
\begin{figure}[H]
    \begin{center}
    \includegraphics[width=0.5\textwidth, trim=0 45 0 0, clip]{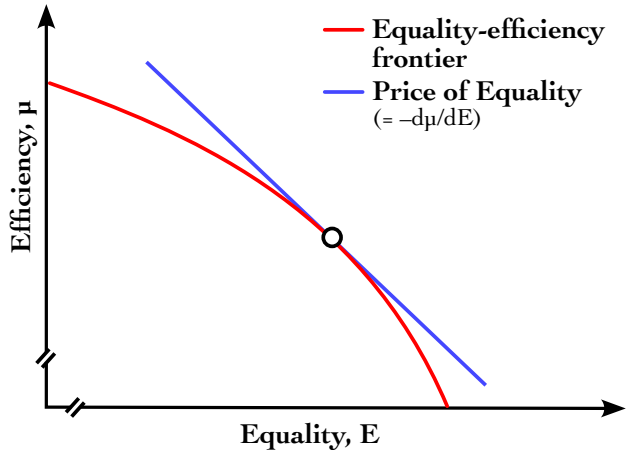}
    \vspace{-0.5em}
    \caption{The price of equality}
    \label{fig:price_equality}
    \end{center}
    \vspace{-0.75em}
    {\footnotesize\textit{Note.} This figure shows the price of equality. This corresponds to the slope of the frontier of the feasible set of efficiency and equality combinations that can be attained.}
\end{figure}
\onehalfspacing

\begin{lem}\label{lem:substitution}
The sign of the substitution effect on equality is equal to the sign of $-\displaystyle \pdv{p_E}{s}\l(E\l(s\r),s\r)$.
\end{lem}
That is to say, the substitution effect is negative or positive, depending on whether the price of equality rises or falls with salience.  The following lemma settles which it is.
\begin{lem}\label{price of equality lemma}
The price of equality is increasing in salience: $\displaystyle \pdv{p_E}{s} > 0$.  
\end{lem}
This implies that the substitution effect is negative when salience rises, leading to a reduction in equality.  We establish this lemma in the subsequent sections. 

\subsection{Iso-equality paths}

To establish the sign of $ \pdv{p_E}{s}$, it is useful to derive a more explicit expression for $p_E$.  Differentiating the relation $\check{\mu}\l(E,s\r) = \mu\l(\check{\tau}\l(E,s\r),s\r)$ with respect to $E$, appealing to the definition (\ref{definition of pE}) of $p_E$ and the inverse derivative rule, we have:
\begin{align}\label{expression for pE}
p_E\l(E,s\r) = \frac{-\mu_\tau\l(\check{\tau}\l(E,s\r),s\r)}{E_\tau\l(\check{\tau}\l(E,s\r),s\r)}.
\end{align}
Thus, the price of equality is equal to the ratio of the marginal efficiency cost of the tax rate to its marginal equality benefit. Hence, 
\begin{align*}
\pdv{p_E}{s}\l(E,s\r)=\dv{s} \frac{-\mu_\tau\l(\check{\tau}\l(E,s\r),s\r)}{E_\tau\l(\check{\tau}\l(E,s\r),s\r)}.
\end{align*} 
Since both the numerator $-\mu_\tau$ and denominator $E_\tau$ of the right-hand side are positive (Table \ref{tab:partials}), we have the following lemma.
\begin{lem}\label{paths lemma}
If $\dv{s} \l[-\mu_\tau\l(\check{\tau}\l(E,s\r),s\r)\r] > 0$ and $\dv{s} E_\tau\l(\check{\tau}\l(E,s\r),s\r)<0$, then $\pdv{p_E}{s}\l(E,s\r) > 0$.
\end{lem}
To interpret the two $s$-derivatives involving $\mu_\tau$ and $E_\tau$, we want to think about how the arguments $\l(\check{\tau}\l(E,s\r),s\r)
$ vary with $s$. Refer to a locus of tax rate-salience pairs $\l(\check{\tau}\l(E,s\r),s\r)$ that hold equality fixed at some level $E$ as salience varies as an \textbf{iso-equality path}:

\vspace{-1.5em}
\singlespacing
\begin{figure}[H]
    \begin{center}
    \includegraphics[width=0.5\textwidth]{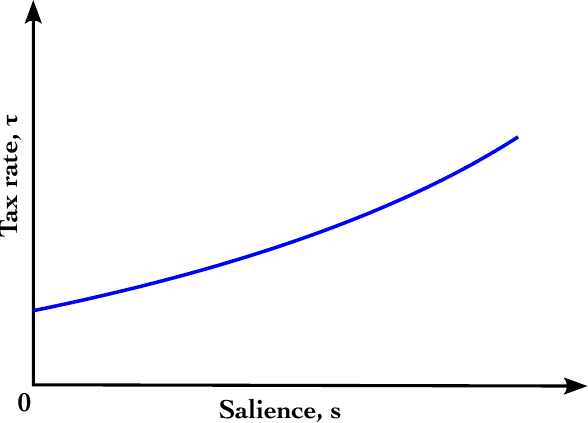}
    \vspace{-0.5em}
    \caption{An iso-equality path}
    \label{fig:iso-equality path}
    \end{center}
    \vspace{-0.75em}
    {\footnotesize\textit{Note.} This figure presents a stylized iso-equality path, which is the set of ($s,\tau$) combinations that achieve a given level of equality. This set is upward sloping, reflecting that higher salience reduces equality so that it must be accompanied by a higher tax rate if the level of equality is to be preserved.}
\end{figure}
\onehalfspacing

\vspace{-0.5em}

\noindent The slope is positive because an increase in salience, holding the tax rate fixed, reduces equality, and so must be countered by an increase in the marginal tax rate to keep equality constant.

We can now interpret Lemma \ref{paths lemma}: to establish that the price of equality rises, it is sufficient to show that, along an iso-equality path in the direction of greater salience, the marginal efficiency cost of the tax rate $-\mu_\tau$ rises while the marginal equality benefit of the tax rate $E_\tau$ falls.

\subsection{The marginal efficiency cost of the tax rate rises}\label{sec:marginal efficiency cost}
As we move along an iso-equality path in the direction of greater salience, there are two effects on the marginal efficiency cost of taxes $-\mu_\tau$: a direct effect through salience $-\mu_{s \tau}$ and an indirect effect $-\mu_{\tau \tau} \check{\tau}_s$ through the increase in tax rates required to maintain equality.  Formally,
\begin{equation*}
\dv{s}\l[-\mu_\tau\l(\check{\tau}\l(E,s\r),s\r)\r]=  - \mu_{\tau\tau}\l(\check{\tau}\l(E,s\r),s\r)\check{\tau}_s\l(E,s\r)-\mu_{s\tau}\l(\check{\tau}\l(E,s\r),s\r).  
\end{equation*}
When salience rises, the marginal efficiency cost of taxation increases, because (i) income falls, and hence, due to effort cost convexity, so does the marginal cost-savings of tax-induced income decline, and (ii) income becomes more responsive to taxes, increasing the magnitude of behavioral responses. So,    
\begin{equation}\label{mustau}
-\mu_{s \tau} >0.
\end{equation}
As explained above, $\check{\tau}_s > 0$, that is, the tax rate required to keep equality fixed increases with salience (see Figure \ref{fig:iso-equality path}); and it is also the case that increasing the marginal tax rate increases the marginal efficiency cost of further tax rate increases:
\begin{equation}\label{mutautau}
-\mu_{\tau \tau} >0.
\end{equation}
Inequalities (\ref{mustau}) and (\ref{mutautau}), which are formally established in the Appendix, imply that the marginal efficiency cost of the tax rate rises along an iso-equality path as salience increases.

\subsection{The marginal equality benefit of the tax rate falls}\label{sec:marginal equality benefit}

The next step is to argue that, along an iso-equality path, the marginal equality benefit of the tax rate falls as salience increases.  

We know from Section \ref{sec:income effect} that increasing salience causes the equality-efficiency frontier to contract, so that any level of equality is compatible with a lower level of efficiency, or equivalently of average equivalent consumption.  The key to the argument is to show more specifically that, along an iso-equality path, as salience rises, equality is held fixed by adjusting the tax rate in such a way that the equivalent consumption of all agents falls proportionately.  

The requirement that equivalent consumption falls proportionately implies the following ordinary differential equation (for any fixed $E$):
\begin{align*}
\check{\tau}_s\l(E,s\r) = \varepsilon \l[\check{\tau}\l(E,s\r)\r]^2. 
\end{align*}
This is established in the Appendix.  It follows that 
\begin{align*}
 \dv{s}\frac{1}{E_\tau\l(\check{\tau}\l(E,s\r),s\r)}=\pdv{s} \check{\tau}_E\l(E,s\r)=\pdv{E}\check{\tau}_{s}\l(E,s\r) = 2\varepsilon \check{\tau}\l(E,s\r) \check{\tau}_E\l(E,s\r) >0,
\end{align*}
where the inequality follows from $\pdv{E}{\tau}>0$ in Table \ref{tab:partials}. Hence, $\dv{s}E_\tau\l(\check{\tau}\l(E,s\r),s\r)< 0$, which is what we wanted to show. 

\subsection{Putting it all together: Salience reduces optimal equality}

In Sections \ref{sec:marginal efficiency cost} and \ref{sec:marginal equality benefit}, we established that the marginal efficiency cost of taxation rises with salience $\dv{s} \l[-\mu_\tau\l(\check{\tau}\l(E,s\r),s\r)\r] > 0$ and that the marginal equality benefit falls with salience $\dv{s} E_\tau\l(\check{\tau}\l(E,s\r),s\r)<0$.  By Lemma \ref{paths lemma}, this implies that $\pdv{p_E}{s}\l(E,s\r) > 0$ (establishing Lemma \ref{price of equality lemma}) and, by Lemma \ref{lem:substitution}, this means the substitution effect is negative. In Section \ref{sec:income effect} we established that the income effect is negative.  Thus both the substitution effect and the income effect move equality in the same direction. Optimal equality must therefore fall as salience rises (see Lemma \ref{lem:decomposition}). This completes the proof of Theorem \ref{equality theorem}.

\subsection{Ambiguity of the effect on efficiency} \label{efficiency ambiguous}
\noindent In contrast to equality, the effect of salience on efficiency at the optimal tax rate is ambiguous because the income and substitution effects point in opposite directions.  A rise in salience reduces efficiency through the income effect. However, it increases efficiency through the substitution effect because the relative price of efficiency falls. The decomposition for efficiency and signs of the income and substitution effects are established in Appendix \ref{app: substitution income effects efficiency}. We illustrate the ambiguity of the overall impact of salience using numerical exercises in Appendix \ref{app:simulations}.

The relationship between efficiency and the perceived tax rate provides a useful lens.  We can write $\mu = \int \l[z_i - v_i\l(z_i\r)\r] \dd i $, because redistributive transfers cancel out when considering efficiency.  Observe that $z_i - v_i\l(z_i\r)$ is strictly concave in $z_i$, and maximized when the perceived tax rate is zero. Moreover $z_i$ is decreasing in the perceived tax rate.  It follows that there is a one-to-one relationship between efficiency $\mu$ and the perceived tax rate $\tilde{\tau}$ and that efficiency falls as the perceived tax rate rises.\footnote{The perceived tax rate does not alone determine equality because, if we hold the perceived tax rate fixed and raise the actual tax rate, equality rises.} This also implies that \textit{iso-efficiency paths}, where efficiency is held fixed, take the simple form $\tau = C/s$ for some constant $C$.

 The correspondence between efficiency and $\tilde{\tau}$ implies that the ambiguity of the effect of salience on efficiency translates into an ambiguity of the effect of salience on the perceived tax rate.  Increasing salience reduces efficiency when $\dv{s} s \tau\l(s\r) = \tau+s\dv{\tau}{s} > 0$, or equivalently, when 
\begin{align}\label{elasticity of tau with respect to s}
\dv{\tau}{s}\frac{s}{\tau} > -1.
\end{align}
Thus, we can express the condition that salience reduces efficiency in terms of the elasticity of the optimal tax rate with respect to salience.  

A necessary, but not sufficient, condition for (\ref{elasticity of tau with respect to s}) to fail---that is, for salience to increase efficiency---is $\dv{\tau}{s} < 0$, or in other words, that the optimal tax rate falls when salience rises.  One might think that $\dv{\tau}{s} < 0$ must always hold. Proposition \ref{rho less than 2 proposition} of the Appendix shows that, when $s=1$ and $\rho \leq 2$ (moderate inequality aversion), $\dv{\tau}{s} < 0$.  Outside of $s=1$, there are simple counterexamples with $\dv{\tau}{s} > 0$.\footnote{The result would also cease to be true if we removed the assumption that $\rho \leq 2$: We have found counterexamples with $\rho > 2$ and $s=1$, but those we have found involve extreme income distributions.}  Note that $\dv{\tau}{s} > 0$ is consistent with our result that equality falls with salience because a rise in salience also has a direct effect, $\pdv{s}E\l(\tau,s\r) < 0$, that reduces equality.

\section{Conclusion}

Exploiting tax salience to improve welfare is something that many economists would instinctively regard as troubling, yet this unease has no formal expression in standard welfare economics. This paper provides an approach to taking it seriously. We model the conflict between welfare and honesty using a multi-criterion framework rather than a single social objective, because we do not have a basis for specifying the weight that should be placed on these competing considerations. Ultimately, the relative weight of welfare and honesty is a matter to be settled through ethical deliberation and the political process. However, insofar as such choices are made, it would be possible to infer the implicit weights from observed policy decisions. What we have been able to study formally is how different positions on this trade-off affect the balance between efficiency and equality for a given level of inequality aversion.

Exploiting salience is not the only instance in which economists and the public may feel moral unease about neglecting considerations that do not fit neatly into the standard welfare framework. A similar multi-criterion approach could be applied to balancing competing considerations in other such settings, such as paternalism, privacy, discrimination and desert.

Several extensions merit further study. We studied a simple setting with a linear income tax and uniform salience, and thought of changes in salience as stemming from deliberate government decisions. However, perfect salience is not always the default. The baseline may involve imperfect understanding, and there may be an important moral distinction between deliberate obfuscation and a failure to inform. Salience likely varies along the income distribution, and governments may deliberately target specific subpopulations, which may be especially attractive from a utilitarian standpoint if the tax system is nonlinear and contains complicated provisions. Higher-income taxpayers may hire advisors to help them navigate their taxes, though for our purposes what matters is how misperception affects labor effort, and such effects may survive the presence of tax advice at the end of the year.   In settings with heterogeneous salience, salience could no longer be summarized in a single dimension. Still, any given reduction in salience comes at a potential moral cost, and it would be interesting to explore how perceptions of that cost vary when a lack of salience is targeted at different parts of the income distribution for different purposes. Our analysis also relies on quasilinear preferences, which eliminate income effects. Introducing income effects would complicate the model in several ways: Misperceptions about total tax liability, not just marginal tax rates, would matter, and the measurement of equality and efficiency would itself become more involved.

A central analytical tool in this paper is the price of equality.  This requires a decomposition of  utilitarian welfare into equality and efficiency and a change of variables where equality, rather than the tax rate, becomes the choice variable.  This approach could also be valuable for the extensions mentioned above.  For example, one could ask about how a targeted attempt to reduce salience at one part of the income distribution would affect the price of equality.  Moreover, the price of equality may have purchase in other settings in which government policy shifts the equality-efficiency frontier, whether through salience, tax enforcement, or other channels.

\bigskip\bigskip
\appendix

\numberwithin{equation}{section}
\numberwithin{lem}{section}
\numberwithin{prop}{section}
\numberwithin{cor}{section}
\numberwithin{thm}{section}
\numberwithin{rem}{section}

\section*{Appendix}
\section{Mathematical proofs}

\subsection{Preliminaries}

\subsubsection{Domain of tax rates}

\begin{rem}\label{domain remark}  Throughout the appendix, at any salience level $s$, we restrict attention to tax rates in the interval $\l[0,\frac{1}{s}\r)$.  When $\tau \geq \frac{1}{s}$, the perceived tax rate is at least 1, leading to zero labor supply.  Such policies are Pareto dominated by setting the tax rate equal to zero.
\end{rem}

\subsubsection{Comparative statics of income}\label{sec: comparative statics of income}

This section collects some useful facts related to the comparative statics of income. These facts will be useful specifically for Lemma \ref{lem: effect on marginal efficiency cost of taxes} below. 

It follows from the agent's first order condition that $v_i'\l(z_i\r) =1- s \tau > 0$ (see Remark \ref{domain remark}). It follows that
\begin{align*}
v_i''\l(z_i\r) = \dv[2]{z_i}\frac{\l(\frac{z_i}{w_i}\r)^{1+\frac{1}{\varepsilon}}}{1+\frac{1}{\varepsilon}} 
= \dv{z_i} w_i^{-\l({1+\frac{1}{\varepsilon}}\r)}z_i^{\frac{1}{\varepsilon}}
= \frac{1}{\varepsilon} w_i^{-\l({1+\frac{1}{\varepsilon}}\r)}z_i^{\frac{1}{\varepsilon}-1}
= \frac{1}{\varepsilon}\frac{v_i'(z_i)}{z_i}
= \frac{1}{\varepsilon}\frac{1-s\tau}{z_i} >0.
\end{align*}

Recall from Section \ref{sec: behavioral responses} that $\pdv{z_i}{s}= -\frac{\varepsilon\tau}{1-s\tau}z_i$ and $ \pdv{z_i}{\tau} = -\frac{\varepsilon s}{1-s\tau}z_i$.  If $\tau > 0$, then $\pdv{z_i}{s}< 0$, and if $s > 0$, as we assume throughout (see Section \ref{sec: taxes and salience}), then $\pdv{z_i}{\tau}< 0$.  We have:
\begin{align}\label{cross partial parts}
\pdv{z_i}{s}{\tau} = \pdv{s} \pdv{z_i}{\tau}= \pdv{s}\l[-\frac{\varepsilon s}{1-s\tau}z_i\r] = -\varepsilon\l[\l( \pdv{s} \frac{s}{1-s\tau}\r) z_i +  \frac{s}{1-s\tau} \pdv{z_i}{s}\r]
 \end{align}
 Moreover,
 \begin{align*}
\pdv{s}\frac{s}{1-s\tau} = \frac{\l(1-s\tau\r)-s\l(-\tau\r)}{\l(1-s\tau\r)^2} = \frac{1}{\l(1-s\tau\r)^2}.
 \end{align*}
 Hence, appealing to (\ref{cross partial parts}), and the expression for $\pdv{z_i}{s}$ above,
 \begin{align}\label{second cross partial expression}
\pdv{z_i}{s}{\tau} = -\varepsilon \l[\frac{1}{\l(1-s\tau\r)^2}z_i +\frac{s}{1-s\tau}\l(-\frac{\tau \varepsilon}{1-s\tau}z_i\r) \r] = -\frac{\varepsilon \l(1-s\tau\varepsilon\r)}{\l(1-s\tau\r)^2}z_i. 
\end{align}
Revenue efficiency is equivalent to $\frac{s \tau \varepsilon}{1-s \tau} \leq 1 $ (see Proposition \ref{prop: revenue maximizing not optimal}), and hence implies $s\tau \varepsilon < 1$. This in turn implies, via (\ref{second cross partial expression}), that, if $\tau$ is revenue efficient, then $\pdv{z_i}{s}{\tau} < 0$.
We also have:
\begin{align*}
\pdv[2]{z_i}{\tau} = &\;\pdv{\tau}\pdv{z_i}{\tau} = \pdv{\tau}\l[-\frac{\varepsilon s}{1-s\tau}z_i\r]= -\varepsilon s\l[\l(\pdv{\tau}\frac{1}{1-s\tau}\r)z_i + \frac{1}{1-s\tau} \pdv{z_i}{\tau} \r]\\
=&\; -\varepsilon s\l[\frac{s}{\l(1-s\tau\r)^2}z_i + \frac{1}{1-s\tau}\l(-\frac{s \varepsilon}{1-s\tau}z_i\r) \r] = -\frac{\varepsilon \l(1-\varepsilon\r)s^2 }{\l(1-s\tau\r)^2}z_i,
\end{align*}
from which it follows that
\begin{align*}
    \pdv[2]{z_i}{\tau}= -\frac{\varepsilon(1-\varepsilon)s^2}{(1-s\tau)^2}\,z_i
\;\begin{cases}
<0 & \text{if } 0<\varepsilon<1,\\
=0 & \text{if } \varepsilon=1,\\
>0 & \text{if } \varepsilon>1
\end{cases}.
\end{align*}

\subsubsection{Relation between $\hat{c}_i$ and $z_i$}\label{relation c z subsection}
It is useful for some of the arguments below to develop a relation between $\hat{c}_i$ and $z_i$. Recall from (\ref{zi expression}) that optimal income for agent $i$ is given by: 
\begin{align}\label{optimal income s}
z_i = w_i^{\varepsilon+1}\l(1-s\tau\r)^\varepsilon.
\end{align}
This follows immediately from the first order condition with our specification (\ref{form cost of earning income}) for the cost function: $1-s\tau = \l.z_i^{\frac{1}{\varepsilon}}\r/w_i^{1+\frac{1}{\varepsilon}}$.  It follows from (\ref{optimal income s}) that 
\begin{align*}
\frac{\l(\frac{z_i}{w_i}\r)^{1+\frac{1}{\varepsilon}}}{1+\frac{1}{\varepsilon}}= \frac{\varepsilon}{1+\varepsilon} \l[w_i\l(1-s\tau\r)\r]^{1+\varepsilon}= \frac{\varepsilon}{1+\varepsilon} z_i\l(1-s\tau\r).
\end{align*}
Hence, using (\ref{money-metric utility}) and (\ref{form cost of earning income}), 
\begin{align}\label{simplification of an expression s}
\begin{split}
\hat{c}_i=&\; z_i\l(1-\tau\r) - \frac{\l(\frac{z_i}{w_i}\r)^{1+\frac{1}{\varepsilon}}}{1+\frac{1}{\varepsilon}}+ \tau \bar{z}\;\\
=&\;  \l[\l(1-\tau\r)-\frac{\varepsilon}{1+\varepsilon}\l(1-s\tau\r)\r] z_i +\tau \bar{z}\\
=&\; \frac{1-\tau\l(1+\varepsilon\l(1-s\r)\r)}{1+\varepsilon} z_i +\tau \bar{z}.
\end{split}
\end{align}
Noting that $z_i$ is strictly increasing in $i$, it follows that: 
\begin{prop}\label{prop: order preserving condition}
A tax rate $\tau$ is order-preserving at $s$ if and only if $\displaystyle \tau < \frac{1}{1+\varepsilon\l(1-s\r)}$.
\end{prop}    Moreover, note that we can write $\hat{c}_i= a\l(\tau,s\r) z_i + b\l(\tau,s\r)$ for positive numbers $a\l(\tau,s\r), b\l(\tau,s\r)$ that depend on $s$ and $\tau$ but not on $i$.

\begin{rem}
\label{domain restriction remark 2}Because $\frac{1}{1+\varepsilon\l(1-s\r)} \leq \frac{1}{s}$, every order-preserving tax rate satisfies the domain restriction $s\tau <1$.  Likewise, every revenue efficient tax rate satisfies the domain restriction.
\end{rem}

\subsubsection{Expressing $\hat{c}_i$ in terms of $w_i$}
It is also useful to express $\hat{c}_i$ in terms of $w_i$ rather than $z_i$.  $w_i$, as opposed to $z_i$, is exogenous.

Using (\ref{optimal income s}) and (\ref{simplification of an expression s}), we have:
\begin{align}\label{hat c i in terms of w i}
\hat{c}_i= \l(1-s\tau\r)^\varepsilon \l[ \frac{1-\tau\l(1+\varepsilon\l(1-s\r)\r)}{1+\varepsilon} w^{1+\varepsilon}_i +\tau \int w_j^{1+\varepsilon}\dd j\r]. 
\end{align}

\subsubsection{Positivity of equivalent consumption}
\begin{prop}\label{prop: positivity}
 If $\tau$ is order-preserving at $s$, then $\hat{c}_i\l(\tau,s\r) > 0$, for all $i$.  
 \end{prop}
 \textit{Proof.}  Assume that $\tau$ is order-preserving.   Then, by Proposition \ref{prop: order preserving condition}, the coefficient on $z_i$ in (\ref{simplification of an expression s}), $\frac{1-\tau\l(1+\varepsilon\l(1-s\r)\r)}{1+\varepsilon}$, is positive.  Because $s\tau < 1$ (see Remarks \ref{domain remark} and \ref{domain restriction remark 2}), it follows from (\ref{hat c i in terms of w i}) and the assumption that $w_i > 0, \forall i$, that $\hat{c}_i > 0, \forall i$. $\square$

 \subsubsection{The case $\rho=1$}

When $\rho = 1$, the formula $u\l(\hat{c}\r)=\frac{\hat{c}^{1-\rho}}{1-\rho}$ is undefined.  However, the equivalent representation $u\l(\hat{c}\r) = \l(\hat{c}^{\,1-\rho} - 1\r)/\l(1-\rho\r)$, which differs from $u$ only by an additive constant, converges pointwise to $\log \hat{c}$ as $\rho \rightarrow 1$.  We therefore use $u\l(\hat{c}\r) = \log \hat{c}$ at $\rho = 1$, as is standard.  The derivatives
\begin{align}\label{derivative functional form}
u'\l(\hat{c}\r) = \hat{c}^{-\rho}, \qquad u''\l(\hat{c}\r) = -\rho\, \hat{c}^{-\rho-1}
\end{align}
are well-defined for all $\rho > 0$, including $\rho = 1$.  Applying the general definition of the Atkinson index (\ref{equality definition}) to $u\l(\hat{c}\r) = \log \hat{c}$ yields
\begin{align}\label{Atkinson log}
E = \frac{1}{\mu} \exp \l(\int_0^1 \log \hat{c}_i \dd i\r),
\end{align}
which retains the scale invariance we appeal to elsewhere.

All of our results extend to $\rho = 1$.  Many arguments do not depend on the specific functional form of $u$ beyond monotonicity and concavity.  Where the form does enter, it enters either through the derivatives $u'$ and $u''$, given by (\ref{derivative functional form}) for all $\rho > 0$, or through the Atkinson index, whose $\rho = 1$ form is given by (\ref{Atkinson log}).  In neither case does $\rho = 1$ pose a difficulty.

\subsection{Optimal tax characterization}

\subsubsection{Revenue efficiency}

\begin{prop}\label{revenue efficiency prop}
The revenue function $R\l(\tau\r) = \tau \bar{z}(\tau)$ is single-peaked on $[0, \frac{1}{s})$, with unique maximum at $\tau^{\mathrm{max}} = \frac{1}{s\l(1+\varepsilon\r)}$.
\end{prop}
\textit{Proof.}  Differentiating revenue:
\begin{align*}
R'\l(\tau\r) = \bar{z} + \tau \pdv{\bar{z}}{\tau} = \bar{z} - \frac{\varepsilon s \tau}{1-s\tau}\bar{z} = \frac{\bar{z}}{1-s\tau}\l[1 - s\tau\l(1+\varepsilon\r)\r].
\end{align*}
Since $\bar{z} > 0$ and $1-s\tau > 0$ on the domain $\l[0, \frac{1}{s}\r)$, the sign of $R'\l(\tau\r)$ is determined by
\begin{align*}
\phi\l(\tau\r) := 1 - s\tau\l(1+\varepsilon\r).
\end{align*}
We have $\phi'\l(\tau\r) = -s\l(1+\varepsilon\r) < 0$, so $\phi$ is strictly decreasing. Moreover, $\phi(0) = 1 > 0$ and $\phi\l(\frac{1}{s}\r) = -\varepsilon < 0$. By continuity and the fact that $\phi$ is strictly decreasing, there exists a unique $\tau^{\mathrm{max}} \in \l(0, \frac{1}{s}\r)$ such that $\phi\l(\tau^{\mathrm{max}}\r) = 0$. Solving yields $\tau^{\mathrm{max}} = \frac{1}{s\l(1+\varepsilon\r)}$.

Again because $\phi$ is strictly decreasing, $\phi\l(\tau\r) > 0$ for $\tau < \tau^{\mathrm{max}}$ and $\phi\l(\tau\r) < 0$ for $\tau > \tau^{\mathrm{max}}$. Hence $R'\l(\tau\r) > 0$ for $\tau < \tau^{\mathrm{max}}$ and $R'\l(\tau\r) < 0$ for $\tau > \tau^{\mathrm{max}}$, establishing single-peakedness. $\square$

\begin{prop}\label{revenue efficiency Proposition} Any $s$-optimal tax rate must be revenue efficient. 
\end{prop}
\textit{Proof.}  Utility can be split into non-lump-sum and lump-sum parts, $\hat{c}_i^{\mathrm{NL}}$ and $\hat{c}^L$, respectively:
\begin{align}\label{lumpsum and nonlumpsum parts}
\hat{c}_i\l(\tau,s\r) =  \underbrace{z_i \l(1-\tau\r) - v_i\l(z_i\l(\tau,s\r)\r)}_{\hat{c}_i^{\mathrm{NL}}} + \underbrace{\tau \cdot \bar{z}\l(\tau,s\r)}_{\hat{c}^L}.
\end{align}
Using a derivation similar to that of (\ref{hat c i in terms of w i}), we have:
\begin{align}\label{cNL expression}
\begin{split}
\hat{c}^{\mathrm{NL}}_i\l(\tau,s\r)=&\; \l(1-s\tau\r)^\varepsilon  \frac{1-\tau\l(1+\varepsilon\l(1-s\r)\r)}{1+\varepsilon} w^{1+\varepsilon}_i\\
=&\; \l(1-s\tau\r)^\varepsilon \tau \underbrace{\frac{1-\tau\l(1+\varepsilon\l(1-s\r)\r)}{\tau \l(1+\varepsilon\r)}}_{\psi\l(\tau\r)} w^{1+\varepsilon}_i.
\end{split}
 \end{align}
 Now consider any $\tau_1$ that is greater than the Laffer rate $\tau^{\mathrm{max}}$.  It follows from the fact that a zero tax rate raises zero revenue, and the continuity and single-peakedness of the revenue function $R\l(\tau\r)$ (see Proposition \ref{revenue efficiency prop} for single-peakedness) that there exists a unique tax policy $\tau_0 < \tau^{\mathrm{max}}$ such that $R\l(\tau_0\r)=R\l(\tau_1\r)$.  It is immediate that $\hat{c}^{\mathrm{L}}\l(\tau_0,s\r) = \hat{c}^{\mathrm{L}}\l(\tau_1,s\r)$. 
 
Next observe that, by (\ref{optimal income s}),   
\begin{align*}
R\l(\tau\r) = \tau \bar{z} = \l(1-s\tau\r)^\varepsilon \tau \int w_i^{1+\varepsilon} \dd i.
\end{align*}
Because $R\l(\tau_0\r) = R\l(\tau_1\r)$, it follows that $\l(1-s\tau_0\r)^\varepsilon \tau_0= \l(1-s\tau_1\r)^\varepsilon \tau_1$. Given that $\psi(\tau) = \frac{1}{\tau(1+\varepsilon)} - \frac{1+\varepsilon(1-s)}{1+\varepsilon}$ is strictly decreasing in $\tau$, it follows from (\ref{cNL expression}) that $\hat{c}^{\mathrm{NL}}_i\l(\tau_0,s\r) > \hat{c}^{\mathrm{NL}}_i\l(\tau_1,s\r)$.  Since it is also the case that $\hat{c}^{\mathrm{L}}\l(\tau_0,s\r) = \hat{c}^{\mathrm{L}}\l(\tau_1,s\r)$, it follows that $\hat{c}_i\l(\tau_0,s\r) > \hat{c}_i\l(\tau_1,s\r)$, for all $i$.  So $\tau_0$ Pareto dominates $\tau_1$ at $s$, hence $\tau_1$ cannot be $s$-optimal. $\square$

\begin{prop}\label{prop: revenue maximizing not optimal} The revenue maximizing tax rate $\tau^{\mathrm{max}} = \frac{1}{s\l(1+\varepsilon\r)}$ cannot be $s$-optimal.
\end{prop}
\textit{Proof.}  Define $\hat{c}^{\mathrm{NL}}_i$ and $\hat{c}^{\mathrm{L}}$ as in (\ref{lumpsum and nonlumpsum parts}). At $\tau^{\mathrm{max}}$, $\pdv{\hat{c}^{\mathrm{L}}}{\tau}=0$.  Note also that, at $\tau=\tau^{\mathrm{max}}$, $\frac{s \tau \varepsilon}{1-s\tau}=1$.  Hence,\footnote{We assume that $w_i > 0, \forall i$ (Section \ref{sec: taxes and salience}), and it follows from (\ref{optimal income s}) that $z_i > 0, \forall i$.} 
\begin{align*}
\pdv{\hat{c}_i^{\mathrm{NL}}}{\tau}
=
z_i\l[-1 + \l(1-s\r)\tau\frac{s\varepsilon}{1-s\tau}\r] = -sz_i  < 0.
\end{align*}
Hence, at $\tau^{\mathrm{max}}$, $\pdv{\hat{c}_i}{\tau}<0, \forall i$, and reducing the tax rate slightly leads to a Pareto improvement, so that $\tau^{\mathrm{max}}$ cannot be optimal. $\square$

\subsubsection{Proof of Lemma \ref{increasing normal lemma}: $s$-optimal tax rates are order-preserving}\label{app: lemma 1}
Here we prove that any $s$-optimal tax rate is order-preserving; uniqueness of the $s$-optimal tax rate is proven in Section \ref{uniqueness appendix}. 

We have:
 \begin{align}\label{expression for derivative}
 \begin{split}
\pdv{W}{\tau}=&\int g_i\l(\bar{z}-z_i\r) \dd i  +\tau\l(s-1\r)  \int g_i \pdv{z_i}{\tau} \dd i + \bar{g}\tau\pdv{\bar{z}}{\tau}\\=&\int g_i\l(\bar{z}-z_i\r) \dd i  +\tau\l(1-s\r)  \frac{s\varepsilon}{1-s\tau} \int g_i z_i \dd i - \bar{g}\tau\frac{s\varepsilon\bar{z}}{1-s\tau}\\
 =& \underbrace{\l[1-\frac{ s \tau\varepsilon}{1-s\tau }\r]}_A \int g_i\l(\bar{z}-z_i\r) \dd i \underbrace{- \frac{ s^2 \tau\varepsilon}{1-s\tau }\int g_i z_i \dd i}_B,
 \end{split}
 \end{align}
 where the first equality appeals to (\ref{first order condition internality}). Because $w_i$ is increasing in $i$ and $\hat{c}_i$ is an affine function of $w^{1+\varepsilon}_i$ for fixed $s$ and $\tau$ (see eq. \ref{hat c i in terms of w i}), it follows that, at any salience level and marginal tax rate, $\hat{c}_i$ is either increasing, decreasing, or constant in $i$; in other words, $i \mapsto \hat{c}_i$ cannot be non-monotonic. Next note that, by (\ref{optimal income s}), $i \mapsto z_i$ is increasing in $i$ and that $i \mapsto g_i$ is decreasing/constant/increasing in $i$ if and only if $i \mapsto \hat{c}_i$ is increasing/constant/decreasing in $i$.  It follows from the Chebyshev integral inequality that
\begin{align}\label{equivalences}
\int g_i \l(\bar{z}-z_i\r) \dd i \l\{\begin{array}{c} > 0,\\
=0,\\
< 0,
\end{array}\r\} \Leftrightarrow i \mapsto \hat{c}_i  \l\{\begin{array}{l}\textup{is increasing},\\
\textup{is constant},\\
\textup{is decreasing},
\end{array}\r\}. 
\end{align}
Note that, if $\tau=0$, then the last line of (\ref{expression for derivative}) reduces to $\int g_i \l(\bar{z}-z_i\r) \dd i$, which is positive because a zero tax rate is order-preserving.  Hence, welfare could be increased by increasing the tax rate slightly, so that a zero tax rate cannot be optimal.  It follows that, at any $s$-optimal tax rate, the first order condition $\pdv{W}{\tau}=0$ must hold.  Moreover, at any $s$-optimal tax rate $B< 0$. 

By Proposition \ref{revenue efficiency prop}, $A \geq 0$ if and only if $\tau$ is revenue efficient at $s$, and, by Proposition \ref{revenue efficiency Proposition}, any optimal tax rate must be revenue efficient.  So, if an $s$-optimal tax rate were not order-preserving, then by (\ref{equivalences}), $A \int g_i \l(\bar{z}-z_i\r) \dd i \leq 0$, and so the last line of (\ref{expression for derivative}) would be negative, which implies that the first order condition is not satisfied, contradicting optimality. $\square$

\subsubsection{Uniqueness of the optimal tax rate}\label{uniqueness appendix}

\begin{lem}\label{increasing second derivative lemma}
At any revenue efficient tax rate $\tau$,    $\frac{\partial^2 \hat{c}_i}{\partial \tau^2}$ is increasing in $i$.
\end{lem}
\textit{Proof.} It follows from (\ref{hat c i in terms of w i}) that 
\begin{align*}
\hat{c}_i 
= A\l(\tau\r) w_i^{\varepsilon+1} + \l(1-s\tau\r)^{\varepsilon}\tau \int w_j^{1+\varepsilon} \dd j,
\end{align*}
where 
\begin{align*}
A\l(\tau\r) = B\l(\tau\r)C\l(\tau\r), \qquad B\l(\tau\r) = \l(1-s\tau\r)^{\varepsilon}, \qquad C\l(\tau\r) = \frac{1-\tau\l(1+\varepsilon\l(1-s\r)\r)}{1+\varepsilon}.
\end{align*}
To show that $\frac{\partial^2 \hat{c}_i}{\partial \tau^2}$ is increasing in $i$ it is sufficient to show that $A''\l(\tau\r) >0$ because  $w_i^{\varepsilon+1}$ is increasing in $i$ and does not depend on $\tau$ while $\l(1-s\tau\r)^{\varepsilon}\tau \int w_j^{1+\varepsilon} \dd j$ does not depend on $i$. 

We have
\begin{align*}
A'\l(\tau\r) = B'\l(\tau\r)C\l(\tau\r) + B\l(\tau\r)C'\l(\tau\r).
\end{align*}
Using the fact that $C''\l(\tau\r) =0$, we have
\begin{align*}
A''\l(\tau\r) &= B''\l(\tau\r)C\l(\tau\r) + 2B'\l(\tau\r)C'\l(\tau\r).
\end{align*} 
We have
\begin{align*}
B'(\tau) = -s\varepsilon\l(1-s\tau\r)^{\varepsilon-1}, \qquad B''\l(\tau\r) = s^2\varepsilon\l(\varepsilon-1\r)\l(1-s\tau\r)^{\varepsilon-2}, \qquad C'\l(\tau\r) = -\frac{1 + (1-s)\varepsilon}{1+\varepsilon}.
\end{align*}
It follows that
\begin{align*}
A''\l(\tau\r) &= s^2\varepsilon\l(\varepsilon-1\r)\l(1-s\tau\r)^{\varepsilon-2} \cdot \frac{1-\tau\l(1+\varepsilon\l(1-s\r)\r)}{1+\varepsilon} + 2\l(-s\varepsilon\r)\l(1-s\tau\r)^{\varepsilon-1} \cdot \l(-\frac{1 + (1-s)\varepsilon}{1+\varepsilon}\r) \\
&= \frac{s\varepsilon\l(1-s\tau\r)^{\varepsilon-2}}{1+\varepsilon} \l[ s\l(\varepsilon-1\r)\l(1-\tau\l(1+\varepsilon\l(1-s\r)\r)\r) + 2\l(1-s\tau\r)\l(1 + \l(1-s\r)\varepsilon\r) \r].
\end{align*}
Since $\frac{s\varepsilon\l(1-s\tau\r)^{\varepsilon-2}}{1+\varepsilon} > 0$,\footnote{See Remark \ref{domain remark}.} we have $A''(\tau) > 0$ if and only if
\begin{align*}
 s\l(\varepsilon-1\r)\l(1-\tau\l(1+\varepsilon\l(1-s\r)\r)\r) + 2\l(1-s\tau\r)\l(1 + \l(1-s\r)\varepsilon\r)> 0.
\end{align*}
Expanding and collecting terms that multiply $\tau$:
\begin{align*}
s\l(\varepsilon-1\r) + 2\l(1 + \l(1-s\r)\varepsilon\r) &> \tau \l[ s\l(\varepsilon-1\r)\l(1 + \l(1-s\r)\varepsilon\r) + 2s\l(1 + \l(1-s\r)\varepsilon\r) \r].
\end{align*}
The left-hand side simplifies to $\l(2-s\r)\l(1+\varepsilon\r)$, and on the right-hand side we can factor out $s\l(1 + \l(1-s\r)\varepsilon\r)$:
\begin{align*}
\l(2-s\r)\l(1+\varepsilon\r) &> \tau \cdot s\l(1 + \l(1-s\r)\varepsilon\r)\l(\varepsilon - 1 + 2\r) =\tau \cdot s\l(1+\varepsilon\r)\l(1 + \l(1-s\r)\varepsilon\r).
\end{align*}
Dividing both sides by $s\l(1+\varepsilon\r)\l(1 + \l(1-s\r)\varepsilon\r) > 0$, it follows that $A''\l(\tau\r) > 0$ if and only if
\begin{align*}
\tau < \frac{2-s}{s(1 + (1-s)\varepsilon)}.
\end{align*}
Revenue efficiency is characterized by
\begin{align*}
\tau \leq \frac{1}{s\l(1+\varepsilon\r)}. 
\end{align*}
So to establish the result, it is sufficient to show that
\begin{align*}
\frac{1}{s\l(1+\varepsilon\r)} < \frac{2-s}{s(1 + \l(1-s\r)\varepsilon)}.
\end{align*} 
Since $1 \leq 2-s$ and $1+ \varepsilon > 1+\l(1-s\r) \varepsilon$, the numerator is weakly larger and the denominator is strictly smaller in the right-hand-side fraction.  $\square$

\begin{lem}\label{strictly concave W}
 At any tax policy $\tau$ that is admissible at $s$,  $\pdv[2]{}{\tau}W\l(\tau,s\r) < 0$.
 \end{lem}

\noindent \textit{Proof.}  We have 
\begin{align*}
\pdv[2]{W}{\tau}= \underbrace{\int u''\l(\hat{c}_i\r) \times \l(\pdv{\hat{c}_i}{\tau}\r)^2 \dd i}_A + \underbrace{\int  u'\l(\hat{c}_i\r) \pdv[2]{\hat{c}_i}{\tau} \dd i}_B.
\end{align*}
Since $u''\l(\hat{c}_i\r) < 0$ and, for all but at most one $i$, $\pdv{\hat{c}_i}{\tau} \neq 0$,\footnote{We have $
\pdv{\hat c_i}{\tau}= \l(\bar{z} - z_i\r)
   + \l(1-s\r)\frac{s\tau\varepsilon}{1-s\tau} z_i
   - \frac{s\tau\varepsilon}{1-s\tau} \bar{z}
=
\l(
1-\frac{s\tau\varepsilon}{1-s\tau}
\r)\bar{z}
+
\l(
-1+(1-s)\frac{s\tau\varepsilon}{1-s\tau}
\r) z_i$. Observe that $\frac{s\tau\varepsilon}{1-s\tau} \leq 1 \Leftrightarrow \tau \leq \frac{1}{s\l(1+\varepsilon\r)}$, which, by Proposition \ref{revenue efficiency prop}, is equivalent to the condition that $\tau$ is revenue efficient.  Given our assumption that $s >0$, it follows that, when $\tau$ is admissible, and hence revenue efficient, $-1+(1-s)\frac{s\tau\varepsilon}{1-s\tau} <0$ and so $\pdv{\hat c_i}{\tau}$ is strictly decreasing in $i$, so that, since $z_i$ is strictly increasing in $i$, $\pdv{\hat c_i}{\tau}=0$ for at most one $i$.}  it follows that $A < 0$.  Next, observe that
\begin{align*}
B =&\; \l(\int u'\l(\hat{c}_i\r) \dd i \times \int \pdv[2]{\hat{c}_i}{\tau} \dd i\r) + \int_i \pdv[2]{\hat{c}_i}{\tau}\l(u'\l(\hat{c}_i\r)-\int u'\l(\hat{c}_j\r) \dd j\r) \dd i\\
 = &\; \l(\underbrace{\int u'\l(\hat{c}_i\r) \dd i}_+ \times \underbrace{\pdv[2]{\mu}{\tau}}_- \r)+ \underbrace{\int_i \pdv[2]{\hat{c}_i}{\tau}\l(u'\l(\hat{c}_i\r)-\int u'\l(\hat{c}_j\r) \dd j\r) \dd i}_- < 0,
\end{align*}
where we have used that (i) $u'\l(\hat{c}_i\r) > 0$, (ii) the fact that, for revenue efficient tax policies, $\pdv[2]{\mu}{\tau} <0$,\footnote{This is established by Lemma \ref{lem: effect on marginal efficiency cost of taxes}. (The hypothesis that $\tau > 0$ is not needed for the part of Lemma \ref{lem: effect on marginal efficiency cost of taxes} that establishes that $\pdv[2]{\mu}{\tau} <0$.)} and (iii) the Chebyshev integral inequality, appealing to the facts that $u'\l(\hat{c}_i\r)$ is decreasing in $i$ (because we are assuming an admissible and hence order-preserving tax rate), and that $\pdv[2]{\hat{c}_i}{\tau}$ is increasing in $i$ if $\tau$ is revenue-efficient (Lemma \ref{increasing second derivative lemma}). This completes the proof. $\square$

\begin{prop}\label{prop: uniqueness}
The $s$-optimal tax rate is unique and strictly positive.
\end{prop}
\textit{Proof.} By Proposition \ref{revenue efficiency Proposition} and Lemma \ref{increasing normal lemma} the optimal tax rate is revenue efficient and order-preserving, hence admissible.  By Lemma \ref{strictly concave W} the utilitarian objective is strictly concave on the admissible domain.  Hence, the optimal tax rate is unique.  During the course of the proof of Lemma \ref{increasing normal lemma}, it was shown that,  starting at $\tau=0$, a slight increase in the marginal tax rate increases utilitarian social welfare. $\square$

\subsubsection{Proof of Proposition \ref{optimal tax formula proposition}}
At a fixed level of salience, the social optimization problem is to choose $\tau$ to maximize:
\begin{align*}
W=\int_i u\l(z_i\l(1-\tau\r)-v_i\l(z_i\r)+\bar{z}\tau\r) \dd i.
\end{align*}
By Proposition \ref{prop: uniqueness}, the optimal tax rate is positive.  As in (\ref{expression for derivative}), the  first order condition is:
\begin{align*}
\pdv{W}{\tau} = \int g_i \l(\bar{z}-z_i\r)\dd i+\tau\l(1-s\r)
\frac{s \varepsilon}{1-s\tau}\int g_i z_i \dd i - \bar{g}\tau \frac{ s\varepsilon \bar{z}}{1-s\tau}   = 0.
\end{align*}
Collecting the terms involving $\frac{s\tau\varepsilon}{1-s\tau}$, this becomes
\begin{align*}
\int g_i \l(\bar{z}-z_i\r)\dd i - \frac{s\tau \varepsilon}{1-s \tau}\l[\bar{g}\bar{z}-\l(1-s\r)\int g_i z_i \dd i \r] =0.
\end{align*}
Recall that $\bar{h} = \frac{\int g_i z_i \dd i}{\bar{g}\bar{z}}$. 
Then, dividing through by $\bar{g}\bar{z}$, and solving for $\frac{s\tau\varepsilon}{1-s\tau}$, we have
\begin{align*}
\frac{s\tau \varepsilon }{1-s\tau} = \frac{1-\bar{h}}{1-\l(1-s\r)\bar{h}}.
\end{align*}
 $\square$

\subsubsection{When $\rho \leq 2$, the optimal marginal tax rate is decreasing in salience at $s=1$}

\begin{prop}\label{rho less than 2 proposition}
Suppose that $\rho \leq 2$.  Then the optimal tax rate is decreasing in salience at $s=1$.  That is, $\tau'\l(1\r) < 0$. 
\end{prop}
\textit{Proof.}  It follows from the implicit function theorem applied to the first order conditions of the utilitarian welfare maximization problem and Lemma \ref{strictly concave W} that, to establish that $\tau'\l(1\r)< 0$, it is sufficient to show that $\displaystyle \l.\pdv{W}{s}{\tau}\r|_{\tau=\tau\l(1\r),s=1}< 0$. 

Using (\ref{ci derivative s}), which shows that, when $s=1$, $\pdv{\hat{c}_i}{s}$ does not depend on $i$, we have:
\begin{align}\label{expression for -W}
\l.\pdv{W}{s}\r|_{s=1} = \int u'\l(\hat{c}_i\r) \pdv{\hat{c}_i}{s} \dd i = -\overline{\hat{c}^{-\rho}} \frac{\varepsilon \tau^2 \bar{z}}{1-\tau},  
\end{align}
where for any $k$, $\overline{\hat{c}^{-k}} = \int \hat{c}_i^{-k} \dd i$. By Theorem \ref{salience vs welfare proposition}, $\l.\pdv{W}{s}\r|_{\tau=\tau\l(s\r)}< 0$.\footnote{The proof shows that $\pdv{W}{s}< 0$, and not just that $W$ is decreasing in $s$.} Because $\pdv{\tau}\log\l(-\pdv{W}{s}\r)=\pdv{W}{s}{\tau}\big/\pdv{W}{s}$, it follows that $\l.\pdv{W}{s}{\tau}\r|_{\tau=\tau\l(1\r),s=1}$ and $\l.\pdv{\tau}\log\l(-\pdv{W}{s}\r)\r|_{\tau=\tau\l(1\r),s=1}$ have opposite signs. Thus, to complete the proof, it is sufficient to show that $\l.\pdv{\tau}\log\l(-\pdv{W}{s}\r)\r|_{\tau=\tau\l(1\r),s=1}>0$.

Using (\ref{expression for -W}), we have
\begin{align}\label{derivative of log Ws}
\l.\pdv{\tau}\log\l(-\pdv{W}{s}\r)\r|_{\tau=\tau\l(1\r),s=1} =\pdv{\tau} \log\overline{\hat{c}^{-\rho}}+ \frac{2}{\tau} +\frac{1-\varepsilon}{1-\tau},
\end{align}
where we have used the fact that, by (\ref{tax derivatives}), $\l.\pdv{\tau} \log \bar{z}\r|_{s=1}=\l.\l.\l(\pdv{\bar{z}}{\tau}\r)\r/\bar{z}\r|_{s=1} = -\frac{\varepsilon}{1-\tau}$. 

To analyze $\pdv{\tau} \log\overline{\hat{c}^{-\rho}}$, it is useful to do some preliminary calculations.  Observe that, by (\ref{simplification of an expression s}), at $s=1$, $\hat{c}_i=\frac{1-\tau}{1+\varepsilon}z_i +\tau \bar{z}$.  So, $z_i = \frac{1+\varepsilon}{1-\tau}\l(\hat{c}_i -\tau \bar{z}\r) $. It follows that 
\begin{align}\label{preliminary calculation}
\int \frac{\hat{c}^{-\rho-1}_i}{\overline{\hat{c}^{-\rho}}} z_i \dd i= \frac{1+\varepsilon}{1-\tau}\l(1-\frac{\overline{\hat{c}^{-\rho-1}}}{\overline{\hat{c}^{-\rho}}}\tau\bar{z}\r). 
\end{align}
Appealing to (\ref{money-metric utility}) and (\ref{tax derivatives}), and applying the envelope theorem (which is appropriate at $s=1$), we have $\l.\pdv{\hat{c}_i}{\tau}\r|_{s=1} = \l(\bar{z}-z_i\r) - \frac{\varepsilon\tau}{1-\tau}\bar{z}= \frac{1 - \tau\l(1+\varepsilon\r)}{1-\tau}\bar{z} - z_i$.
It follows that
\begin{align}\label{preliminary calculation 2}
   \pdv{\tau} \log \overline{\hat{c}^{-\rho}}= \frac{\pdv{\tau}\overline{\hat{c}^{-\rho}}}{\overline{\hat{c}^{-\rho}}}=  \rho\l[\int \frac{\hat{c}^{-\rho-1}_i}{\overline{\hat{c}^{-\rho}}} z_i \dd i -\frac{1-\tau\l(1+\varepsilon\r)}{1-\tau}\frac{\overline{\hat{c}^{-\rho-1}}}{\overline{\hat{c}^{-\rho}}} \bar{z} \r]= \frac{\rho}{1-\tau} \l[1+\varepsilon-\frac{\overline{\hat{c}^{-\rho-1}}}{\overline{\hat{c}^{-\rho}}}\bar{z}\r],
\end{align}
where the last equality follows from (\ref{preliminary calculation}). Multiplying both sides of (\ref{derivative of log Ws}) by $\tau\l(1-\tau\r)$ and appealing to (\ref{preliminary calculation 2}), it follows that $\pdv{\tau}\log \l(-\pdv{W}{s}\r)  > 0$ if and only if 
\begin{align}\label{wts for proof}
\rho \l[\l(1+\varepsilon\r)\tau-\frac{\overline{\hat{c}^{-\rho-1}}}{\overline{\hat{c}^{-\rho}}}\bar{z}\tau\r] + 2\l(1-\tau\r) +\tau\l(1-\varepsilon\r) > 0.\end{align}
Next observe that
$
\frac{\overline{\hat{c}^{-\rho}}}{\overline{\hat{c}^{-\rho-1}}} = \int \frac{\hat{c}_i^{-\rho-1}}{\overline{\hat{c}^{-\rho-1}}} \hat{c}_i \dd i \geq \tau\bar{z}$, where the last inequality follows from (\ref{simplification of an expression s}), evaluated at $s=1$.  So $
\frac{\overline{\hat{c}^{-\rho-1}}}{\overline{\hat{c}^{-\rho}}} \leq \frac{1}{\tau \bar{z}}$. Hence to establish (\ref{wts for proof}), it is sufficient to establish 
$\rho \l[\l(1+\varepsilon\r)\tau-1\r] + 2\l(1-\tau\r) +\tau\l(1-\varepsilon\r) > 0$, or equivalently, 
\begin{align}\label{final wts}
2 + \l(\rho-1\r) \l(1+\varepsilon\r)\tau > \rho. 
\end{align}
The above inequality always holds when $\rho \in \l[1,2\r]$.  Suppose that $\rho < 1$.  
By (\ref{useful optimal tax formula}), applied at $s=1$, the optimal tax rate satisfies $\frac{\tau \varepsilon}{1-\tau} = \l(1-\bar{h}\r)$. So $\tau \varepsilon = \l(1-\bar{h}\r)\l(1-\tau\r)$.  So $\tau\l(1+\varepsilon\r) =  \l(1-\bar{h}\r)\l(1-\tau\r) + \tau < \l(1-\tau\r) + \tau = 1$.  So to establish (\ref{final wts}), it is sufficient to establish that $2+\l(\rho-1\r) > \rho$, which always holds. This implies that $\l.\pdv{\tau}\log\l(-\pdv{W}{s}\r)\r|_{\tau=\tau\l(1\r),s=1}>0$, and so completes the proof. $\square$

\subsection{Equality, efficiency, and the price of equality}
We begin by defining notation to be used in Sections \ref{appendix: characterization atkinson} and \ref{appendix: supporting lemmsa}.  Let $\hat{C} = \l(\hat{c}_i\r)_{i \in \l[0,1\r]}$ be a profile of equivalent consumptions.  We extend the notation for  welfare, the equally distributed equivalent, and efficiency so that their arguments are equivalent consumption profiles.  Then $W\l(\hat{C}\r) = \int u\l(\hat{c}_i\r) \dd i$, $\xi\l(\hat{C}\r) = u^{-1}\l(W\l(\hat{C}\r)\r),$ and $ \mu\l(\hat{C}\r) = \int \hat{c}_i \dd i$.

\subsubsection{Characterization of the Atkinson equality measure}\label{appendix: characterization atkinson}
Say that $\hat{C}=\l(\hat{c}_i\r)_{i \in \l[0,1\r]}$ is \textbf{constant} if it assigns the same equivalent consumption to all agents.
\begin{prop} 
$E(\hat{C})= \xi(\hat{C})/\mu(\hat{C})$ is the unique equality measure satisfying
\begin{enumerate}
\item\label{first E property} $E(\hat{C})\mu(\hat{C}) \geq E(\hat{C}')\mu(\hat{C}')\Leftrightarrow W(\hat{C}) \geq W(\hat{C}'),$ for all  $\hat{C}, \hat{C}'$, 
\item\label{second E property}   $E(\hat{C}) = 1$ for all constant $\hat{C}$. 
\end{enumerate} 
\end{prop} 
Proof.  $E=\xi/\mu$ satisfies these two properties.  

Now consider an equality measure that satisfies the two properties. It follows from property \ref{first E property} that there is a strictly increasing transformation $f$ such that $f(W(\hat{C})) = E(\hat{C}) \mu(\hat{C})$.  Hence 
\begin{align}\label{equality form}
E\l(\hat{C}\r) = \frac{f\l(W\l(\hat{C}\r)\r)}{\mu\l(\hat{C}\r)}. 
\end{align}
Without property \ref{second E property}, this is all we could say about the form of the equality measure.  Let $\hat{C}^e=\l(\hat{c}^e_i\r)_{i \in \l[0,1\r]}$ be the constant allocation such that $\hat{c}^e_i = \xi\l(\hat{C}\r)$, for all $i$.  In other words, in $
\hat{C}^e$, all agents receive the equally distributed equivalent of $\hat{C}$.  Hence 
\begin{align}\label{mu C hat}
\mu\l(\hat{C}^e\r) = \xi\l(\hat{C}\r).
\end{align} 

By property \ref{second E property}, $E\l(\hat{C}^e\r)=1$.  So, by (\ref{equality form}),
\begin{align*}
1=\frac{f\l(W\l(\hat{C}^e\r)\r)}{\mu\l(\hat{C}^e\r)}=\frac{f\l(W\l(\hat{C}\r)\r)}{\xi\l(\hat{C}\r)},
\end{align*}
where the second equality invokes (\ref{mu C hat}) and the fact that, since $\hat{C}^e$ gives agents the equally distributed equivalent of $\hat{C}$, $W\l(\hat{C}\r)=W\l(\hat{C}^e\r)$. It follows that $f\l(W\l(\hat{C}\r)\r)=\xi\l(\hat{C}\r)$, so that, by (\ref{equality form}), $E=\xi/\mu$.  $\square$

\subsubsection{Lemmas supporting Proposition \ref{prop: partials}}\label{appendix: supporting lemmsa}
The results in this section connecting Lorenz dominance to inequality are standard in the inequality literature.  We include proofs for completeness and to establish the results in our setting and notation.  

Let $\hat{C}=\l(\hat{c}_i\r)_{i\in \l[0,1\r]}$ be a profile of equivalent consumptions such that $i \mapsto \hat{c}_i$ is increasing; we will refer to $\hat{C}$ as an \textbf{increasing} profile.   Define a parameterized family $\hat{C}^\theta=\l(\hat{c}_i\l(\theta\r)\r)_{i \in \l[0,1\r]}$ of increasing equivalent consumption profiles, and let $\mu\l(\theta\r)=\mu\l(\hat{C}^\theta\r)= \int \hat{c}_i\l(\theta\r) \dd i$.  Define 
\begin{align}\label{Lorenz definition}
L\l(i,\theta\r)= \frac{\int_0^i \hat{c}_j\l(\theta\r) \dd j}{\mu\l(\theta\r)}. 
\end{align}
Note that the map $i \mapsto L\l(i,\theta\r)$ is the Lorenz curve at $\theta$.  

\begin{lem}\label{effect on welfare proposition}
Let $\l(\hat{C}^\theta\r)_\theta$ be a parameterized family of increasing equivalent consumption profiles with constant mean equivalent consumption $\mu$, let $u$ be a strictly concave and strictly increasing function and let
\begin{align*}
W\l(\theta\r)= \int_0^1 u\l(\hat{c}_i\l(\theta\r)\r) \dd i.
\end{align*}
If, for all $i$ in $\l(0,1\r)$, $\pdv{\theta}L\l(i,\theta\r) < 0$, then $W'\l(\theta\r) < 0$.  
\end{lem}
\textit{Proof.}  It is convenient to write $\hat{c}\l(i,\theta\r)=\hat{c}_i\l(\theta\r)$. Using integration by parts, we have
\begin{align*}
W'\l(\theta\r)=&\;\int_0^1 u'\l(\hat{c}\l(i,\theta\r)\r) \pdv{\theta}\hat{c}\l(i,\theta\r) \dd i \\= &\; u'\l(\hat{c}\l(1,\theta\r)\r)  \int_0^1 \pdv{\theta} \hat{c}\l(i,\theta\r) \dd i-\underbrace{\int_0^1 u''\l(\hat{c}\l(i,\theta\r)\r) \pdv{i} \hat{c}\l(i,\theta\r) \int_0^i \pdv{\theta} \hat{c}\l(j,\theta\r) \dd j \dd i}_A.
\end{align*}
Next observe that
\begin{align*}
 \int_0^1 \pdv{\theta} \hat{c}\l(i,\theta\r) \dd i = \pdv{\theta} \int_0^1  \hat{c}\l(i,\theta\r) \dd i = \mu'\l(\theta\r)=0.
\end{align*}
Again, using the fact that $\mu'\l(\theta\r) = 0$, 
\begin{align*}
\pdv{\theta} L\l(i,\theta\r) = \frac{\pdv{\theta} \int_0^i \hat{c}\l(j,\theta\r) \dd j}{\mu}.
\end{align*}
So,
\begin{align*}
A=& \int_0^1 u''\l(\hat{c}\l(i,\theta\r)\r) \pdv{i} \hat{c}\l(i,\theta\r) \l[\pdv{\theta}\int_0^i  \hat{c}\l(j,\theta\r) \dd j\r] \dd i\\
=& \mu \int_0^1 \underbrace{u''\l(\hat{c}\l(i,\theta\r)\r)}_- \underbrace{\pdv{i} \hat{c}\l(i,\theta\r)}_+ \underbrace{\pdv{\theta}L\l(i,\theta\r)}_- \dd i >0.
\end{align*}
So $W'\l(\theta\r) < 0$.  $\square$

Unlike in the preceding lemma, the following lemma does not assume that the family $\l(\hat{C}^\theta\r)_\theta$ of increasing equivalent consumption profiles has constant mean.

\begin{lem}\label{useful inequality lemma}
Let equality be measured by the Atkinson index (\ref{Atkinson inequality index}). If, for all $i$ in $\l(0,1\r)$, $L\l(i,\theta\r)$ is $\begin{array}{l} \textit{increasing}\\\textit{constant}\\ \textit{decreasing} \end{array}$ in $\theta$, then equality is $\begin{array}{l} \textit{increasing}\\\textit{constant}\\ \textit{decreasing} \end{array}$ in $\theta$.
\end{lem}
\textit{Proof.}  Choose some $\mu^* >0$, and define the profile $\hat{C}^\theta_* =\l(\hat{c}^*_i\l(\theta\r)\r)_{i \in \l[0,1\r]}$ by: 
\begin{align*}
\hat{c}^*_i\l(\theta\r) = \hat{c}_i\l(\theta\r) \frac{\mu^*}{\mu\l(\hat{C}^\theta\r)}, \;\;\; \forall i, \forall \theta.
\end{align*}  
Then, by construction, $\mu\l(\hat{C}^\theta_*\r)=\mu^*, \forall \theta$.  
So $\l(\hat{C}_*^\theta\r)_\theta$, unlike $\l(\hat{C}^\theta\r)_\theta$, has a constant mean.  Define:
\begin{align*}
L^*\l(i,\theta\r) = \frac{\int_0^i \hat{c}^*_j\l(\theta\r) \dd j }{\mu^*},
\end{align*}  
so that $i \mapsto L^*\l(i,\theta\r)$ is the Lorenz curve for $\hat{C}^\theta_*$ whereas $i \mapsto L\l(i,\theta\r) $ is the Lorenz curve for $\hat{C}^\theta$.  It is easy to see that $L^*\l(i,\theta\r)=L\l(i,\theta\r), \forall i, \forall \theta$.

Now suppose that we measure the equality $E\l(\hat{C}\r)$ in any profile $\hat{C}=\l(\hat{c}_i\r)_{i \in \l[0,1\r]}$ by the Atkinson index
\begin{align*}
E\l(\hat{C}\r) = \frac{1}{\mu\l(\hat{C}\r)} \l(\int_0^1\hat{c}_i^{1-\rho} \dd i\r)^{\frac{1}{1-\rho}}.
\end{align*} 
The Atkinson index is scale invariant in the sense that, for all $\hat{C}$ and $\lambda  > 0$, $E\l(\lambda \hat{C}\r)=E\l(\hat{C}\r)$, where $\lambda \hat{C}= \l(\lambda \hat{c}_i\r)_{i\in \l[0,1\r]}$.  This is the equality index that corresponds to the utility function $u\l(\hat{c}\r)= \frac{\hat{c}^{1-\rho}}{1-\rho}$ in the decomposition developed in Section \ref{sec: decomposition derivation}.  Define $W^*\l(\theta\r)= \int_0^1\frac{ \l(\hat{c}^*_i\l(\theta\r)\r)^{1-\rho}}{1-\rho} \dd i$. Lemma \ref{effect on welfare proposition} implies that, if $L\l(i,\theta\r)$ is decreasing in $\theta$, $W^*\l(\theta\r)$ is decreasing in $\theta$.  Since $W^*\l(\theta\r)$ is ordinally equivalent to $\xi\l(\theta\r) :=  E\l(\hat{C}^\theta_*\r)\mu\l(\hat{C}^\theta_*\r)$, it follows that $\xi\l(\theta\r)$ is decreasing in $\theta$ as well.  Since $\mu\l(\hat{C}^\theta_*\r)=\mu^*$ is constant in $\theta$, it follows that $E\l(\hat{C}_*^\theta\r)$ is decreasing in $\theta$.  Since $E$ is scale invariant and $\hat{C}_*^\theta = \frac{\mu^*}{\mu\l(\hat{C}^\theta\r)} \hat{C}^\theta, \forall \theta$, $E\l(\hat{C}_*^\theta\r)=E\l(\hat{C}^\theta\r), \forall \theta$.  So $E\l(\hat{C}^\theta\r)$ is decreasing in $\theta$ as well.  The argument in the case that $L\l(i,\theta\r)$ is increasing or constant in $\theta$ is similar.  This completes the proof of Lemma \ref{useful inequality lemma}.  $\square$

\begin{lem}\label{Lorenz lemma}
Let $\hat{C}^\theta=\l(\hat{c}_i\l(\theta\r)\r)_{i \in \l[0,1\r]}$ be a family of increasing equivalent consumption profiles.  Let equality be measured by the Atkinson index (\ref{Atkinson inequality index}). Then, if $\pdv{\hat{c}_i}{\theta}/ \hat{c}_i$ is $\begin{array}{l}\textit{increasing}  \\ \textit{decreasing}\end{array}$  in $i$, equality is $\begin{array}{l} \textit{decreasing}\\ \textit{increasing} \end{array}$ in $\theta$.  
\end{lem}
\textit{Proof.}  Define $k_i\l(\theta\r):=\pdv{\hat{c}_i}{\theta}/ \hat{c}_i$. We have
\begin{align*}
\pdv{\theta} L\l(i,\theta\r) = &\; \frac{\int_0^i \hat{c}'_j\l(\theta\r) \dd j}{\mu\l(\theta\r)}- \frac{\int_0^i \hat{c}_j\l(\theta\r) \dd j}{\l[\mu\l(\theta\r)\r]^2}\mu'\l(\theta\r)\\
=&\; \frac{\int_0^i \hat{c}_j\l(\theta\r) \dd j}{\mu\l(\theta\r)}\l[\frac{\int_0^i \hat{c}'_j\l(\theta\r) \dd j}{\int_0^i \hat{c}_j\l(\theta\r) \dd j}-\frac{\int_0^1 \hat{c}'_j\l(\theta\r) \dd j}{\int_0^1 \hat{c}_j\l(\theta\r) \dd j} \r]\\
=&\; \frac{\int_0^i \hat{c}_j\l(\theta\r) \dd j}{\mu\l(\theta\r)}\l[\int_0^i\frac{ \hat{c}_j\l(\theta\r)}{\int_0^i \hat{c}_j\l(\theta\r) \dd j}k_j\l(\theta\r) \dd j-\int_0^1\frac{ \hat{c}_j\l(\theta\r)}{\int_0^1 \hat{c}_j\l(\theta\r) \dd j}k_j\l(\theta\r) \dd j \r]
\end{align*}
Note that the expression in brackets is the difference between a weighted average of the $k_j$ terms, one over the interval $\l[0,i\r]$, and the other over the interval $\l[0,1\r]$, with common relative weights on the interval $\l[0,i\r]$.  It follows that if $k_i\l(\theta\r)$ is increasing (resp., decreasing) in $i$, then $L\l(i,\theta\r)$ is decreasing (resp., increasing) in $\theta$. Lemma \ref{useful inequality lemma} now completes the proof.  $\square$

\subsubsection{Proof of Proposition \ref{prop: partials}} \label{proof of prop 4}

With regard to efficiency, we have proved in the main text that $\pdv{\mu}{s} < 0$. Similarly,   
\begin{align*}
\pdv{\tau} \mu\l(\tau,s\r) = \underbrace{\int \l(\bar{z}-z_i\r) \dd i}_0 +\tau\l(1-s\r) \frac{s \varepsilon \bar{z}}{1-s\tau}  -\tau\frac{s \varepsilon \bar{z}}{1-s\tau}= - \frac{s^2 \tau \varepsilon \bar{z}}{1-s\tau} < 0.
\end{align*}
Turning to equality, and given the assumption that we are at an admissible tax policy, and hence one that generates an increasing equivalent consumption profile, it follows from Lemma \ref{Lorenz lemma} that it is sufficient to show that $\pdv{\hat{c}_i}{s}/\hat{c}_i$ is increasing in $i$ and $\pdv{\hat{c}_i}{\tau}/\hat{c}_i$ is decreasing in $i$.

We have:
\begin{align*}
\l.\pdv{\hat c_i}{s}\r/\hat{c}_i \;& =
\phantom{\underbrace{\frac{\l(\bar{z} - z_i\r)}{\hat{c}_i}}_{\text{redistribution}}\;+\;}
\underbrace{(1-s)\frac{\tau^{2}\varepsilon }{1-s\tau}\frac{z_i}{\hat{c}_i}}_{\text{internality}}
\underbrace{-\frac{\tau^{2}\varepsilon }{1-s\tau}\frac{\bar{z}}{\hat{c}_i}}_{\text{fiscal externality}}
\\[0.7em]
\l.\pdv{\hat{c}_i}{\tau}\r/\hat{c}_i \;&=
\underbrace{\frac{\l(\bar{z} - z_i\r)}{\hat{c}_i}}_{\text{redistribution}}\;+\;
\underbrace{(1-s)\frac{s\tau \varepsilon }{1-s\tau}\frac{z_i}{\hat{c}_i}}_{\text{internality}}
\underbrace{-\frac{s\tau \varepsilon }{1-s\tau}\frac{\bar{z}}{\hat{c}_i}}_{\text{fiscal externality}}\\
&=  \l[\l(1-s\r) \frac{s\tau \varepsilon}{1-s\tau} - 1\r] \frac{z_i}{\hat{c}_i}+\l[1 - \frac{s\tau \varepsilon}{1-s\tau}\r] \frac{\bar{z}}{\hat{c}_i}.
\end{align*}
It follows from calculations in Section \ref{relation c z subsection} that we can write $\hat{c}_i =a z_i + b$ for positive constants $a$ and $b$ that depend on $\tau$ and $s$.  So, since $z_i$ is increasing in $i$, for fixed $\tau$ and $s$, $\frac{z_i}{\hat{c}_i}$ is increasing in $i$.  On the other hand, because $\hat{c}_i$ is increasing in $i$, $\frac{\bar{z}}{\hat{c}_i}$ is decreasing in $i$. In the expression for $\pdv{\hat c_i}{s}/\hat{c}_i$, Remark \ref{domain remark} implies that $s\tau < 1$, and hence that, when $\tau >0$, the coefficient on $\frac{z_i}{\hat{c}_i}$ is positive while that on $\frac{\bar{z}}{\hat{c}_i}$ is negative, so that $\pdv{\hat c_i}{s}/\hat{c}_i$ is increasing in $i$.  Revenue efficiency implies that $\frac{s\tau \varepsilon}{1-s\tau} \leq 1$, so that, in the expression for $\pdv{\hat c_i}{\tau}/\hat{c}_i$, the coefficient on $\frac{z_i}{\hat{c}_i}$ is negative and that on $\frac{\bar{z}}{\hat{c}_i}$ is nonnegative, so that $\pdv{\hat c_i}{\tau}/\hat{c}_i$ is decreasing in $i$.  This completes the proof. $\square$

\subsubsection{Comparative statics of $\rho$}\label{appendix: comparative statics of rho}
This section proves that the optimal tax rate is increasing in $\rho$, efficiency is decreasing in $\rho$, and equality is increasing in $\rho$.  In the case of the latter, our formulation in Corollary \ref{cor: comparative statics} accounts for the fact that the equality measure itself depends on $\rho$.  

In what follows, let $\tau\l(\rho,s\r)$ be the 
$s$-optimal tax rate under inequality aversion $\rho$.

\begin{prop}\label{prop: rho comparative statics}
For any $\rho$ and $s$, $\l.\pdv{W}{\rho}{\tau}\r|_{\tau=\tau\l(\rho,s\r)} > 0.$
Consequently, $\pdv{\tau}{\rho} >
0.$
\end{prop}
Proof. First, observe that
\begin{align*}
\pdv{\hat c_i}{\tau}= \l( 1- \frac{s\tau\varepsilon}{1-s\tau} \r)\bar z -\underbrace{\l(
1- \l(1-s\r) \frac{s\tau\varepsilon}{1-s\tau} \r)}_Az_i.
\end{align*}
Let $\tau = \tau\l(\rho,s\r)$.  Since $\tau$ is $s$-optimal, it is revenue efficient, so that $\tau \leq
\frac{1}{s\l(1+\varepsilon\r)}$, or equivalently, $\frac{s\tau\varepsilon}{1-s\tau} \leq 1.$
It follows that $A > 0$. So $\pdv{\hat c_i}{\tau}$ is strictly decreasing in $i$, since $z_i$ is strictly increasing in $i$. Since $\tau$ is $s$-optimal, it is order-preserving and hence $\log\l(\hat c_i\r)$ is strictly increasing in
$i$.  We use these facts when appealing to the Chebyshev integral inequality below. 

The first order condition is:
\begin{align}\label{W FOC}
\pdv{W}{\tau} = \int \hat c_i^{-\rho} \pdv{\hat{c}_i}{\tau} \dd i=0.
\end{align}
Hence, differentiating with respect to $\rho$ gives
\begin{align*}
\l.\pdv{W}{\rho}{\tau}\r|_{\tau=\tau\l(\rho,s\r)} = &\;- \int \log\l(\hat{c}_i\r) \hat{c}_i^{-\rho} \pdv{\hat{c}_i}{\tau} \dd i\\
=&\; - \int \log\l(\hat{c}_i\r) \hat{c}_i^{-\rho} \pdv{\hat{c}_i}{\tau} \dd i - \l(\int \log\l(\hat{c}_i\r) \frac{\hat{c}^{-\rho}_i}{\overline{\hat{c}^{-\rho}}} \dd i\r)\times \l(\int \hat c_i^{-\rho} \pdv{\hat{c}_i}{\tau} \dd i\r)\\
= & \; -\overline{\hat{c}^{-\rho}}\l[\int \log\l(\hat{c}_i\r)  \pdv{\hat{c}_i}{\tau} \frac{\hat{c}_i^{-\rho}}{\overline{\hat{c}^{-\rho}}}\dd i - \l(\int \log\l(\hat{c}_i\r) \frac{\hat{c}^{-\rho}_i}{\overline{\hat{c}^{-\rho}}} \dd i\r)\times \l(\int \pdv{\hat{c}_i}{\tau}\frac{\hat{c}_i^{-\rho}}{\overline{\hat{c}^{-\rho}}} \dd i\r)\r]> 0,
\end{align*}
where $\overline{\hat{c}^{-\rho}}=\int \hat{c}_i^{-\rho} \dd i$. The second equality follows from the first order condition (\ref{W FOC}), and the inequality follows from the Chebyshev integral inequality.  

Applying the implicit function theorem to the first order
condition $W_\tau=0$, and appealing to the fact that at any admissible tax policy, and hence in particular at the optimum, $W_{\tau\tau} < 0$ (Lemma \ref{strictly concave W}), it follows that  $\pdv{\tau}{\rho}=-\frac{W_{\tau\rho}}{W_{\tau\tau}} >0$.
$\square$

The equality measure $E=\frac{1}{\mu} \l(\int_0^1 \hat{c}^{1-\rho}_i \dd i\r)^{\frac{1}{1-\rho}}$ depends on $\rho$, so we write $E_\rho$ to make the dependence explicit.  

\begin{cor}\label{cor: comparative statics}
Fix some $\rho_0 > 0$.  Then $\mu\l(\tau\l(\rho,s\r),s\r)$ is decreasing in $\rho$ and $E_{\rho_0}\l(\tau\l(\rho,s\r),s\r)$ is increasing in $\rho$.  
\end{cor}
Proof.  It follows from Proposition \ref{prop: partials} that $\mu\l(\tau,s\r)$ is decreasing in $\tau$ and $E_{\rho_0}\l(\tau,s\r)$ is increasing in $\tau$ over the admissible domain.  (Note that the admissible domain is independent of $\rho$.)  The result now follows from Proposition \ref{prop: rho comparative statics}. $\square$

\begin{rem}
One might wonder why the result about equality in  Corollary \ref{cor: comparative statics} is independent of the choice of $\rho_0$.  The reason is that $\pdv{E}{\tau} > 0$ in Proposition \ref{prop: partials} was proved by appeal to Lorenz dominance, and hence applies to the entire Atkinson class $E_{\rho_0}$.  
\end{rem}

\subsubsection{Properties of the optimum of the $E$-version of the welfare maximization problem}

\begin{prop}\label{prop: E-version properties}
The optimal level of equality in the $E$-version of the welfare maximization problem (Section \ref{sec: change of variables}) is unique, lies in the open interval $\l(0,1\r)$, and is interior to the set of admissible equality levels at $s$.  Moreover, at the optimum, the second-order condition holds with strict inequality.  
\end{prop}
\textit{Proof.}  The domain for the $E$-version of the welfare maximization problem is the set of admissible equality levels, which correspond one-to-one with the admissible tax rates.  Since there is a unique optimal tax rate (Proposition \ref{prop: uniqueness}) in the $\tau$-version of the problem, there is a unique optimal equality level in the $E$-version.  

At the optimum, $\hat{c}_i$ is increasing in $i$ because the tax rate must be order-preserving (Lemma~\ref{increasing normal lemma}). Thus, consumption equivalents must be unequally distributed, so $\xi < \mu$. It follows that $E =\frac{\xi}{\mu} < 1$. On the other hand, since $\hat{c}_i >0, \forall i$ (Proposition \ref{prop: positivity}), $\xi >0$, and hence $E>0$. Let $\tau^{\mathrm{constant}}= \frac{1}{1+\varepsilon\l(1-s\r)}$ be the tax rate at which all agents attain the same utility, and recall that $\tau^{\mathrm{max}}=\frac{1}{s\l(1+\varepsilon\r)}$ is the revenue-maximizing tax rate.  Then the set of admissible equality levels at $s$ is the interval $\l[E\l(0,s\r),E\l(\tau^{\mathrm{constant}},s\r)\r)$ if $\tau^{\mathrm{constant}}\leq\tau^{\mathrm{max}} $, and $\l[E\l(0,s\r),E\l(\tau^{\mathrm{max}},s\r)\r]$ if $\tau^{\mathrm{max}}< \tau^{\mathrm{constant}}$. It follows from Proposition \ref{prop: uniqueness} that a zero tax rate cannot be optimal, and hence $E\l(0,s\r)$ cannot be the $s$-optimal equality level.  Proposition \ref{prop: revenue maximizing not optimal} implies that $E\l(\tau^{\mathrm{max}},s\r)$ cannot be $s$-optimal.  Thus, the optimal level of equality is interior to the set of admissible equality levels. 

Next observe that:
\begin{align*}
\check{\xi}\l(E,s\r) = u^{-1}\l(W\l(\check{\tau}\l(E,s\r),s\r)\r).
\end{align*}
Then we have $\l.\pdv{\check{\xi}}{E} = \pdv{W}{\tau} \pdv{\check{\tau}}{E}\r/u'\l(\check{\xi}\r)$.  Hence, 
\begin{align*}
\pdv[2]{\check{\xi}}{E} = \pdv[2]{W}{\tau} \frac{ \l(\pdv{\check{\tau}}{E}\r)^2} {u'\l(\check{\xi}\r)}+ \pdv{W}{\tau} \times \pdv{E} \l[\frac{\pdv{\check{\tau}}{E}}{u'\l(\check{\xi}\r)}\r].
\end{align*}
We have (i) $\frac{ \l(\pdv{\check{\tau}}{E}\r)^2} {u'\l(\check{\xi}\r)}>0$, (ii) by Lemma \ref{strictly concave W}, $\pdv[2]{W}{\tau}<0$, and (iii) when $E$ is optimal, and hence also $\check{\tau}\l(E,s\r)$ is optimal, then $\pdv{W}{\tau}=0$, hence also, $\pdv{W}{\tau} \times \pdv{E} \l[\frac{\pdv{\check{\tau}}{E}}{u'\l(\check{\xi}\r)}\r]=0$.  So, at the $s$-optimal $E$, $\pdv[2]{\check{\xi}}{E} <0$. $\square$

\subsubsection{Proof of Lemma \ref{lem:decomposition}: Income and substitution effects on equality}\label{app: decomposition}
Differentiating equation (\ref{compensated definition}), we have:
\begin{align*}
\l.\pdv{E^\mathrm{c}}{s_1}\r|_{s_1=s_0} =&\l.\pdv{E^{\mathrm{rel}}}{s}\r|_{\substack{s = s_0 \\ \delta = 0}}  + \l.\pdv{E^{\mathrm{rel}}}{\delta}\r|_{\substack{s = s_0 \\ \delta = 0}}   \times \l.\pdv{\delta}{s_1}\r|_{s_1=s_0} \\
= & \l. \dv{E}{s}\r|_{s =s_0} -\l.\pdv{E^{\mathrm{rel}}}{\delta}\r|_{\substack{s = s_0 \\ \delta = 0}}\times \l.\pdv{\check{\mu}}{s}\r|_{\substack{E= E\l(s_0\r) \\ s=s_0}},
\end{align*}
where the second equality follows from the fact that $E\l(s\r) = E^{\mathrm{rel}}\l(s,0\r)$ and (\ref{exact compensation}).  The result now follows by rearranging terms to solve for $\dv{E}{s}$. $\square$ 

\subsubsection{Monotonicity of equality in frontier expansion}\label{app: monotonicity appendix}
\begin{lem}\label{lem: E rel monotonicty}
$\displaystyle\l.\pdv{E^{\mathrm{rel}}}{\delta}\r|_{\delta =0} >0$.
\end{lem}
\textit{Proof.}  We can equivalently reformulate the relaxed welfare maximization problem of Section \ref{Sec: Decomposing} as follows:

\vspace{0.8em}
\noindent \textbf{Welfare Maximization (reformulation of relaxed version) { }}\\ Choose $E$ to maximize $\hat{\xi}\l(E,s,\delta\r):=  E \times \l(\check{\mu}\l(E,s\r) + \delta\r)$.
\vspace{0.8em}

\noindent As this is an equivalent version of the problem, the optimal level of equality remains $E^{\mathrm{rel}}\l(s,\delta\r)$. 

We have:
\begin{align*}
\pdv{\hat{\xi}}{\delta} =&\; E,
\\
\pdv{\hat{\xi}}{E}{\delta} =&\; 1. 
\end{align*}
Note that, when $\delta =0$, $E^{\mathrm{rel}}\l(s,\delta\r)=E\l(s\r)$.  By the implicit function theorem, applied to the first order condition of the relaxed version of the welfare maximization problem,
\begin{align*}
\l.\pdv{E^{\mathrm{rel}}}{\delta}\r|_{\delta =0} = \l.-\frac{\pdv{\hat{\xi}}{\delta}{E} }{\pdv[2]{\hat{\xi}}{E}}\r|_{\substack{E= E\l(s\r) \\ \delta = 0}} = -\frac{1}{\l.\pdv[2]{\hat{\xi}}{E}\r|_{\substack{E= E\l(s\r) \\ \delta = 0}}}> 0,
\end{align*}
where we have appealed to Proposition \ref{prop: E-version properties} for $\l.\pdv[2]{\hat{\xi}}{E}\r|_{\substack{E= E\l(s\r) \\ \delta = 0}}< 0$, as the relaxed and unrelaxed versions of the problem coincide when $\delta =0$. $\square$

\subsubsection{Proof of Lemma \ref{lem:substitution}: The substitution effect and the price of equality}

Define the following optimization problem:
  
 \vspace{0.8em}
\noindent \textbf{Compensated welfare maximization problem.}\\ 
Choose $E$ to maximize $\xi^\circ\l(E,s_0,s_1\r) := E \times \l(\check{\mu}\l(E,s_1\r) + \delta\l(s_0,s_1\r)\r)$.  
\vspace{0.8em}

\noindent It follows from (\ref{compensated definition}) that the optimal solution to this compensated problem is just the compensated demand for equality $E^{\mathrm{c}}\l(s_0,s_1\r)$.  Applying the implicit function theorem to the first order condition of the compensated welfare maximization problem, and noting that, when $s_1=s_0$, $E^\mathrm{c}\l(s_0,s_1\r) = E\l(s_0\r)$, we have 
\begin{align*}
\l.\pdv{E^\mathrm{c}}{s_1}\r|_{s_1=s_0} = \l.-\frac{\pdv{\xi^\circ}{s_1}{E} }{\pdv[2]{\xi^\circ}{E}}\r|_{\substack{E= E\l(s_0\r) \\ s_1=s_0}}.
\end{align*}
Because the compensated welfare maximization problem coincides with the $E$-version of the welfare maximization problem when $s_1=s_0$, Proposition \ref{prop: E-version properties} implies that $\l.\pdv[2]{\xi^\circ}{E}\r|_{\substack{E= E\l(s_0\r) \\ s_1=s_0}}< 0$. Thus,
\begin{align}\label{equality of two signs}
\textup{sign}\l(\l.\pdv{E^c}{s_1}\r|_{s_1=s_0}\r)= \textup{sign}\l(\l.\pdv{\xi^\circ}{s_1}{E}\r|_{\substack{E= E\l(s_0\r) \\ s_1=s_0}}\r).
\end{align}
  Observe that $\pdv{\xi^\circ}{E} =  \l(\check{\mu}+ \delta\r)+ E  \pdv{\check{\mu}}{E}$.  So, 
\begin{align}\label{cross partial derivation}
\begin{split}
\l.\pdv{\xi^\circ}{s_1}{E}\r|_{\substack{E= E\l(s_0\r) \\ s_1=s_0}} = &\;\underbrace{\l.\pdv{\check{\mu}}{s}\r|_{\substack{E= E\l(s_0\r) \\ s=s_0}} +\l. \pdv{\delta}{s_1}\r|_{s_1=s_0}}_0 +E\l(s_0\r) \times \l.\pdv{\check{\mu}}{s}{E}\r|_{\substack{E= E\l(s_0\r) \\ s=s_0}}\\
=&\; E\l(s_0\r) \times \l.\pdv{\check{\mu}}{s}{E}\r|_{\substack{E= E\l(s_0\r) \\ s=s_0}}\\
=&\; E\l(s_0\r) \times \l(- \pdv{p_E}{s}\l(E\l(s_0\r),s_0\r)\r),
\end{split}
\end{align}
where the fact that two terms indicated above sum to zero follows from (\ref{exact compensation}), and the last equality follows from (\ref{definition of pE}), the definition of $p_E$. Proposition \ref{prop: E-version properties} implies that $E\l(s_0\r) >0$.  It now follows from (\ref{equality of two signs}) and (\ref{cross partial derivation}) that $\textup{sign}\l(\l.\pdv{E^c}{s_1}\r|_{s_1=s_0}\r)= \textup{sign}\l(- \pdv{p_E}{s}\l(E\l(s_0\r),s_0\r)\r)$. 
 $\square$

 \subsubsection{$-\mu_{s\tau}>0$ and $-\mu_{\tau\tau}>0$.}

\begin{lem}\label{lem: effect on marginal efficiency cost of taxes}
Assume that $\tau$ is admissible and $\tau > 0$.  Then $-\mu_{s\tau}>0$ and $-\mu_{\tau\tau}>0$.
\end{lem}
\textit{Proof.}  We have
 \begin{align*}
 \mu =&\; \int \hat{c}_i \dd i =  \int\l[ z_i\l(1-\tau\r) - v_i\l(z_i\r) +\tau \bar{z}\r]\dd i = \int \l[z_i - v_i\l(z_i\r) \r]\dd i,
  \end{align*}
where $\int \l[z_i - v_i\l(z_i\r) \r]\dd i$ can be interpreted as the sum of agent money-metric utilities prior to redistribution.  That redistributive transfers can be eliminated from the expression for $\mu$ reflects that transfers do not affect efficiency.  It follows that the marginal efficiency cost of the tax rate is the sum of the marginal reductions in pre-redistribution money-metric utility across agents due to behavioral responses to the tax rate: 
\begin{align*}
    -\mu_\tau = -\int \l(1-v'_i\l(z_i\r)\r) \pdv{z_i}{\tau}\dd i.
\end{align*}
Differentiating with respect to $s$, it follows that 
\begin{align*}
-\mu_{s\tau} =   \int \underbrace{v_i''\l(z_i\r)}_+ \underbrace{\pdv{z_i}{s}}_-\underbrace{\pdv{z_i}{\tau}}_- \dd i - \int \underbrace{\l(1-v'_i\l(z_i\r)\r)}_+ \underbrace{\pdv{z_i}{s}{\tau}}_- \dd i >0.
\end{align*}
All these signs are derived in Section \ref{sec: comparative statics of income}. (The only one which is not obvious is $\pdv{z_i}{s}{\tau}$.)  Similarly, 
\begin{align*}
-\mu_{\tau\tau} =&\;  \int \underbrace{v_i''\l(z_i\r)}_+\underbrace{\l( \pdv{z_i}{\tau}\r)^2}_+ \dd i -  \int\underbrace{\l(1- v'_i\l(z_i\r)\r)}_+\underbrace{\pdv[2]{z_i}{\tau}}_? \dd i\\
=&\; \int \frac{1}{\varepsilon} \frac{1-s\tau}{z_i} \l(- \frac{\varepsilon s}{1-s\tau}z_i\r)^2 \dd i - \int s\tau\l( -\frac{\varepsilon(1-\varepsilon)s^2}{(1-s\tau)^2}z_i\r) \dd i\\
= &\; \frac{\varepsilon s^2}{\l(1-s\tau\r)^2}\l[ \l(1-s\tau\r)+ s\tau \l(1-\varepsilon\r)\r]\bar{z}\\
=&\; \frac{\varepsilon s^2}{\l(1-s\tau\r)^2} \l(1-s\tau \varepsilon \r) \bar{z} > 0.
\end{align*}
Unlike $\mu_{s\tau}$, $\mu_{\tau\tau}$ could not be signed on the basis of the signs of the individual terms alone because the sign of $\pdv[2]{z_i}{\tau}$ is ambiguous (see Section \ref{sec: comparative statics of income} for the signs).  So it is necessary to plug in specific expressions from Section \ref{sec: comparative statics of income} and compute the overall sign of the expression.  We use the fact that revenue efficiency implies that $s\tau \varepsilon < 1$ and our domain restriction implies $s\tau < 1$ (see Remarks \ref{domain remark} and \ref{domain restriction remark 2}). $\square$

\subsubsection{Along an iso-equality path, the marginal equality benefit of the tax rate falls as salience increases}\label{sec: marginal equality benefit}

We would like to show that 

\begin{align}\label{E tau along the path}\dv{s}E_\tau\l(\check{\tau}\l(E,s\r),s\r) < 0.\end{align}  By Proposition \ref{prop: E-version properties}, for the broader argument to establish Theorem \ref{equality theorem}, it is sufficient to restrict attention to equality levels $E$ that are interior to the set of admissible equality levels at $s$.\footnote{See the proof of Proposition \ref{prop: E-version properties} for a discussion of the interior set.} We explain in Section \ref{sec:marginal equality benefit} why, in order to establish (\ref{E tau along the path}), it is sufficient to establish the following ordinary differential equation in $s$, fixing $E$.
\begin{lem}
For all $s \in \l(0,1\r]$ and all equality levels $E$ that are interior to the set of admissible equality levels at $s$,
\begin{align}\label{ODE}
    \check{\tau}_s\l(E,s\r) = \varepsilon \l[\check{\tau}\l(E,s\r)\r]^2.
    \end{align}
\end{lem}
\textit{Proof.}  It follows from (\ref{simplification of an expression s}) that
\begin{align}\label{equivalent consumption simplification}
\hat{c}_i=\l[\l(1-\tau\r)-\frac{\varepsilon}{1+\varepsilon}\l(1-s\tau\r)\r] z_i +\tau \bar{z}.
\end{align}
We can split equivalent consumption (or money-metric utility) $\hat{c}_i$ into \textit{common} and \textit{individual specific} parts.  The common part is the tax rebate $\tau \bar{z}$ received by all agents.  The individual specific part is $a\l(\tau,s\r) z_i$, where
\begin{align}\label{average utility per income}
a\l(\tau,s\r)=\l(1-\tau\r)-\frac{\varepsilon}{1+\varepsilon}\l(1-s\tau\r)
\end{align}
is the average individual specific utility per unit of income, equal to the retention rate $\l(1-\tau\r)$ less the average effort cost per unit of income $\frac{v_i\l(z_i\r)}{z_i}= \frac{\varepsilon}{1+\varepsilon}\l(1-s\tau\r)$. As agents optimize in response to perceptions, this average does not depend on $i$ (see Section \ref{relation c z subsection}).  

As $s$ changes,  by definition, the tax rate $\check{\tau}\l(E,s\r)$ must vary so as to keep equality fixed.  The proof that follows argues that, in order for equality to be maintained along an iso-equality path, average individual-specific utility per income $a$ and the tax rate $\tau$ must grow at the same rate as salience changes.  The proof appeals to the fact that $z_i$ and $\bar{z}$ must grow at the same rate, but $z_i$ is increasing in $i$.  If $a$ grows faster than $\tau$, the individual-specific part of utility increases more quickly than the common part and inequality increases (because $a$ multiplies $z_i$ in $\hat{c}_i$).  If $\tau$ increases faster, the common part increases more quickly, and equality increases.  The condition that $a$ and $\tau$ grow at the same rate yields (\ref{ODE}).

First, we explain why $z_i$ and $\bar{z}$ must grow at the same rate.  By (\ref{optimal income s}), when agents best respond to perceptions, $z_i\l(\tau,s\r) = w_i^{1+\varepsilon}\l(1-s\tau\r)^\varepsilon$. It follows that $z_i\l(\tau,s\r) = k_i\bar{z}\l(\tau,s\r)$, where $k_i =\frac{w_i^{1+\varepsilon}}{\int w_j^{1+\varepsilon}\dd j}$. Note that $k_i$ is increasing in $i$.

Appealing to (\ref{average utility per income}), equation (\ref{equivalent consumption simplification}) can now be rewritten as $\hat{c}_i = \bar{z}\l[k_ia + \tau\r]$.  Expressing this relation as a function of $s$, along an iso-equality path for a fixed level of equality $E$, we have
\begin{align*}
\hat{c}_i\l(\check{\tau}\l(E,s\r),s\r) = \bar{z}\l(\check{\tau}\l(E,s\r),s\r) \l[k_i a\l(\check{\tau}\l(E,s\r),s\r) +\check{\tau}\l(E,s\r)\r].  
\end{align*}

Next we examine the growth rate of $\hat{c}_i$ in $s$ as we move along an iso-equality path.  
\begin{align}\label{growth ci path}
\dv{s}\log \hat{c}_i\l(\check{\tau}\l(E,s\r),s\r) = \underbrace{\dv{s}\log \bar{z}\l(\check{\tau}\l(E,s\r),s\r)}_A+ \underbrace{\dv{s}\log \l[k_i a\l(\check{\tau}\l(E,s\r),s\r) +\check{\tau}\l(E,s\r)\r]}_{B_i}.  
\end{align}
We now argue that this growth rate is independent of $i$, which can only happen if $\dv{s} \log a= \dv{s} \log \check{\tau}$.  

We focus on the term $B_i$ because $A$ is independent of $i$.  $B_i$ can be analyzed as follows:
\begin{align*}
\dv{s}\log \l[k_i a +\check{\tau}\r] = \frac{k_i \dv{a}{s}+\dv{\check{\tau}}{s}}{k_i a +\check{\tau}}=:\phi\l(k_i\r).
\end{align*}
We have 
\begin{align*}
\phi'\l(k_i\r) = \frac{\dv{a}{s} \l(k_i a +\check{\tau}\r) - a \l(k_i \dv{a}{s}+\dv{\check{\tau}}{s}\r) }{\l(k_i a +\check{\tau}\r)^2}= \frac{\dv{a}{s} \check{\tau} - a \dv{\check{\tau}}{s} }{\l(k_i a +\check{\tau}\r)^2}.
\end{align*}
Since $k_i$ is increasing in $i$, it follows that $\dv{s}\log \l[k_i a +\check{\tau}\r]$, and hence also $\dv{s} \log \hat{c}_i$, is either increasing, constant, or decreasing in $i$.  Specifically, $\dv{s} \log \hat{c}_i$ is increasing/decreasing in $i$ if and only if  $\dv{a}{s} \check{\tau} - a \dv{\check{\tau}}{s}$ is positive/negative. In particular, note that $\dv{s} \log \hat{c}_i$ cannot be non-monotonic in $i$ (increasing in $i$ at some values of $i$ and decreasing at others). It now follows from Lemma \ref{Lorenz lemma},\footnote{Lemma \ref{Lorenz lemma} requires that $\hat{c}_i$ be increasing in $i$, which holds because $E$ is admissible, and hence $\check{\tau}\l(E,s\r)$ is order-preserving.} by contraposition, that in order to maintain the equality required along an iso-equality path, $\dv{s}\log \hat{c}_i\l(\check{\tau}\l(E,s\r),s\r)$ must be constant in $i$---here we invoke the definition of $\check{\tau}\l(E,s\r)$. That is, \textit{along an iso-equality path, equivalent consumption changes proportionately for all agents}.  Since $\dv{s}\log \hat{c}_i$ is constant in $i$, $\dv{a}{s} \check{\tau} - a \dv{\check{\tau}}{s}$ must be zero. Equivalently, 
\begin{align}\label{equal rates of growth}
\dv{s} \log a\l(\check{\tau}\l(E,s\r),s\r)   = \dv{s} \log \check{\tau}\l(E,s\r).
\end{align}

We can rewrite (\ref{equal rates of growth}) as:
\begin{align*}
\frac{\pdv{a}{s}}{a} + \frac{\pdv{a}{\tau} }{a} \check{\tau}_s =\frac{\check{\tau}_s}{\check{\tau}}.
\end{align*}
Rearranging terms, and appealing to the definition of $a\l(\tau,s\r)$ in  (\ref{average utility per income}), we have 
\begin{align*}
\check{\tau}_s = &\;\frac{\l.\pdv{a}{s}\r/a}{\frac{1}{\check{\tau}}- \l.\pdv{a}{\tau}\r/a}\\
=&\; \frac{\frac{\frac{\varepsilon}{1+\varepsilon}\check{\tau}}{\l(1-\check{\tau}\r)-\frac{\varepsilon}{1+\varepsilon}\l(1-s\check{\tau}\r)}}{\frac{1}{\check{\tau}}-\frac{-1 +\frac{\varepsilon}{1+\varepsilon} s}{\l(1-\check{\tau}\r)-\frac{\varepsilon}{1+\varepsilon}\l(1-s\check{\tau}\r)}}\\
=&\; \frac{\frac{\varepsilon \check{\tau}}{\l(1+\varepsilon\r)\l(1-\check{\tau}\r) - \varepsilon\l(1-s\check{\tau}\r)}}{\frac{1}{\check{\tau}}-\frac{-\l(1+\varepsilon\r) +\varepsilon s}{\l(1+\varepsilon\r)\l(1-\check{\tau}\r) - \varepsilon\l(1-s\check{\tau}\r)}},
\end{align*}
where the last equality is derived by multiplying the numerator and denominator of the two embedded fractions by $\l(1+\varepsilon\r)$.  Multiplying the numerator and denominator of the last expression in the derivation by $\check{\tau}\l[\l(1+\varepsilon\r)\l(1-\check{\tau}\r) - \varepsilon\l(1-s\check{\tau}\r)\r]$, and noting that the denominator then simplifies to $\l(1+\varepsilon\r)\l(1-\check{\tau}\r) - \varepsilon\l(1-s\check{\tau}\r)+\check{\tau}\l[\l(1+\varepsilon\r) -\varepsilon s\r]=1$, we obtain (\ref{ODE}). $\square$

\subsubsection{Substitution and income effects on efficiency}\label{app: substitution income effects efficiency}
This section establishes the signs of the substitution and income effects of a change in salience on efficiency.  An increase in salience reduces efficiency through the income effect and increases efficiency through the substitution effect, so that the overall effect is ambiguous.  

Let $\mu^{\mathrm{rel}}\l(s,\delta\r)$ be the optimal level of efficiency in the relaxed version of the welfare maximization problem in Section \ref{Sec: Decomposing}.  We have
\begin{align}\label{def: mu rel}
\mu^{\mathrm{rel}}\l(s,\delta\r) = \check{\mu}\l(E^{\mathrm{rel}}\l(s,\delta\r),s\r) + \delta,
\end{align}
because the right-hand side gives the level of efficiency that, with the optimal level of equality $E^{\mathrm{rel}}\l(s,\delta\r)$, makes the constraint in the relaxed problem bind.  Define
\begin{align}\label{compensated definition efficiency}
\mu^{\mathrm{c}}\l(s_0,s_1\r) := \mu^{\mathrm{rel}}\l(s_1,\delta\l(s_0,s_1\r)\r).
\end{align}
$\mu^{\mathrm{c}}\l(s_0,s_1\r)$ is the \textbf{compensated demand for efficiency}, analogous to the compensated demand for equality defined by (\ref{compensated definition}).    

The next proposition establishes and defines the substitution and income effects on efficiency.  It is analogous to Lemma \ref{lem:decomposition}, which does the same for equality.  
 
\begin{prop}\label{prop: decomposition efficiency}Let $\mu\l(s\r)$ be the optimal level of efficiency at $s$.  Then the change in optimal efficiency can be decomposed as:
$\displaystyle \l. \dv{\mu}{s}\r|_{s =s_0}=\underbrace{\l.\pdv{\mu^\mathrm{c}}{s_1}\r|_{s_1=s_0}}_{\text{substitution effect}} +\underbrace{\l.\pdv{\mu^{\mathrm{rel}}}{\delta}\r|_{\substack{s = s_0 \\ \delta = 0}}\times \l.\pdv{\check{\mu}}{s}\r|_{\substack{E= E\l(s_0\r) \\ s=s_0}}}_{\text{income effect}}.$
\end{prop}
\textit{Proof.} The argument is parallel to the proof of Lemma \ref{lem:decomposition} in Appendix \ref{app: decomposition}, with $\mu$ taking the role of $E$.  In particular, differentiating equation (\ref{compensated definition efficiency}), we have:
\begin{align*}
\l.\pdv{\mu^\mathrm{c}}{s_1}\r|_{s_1=s_0} =\l.\pdv{\mu^{\mathrm{rel}}}{s}\r|_{\substack{s = s_0 \\ \delta = 0}}  + \l.\pdv{\mu^{\mathrm{rel}}}{\delta}\r|_{\substack{s = s_0 \\ \delta = 0}}   \times \l.\pdv{\delta}{s_1}\r|_{s_1=s_0} 
=  \l. \dv{\mu}{s}\r|_{s =s_0} -\l.\pdv{\mu^{\mathrm{rel}}}{\delta}\r|_{\substack{s = s_0 \\ \delta = 0}}\times \l.\pdv{\check{\mu}}{s}\r|_{\substack{E= E\l(s_0\r) \\ s=s_0}},
\end{align*}
where the second equality follows from the fact that $\mu\l(s\r) = \mu^{\mathrm{rel}}\l(s,0\r)$ and (\ref{exact compensation}).  The result follows by rearranging terms to solve for $\dv{\mu}{s}$. $\square$ 

We next establish the signs of the substitution and income effects on efficiency.  

\begin{prop}
The substitution effect on efficiency in Proposition \ref{prop: decomposition efficiency} is positive.
\end{prop}
\textit{Proof.}  We have:
\begin{align*}
\mu^c\l(s_0,s_1\r) = \check{\mu}\l(E^{\mathrm{rel}}\l(s_1,\delta\l(s_0,s_1\r)\r),s_1\r) + \delta\l(s_0,s_1\r) = \check{\mu}\l(E^{\mathrm{c}}\l(s_0,s_1\r),s_1\r) + \delta\l(s_0,s_1\r),
\end{align*}
where the first equality follows from (\ref{def: mu rel}) and (\ref{compensated definition efficiency}) and the second follows from (\ref{compensated definition}).  Using the fact that $E^{\mathrm{c}}\l(s_0,s_0\r) = E\l(s_0\r)$, it follows that
\begin{align*}
\l.\pdv{\mu^\mathrm{c}}{s_1}\r|_{s_1=s_0}= \l.\pdv{\check{\mu}}{E}\r|_{\substack{E= E\l(s_0\r) \\ s=s_0}}  \l.\pdv{E^\mathrm{c}}{s_1}\r|_{s_1=s_0} + \underbrace{\l.\pdv{\check{\mu}}{s}\r|_{\substack{E= E\l(s_0\r) \\ s=s_0}}+ \l.\pdv{\delta}{s_1}\r|_{s_1=s_0}}_0 
= -p_E\l(E\l(s_0\r),s_0\r)\l.\pdv{E^\mathrm{c}}{s_1}\r|_{s_1=s_0} > 0,
\end{align*}
where the canceling out of the two bracketed terms labeled with a zero follows from (\ref{exact compensation}), the last equality uses the definition of the price of equality (\ref{definition of pE}), and the inequality follows from the fact that the price of equality is positive (see (\ref{expression for pE}) and Proposition \ref{prop: partials}) and that the substitution effect on equality is negative (see Lemmas \ref{lem:substitution} and \ref{price of equality lemma}). $\square$      

\begin{prop}\label{prop: income effect efficiency}
The income effect on efficiency in Proposition \ref{prop: decomposition efficiency} is negative. 
\end{prop}
In order to establish the proposition, it is useful to first prove a lemma.  
\begin{lem}\label{lem: mu E E} At any equality level $E$ that is admissible at $s$ and corresponds to a positive tax rate, 
$\check{\mu}_{EE}\l(E,s\r) < 0$. 
\end{lem}
\textit{Proof.} Let $E$ be an equality level that is admissible at $s$, and let $\tau =\check{\tau}\l(E,s\r)>0$ be the corresponding admissible tax rate.  Let $q_i$ be a rescaling of equivalent consumption given by $q_i:=\hat{c}_i/\l(1-s\tau\r)^\varepsilon$.  Since $\tau$ is admissible, $1-s\tau > 0$, and hence $q_i$ is well-defined and positive.  Then, reflecting the scale invariance of the equality measure, 
\begin{align*}
E =  \frac{\l(\int \hat{c}^{1-\rho}_i \dd i\r)^{\frac{1}{1-\rho}}}{\int \hat{c}_i \dd i} = \frac{\l(\int q^{1-\rho}_i \dd i\r)^{\frac{1}{1-\rho}}}{\int q_i \dd i}.
\end{align*}
Let $N := \l(\int q^{1-\rho}_i \dd i\r)^{\frac{1}{1-\rho}}$ and $D:=\int q_i \dd i$ be, respectively, the numerator and denominator of the last fraction.  It follows from (\ref{hat c i in terms of w i}) that 
\begin{align*}
q_i =\frac{1-\tau\l(1+\varepsilon\l(1-s\r)\r)}{1+\varepsilon} w^{1+\varepsilon}_i +\tau \int w_j^{1+\varepsilon}\dd j = \beta_i \tau + \gamma_i,
\end{align*} 
where $\beta_i :=\int w_j^{1+\varepsilon}\dd j
-
\frac{1+\varepsilon\l(1-s\r)}{1+\varepsilon}w_i^{1+\varepsilon} $ and $\gamma_i :=\frac{1}{1+\varepsilon}w_i^{1+\varepsilon}$. Note that neither $\beta_i$ nor $\gamma_i$ depends on $\tau$, so that $q_i$ is affine in $\tau$.  It follows that $D_{\tau\tau}=0$.  

From above, we have $N = ED$.  So differentiating with respect to $\tau$, $N_\tau = E_\tau D + E D_\tau$.  Differentiating again with respect to $\tau$, $N_{\tau\tau} = E_{\tau\tau} D + 2 E_\tau D_\tau+ E D_{\tau\tau}= E_{\tau\tau} D + 2 E_\tau D_\tau$, where the last equality follows from $D_{\tau\tau}=0$.  Solving for $E_{\tau\tau}$, we have:
\begin{align}\label{Etautau ND}
E_{\tau\tau} = \frac{N_{\tau\tau}-2 E_\tau D_\tau}{D}.
\end{align}
Proposition \ref{prop: partials} tells us that $E_\tau > 0$.  Next observe that $D_\tau = \int \beta_i \dd i =\int \l[\int w_j^{1+\varepsilon}\dd j
-\frac{1+\varepsilon\l(1-s\r)}{1+\varepsilon}
w_i^{1+\varepsilon}\r]\dd i
= \l[1-\frac{1+\varepsilon\l(1-s\r)}{1+\varepsilon}\r]
\int w_i^{1+\varepsilon}\dd i
= \frac{\varepsilon s}{1+\varepsilon}
\int w_i^{1+\varepsilon}\dd i>0$.   We have $N_\tau = N^\rho \int \beta_i q_i^{-\rho} \dd i$.  So 
\begin{align*}
N_{\tau\tau} =& \rho N^{\rho-1} N^\rho \l(\int \beta_i q_i^{-\rho} \dd i\r)^2 -\rho N^\rho \int \beta^2_i q_i^{-\rho-1} \dd i \\
=& \rho N^\rho\l[\l(\int q_i^{1-\rho} \dd i\r)^{-1}\times \l(\int \beta_i q_i^{-\rho} \dd i\r)^2- \int \beta^2_i q_i^{-\rho-1} \dd i \r]. 
\end{align*}
Noting that $q_i > 0$, we have $ \beta_i q_i^{-\rho} = \l(\beta_i q_i^{-\frac{\rho+1}{2}}\r) \times \l(q_i^{\frac{1-\rho}{2}}\r) $.  It now follows from the Cauchy-Schwarz inequality that 
\begin{align*}
N_{\tau\tau} \leq \rho N^\rho\l[\l(\int q_i^{1-\rho} \dd i\r)^{-1}\times \int \beta^2_i q_i^{-\rho-1} \dd i \times  \int  q_i^{1-\rho} \dd i - \int \beta^2_i q_i^{-\rho-1} \dd i \r] =0.
\end{align*} 
To summarize the inequalities derived above, $E_\tau > 0, D_\tau > 0$, and $N_{\tau\tau} \leq 0$.  We also know that $D >0$.  It now follows from (\ref{Etautau ND}) that $E_{\tau\tau} < 0$.   

It follows from (\ref{definition of pE}), (\ref{expression for pE}), and Proposition \ref{prop: partials} that $\check{\mu}_E = \frac{\mu_\tau\l(\check{\tau}\l(E,s\r),s\r)}{E_\tau\l(\check{\tau}\l(E,s\r),s\r)}<0$. Hence,
\begin{align*}
\check{\mu}_{EE}\l(E,s\r)=\pdv{E}\l[\frac{\mu_\tau\l(\check{\tau}\l(E,s\r),s\r)}{E_\tau\l(\check{\tau}\l(E,s\r),s\r)}\r]=\frac{\mu_{\tau\tau}E_\tau-\mu_\tau E_{\tau\tau}}{E_\tau^2}\check{\tau}_E=\frac{\mu_{\tau\tau}E_\tau-\mu_\tau E_{\tau\tau}}{E_\tau^3} <0,
\end{align*}
where the inequality follows from the fact that $\mu_{\tau\tau}< 0, E_\tau > 0, \mu_\tau < 0, E_{\tau\tau} < 0$ (see Proposition \ref{prop: partials}, (\ref{mutautau}), and our above argument about $E_{\tau\tau}$).  $\square$
 
\noindent \textit{Proof of Proposition \ref{prop: income effect efficiency}.}  The first order condition for an optimal solution of the relaxed problem is
\begin{align*}
\check{\mu}\l(E,s\r)+\delta + E\check{\mu}_E\l(E,s\r)=0,
\end{align*}
and is satisfied at $E=E^{\mathrm{rel}}\l(s,\delta\r)$.    Hence (\ref{def: mu rel}) implies
\begin{align*}
\mu^{\mathrm{rel}}\l(s,\delta\r)= -E^{\mathrm{rel}}\l(s,\delta\r) \check{\mu}_E\l( E^{\mathrm{rel}}\l(s,\delta\r),s\r).
\end{align*}   
Hence,   
\begin{align}\label{mu rel inequality}
\l.\pdv{\mu^{\mathrm{rel}}}{\delta}\r|_{\substack{s = s_0 \\ \delta = 0}}= - \l.\pdv{E^{\mathrm{rel}}}{\delta}\r|_{\substack{s = s_0 \\ \delta = 0}}
\times \l(\check{\mu}_E\l(E\l(s_0\r),s_0\r)+E\l(s_0\r)\check{\mu}_{EE}\l(E\l(s_0\r),s_0\r)\r)>0,
\end{align}
where we have used the facts that $\l.\pdv{E^{\mathrm{rel}}}{\delta}\r|_{\substack{s = s_0 \\ \delta = 0}}>0$ (see Lemma \ref{lem: E rel monotonicty}), that $\check{\mu}_E<0$, and that the equality level $E\l(s_0\r)$ is admissible and corresponds to a positive tax rate, and hence, by Lemma \ref{lem: mu E E}, $\check{\mu}_{EE} < 0$. Finally note that by the argument given in Section \ref{sec:income effect}, $\l.\pdv{\check{\mu}}{s}\r|_{\substack{E= E\l(s_0\r) \\ s=s_0}} < 0$.  This last inequality, together with (\ref{mu rel inequality}), implies the result.  $\square$

\section{Reference labor supply: scale invariance and robustness}\label{reference labor appendix}

This appendix discusses the reason for using a reference labor supply of zero for constructing equivalent consumption, as well as robustness to using other reference levels of labor supply.

Our choice of zero as the reference level of labor is motivated by a standard requirement for relative measures of equality: scale invariance. Given our maintained assumption of isoelastic outer utility $u$, Proposition \ref{scale invariant proposition} shows that zero is the unique nonnegative reference level of labor supply for which our equality measure is scale invariant. Under our additional assumption of isoelastic costs of earning income, zero reference labor also ensures that a uniform proportional increase in wages leaves measured equality unchanged, holding the tax rate and salience fixed, whereas this would not be true for other reference levels.

We can also consider alternative reference levels for labor supply if we are willing to sacrifice scale invariance. In this case, we establish substantial robustness of our comparative statics. First, we show that the signs of the partial derivatives $\pdv{E}{s}$, $\pdv{E}{\tau}$, $\pdv{\mu}{s}$, and $\pdv{\mu}{\tau}$ are invariant to the reference level. Second, we turn to the optimal levels of efficiency and equality when the marginal tax rate is chosen to be $s$-optimal, $\mu\l(s\r)$ and $E\l(s\r)$.  The sign and magnitude of the marginal effect of salience on efficiency at the optimal tax, $\mu'\l(s\r)$, are independent of the reference level. Determining the sign of $E'\l(s\r)$ is more involved. When $0 < \rho \leq 1$, then $E'\l(s\r) < 0$, as in Theorem \ref{equality theorem}, regardless of the reference level. When $\rho>1$, the same conclusion holds under a distributional assumption on wages that is satisfied by the lognormal and Pareto distributions, as well as the Pareto-lognormal specification used by \citeasnoun{heathcote2021optimal}.

\subsection{Definitions}

Let $\ell \geq 0$ be labor supply. Define $\hat{v}\l(\ell\r) := \frac{\ell^{1+\frac{1}{\varepsilon}}}{1+\frac{1}{\varepsilon}}$. Then, by (\ref{form cost of earning income}), $v_i\l(z_i\r) = \hat{v}\l(\frac{z_i}{w_i}\r)$.   Define $\hat{U}\l(c_i,\ell_i\r) := u\l(c_i-\hat{v}\l(\ell_i\r)\r)$.  Note that there is no subscript on $\hat{U}$ because agents differ in their wages but not their preferences over consumption and leisure.  Then we can redefine equivalent consumption relative to an arbitrary nonnegative level of labor supply $\hat{\ell}$. This yields $\hat{c}^{\hat{\ell}}_i\l(\tau,s\r)$, which is defined implicitly by:
 \begin{equation*}
\hat{U}\l(\hat{c}^{\hat{\ell}}_i\l(\tau,s\r),\hat{\ell}\r) = \hat{U}\l(z_i\l(\tau,s\r)\l(1-\tau\r) + \tau \bar{z}\l(\tau,s\r),\frac{z_i\l(\tau,s\r)}{w_i}\r).
\end{equation*}
The definition we use in the paper is the special case of this when $\hat{\ell}=0$.  Whereas the paper's version of equivalent consumption $\hat{c}_i^0\l(\tau,s\r) = \l[\l(1-\tau\r)w_i\ell_i +\tau \int w_j \ell_j \dd j\r] - \hat{v}\l(\ell_i\r)$ is the difference between the agent's consumption and their labor supply cost, the generalized notion of equivalent consumption $\hat{c}_i^{\hat{\ell}}\l(\tau,s\r) = \l[\l(1-\tau\r)w_i\ell_i +\tau \int w_j \ell_j \dd j\r] - \l[\hat{v}\l(\ell_i\r)-\hat{v}\l(\hat{\ell}\r)\r]$ is similar, but it resets the level of labor supply cost as the difference from cost at level $\hat{\ell}$. Accordingly $\hat{c}^0_i = 0$ at zero income and zero labor supply, whereas $\hat{c}^{\hat{\ell}}_i = \hat{v}(\hat{\ell})$ at zero income and zero labor supply. In this sense, the measure used in the paper is quite natural. In general, $\hat{c}^{\hat{\ell}}_i = \hat{c}^0_i+\hat{v}(\hat{\ell})$.   

We can also generalize the notions of efficiency, equally distributed equivalent, and equality, using $\hat{c}_i^{\hat{\ell}}$ rather than $\hat{c}_i^0$.  Efficiency $\mu^{\hat{\ell}}$ is defined by:
 \begin{equation*}
\mu^{\hat{\ell}}\l(\tau,s\r) =  \int_0^1 \hat{c}^{\hat{\ell}}_i\l(\tau,s\r) \dd i.  
\end{equation*}
The equally distributed equivalent $\xi^{\hat{\ell}}$ is  defined as the level of consumption, along with labor supply $\hat{\ell}$, that could be given equally to everyone such that social welfare would be unchanged:   
\begin{equation}\label{equally distributed equivalent hat ell}
\int_0^1 \hat{U}\l(\xi^{\hat{\ell}}\l(\tau,s\r),\hat{\ell}\r) \dd i = W\l(\tau,s\r).
\end{equation}
Equality $E^{\hat{\ell}}$ is defined by:
\begin{equation}\label{equality hat ell}
E^{\hat{\ell}}\l(\tau,s\r) = \frac{\xi^{\hat{\ell}}\l(\tau,s\r)}{\mu^{\hat{\ell}}\l(\tau,s\r)} 
\end{equation}
These definitions imply the following relationships between the new notions of equality and efficiency relative to $\hat{\ell}$ and the notions of efficiency and equally distributed equivalent in the main text:\footnote{\label{ell footnote}To see these relations, note that $\mu^{\hat{\ell}} = \int_0^1 \hat{c}_i^{\hat{\ell}} \dd i = \int_0^1 (\hat{c}_i^0+\hat{v}(\hat{\ell})) \dd i= \int_0^1 \hat{c}_i^0 \dd i +\hat{v}(\hat{\ell}) = \mu^0 +\hat{v}(\hat{\ell})$. (\ref{equally distributed equivalent hat ell}) can equivalently be written as $\int_0^1 u(\xi^{\hat{\ell}}\l(\tau,s\r)-\hat{v}(\hat{\ell})) \dd i=W$, which is equivalent to $\xi^{\hat{\ell}} = u^{-1}\l(W\r) + \hat{v}(\hat{\ell}) = \xi^0 + \hat{v}(\hat{\ell})$. The expression for $E^{\hat{\ell}}$ now follows from (\ref{equality hat ell}).}
\begin{align}\label{ell in terms of 0}
\mu^{\hat{\ell}}\l(\tau,s\r) = \mu^0\l(\tau,s\r) + \hat{v}(\hat{\ell}) \qquad  E^{\hat{\ell}}\l(\tau,s\r) = \frac{\xi^0\l(\tau,s\r)+\hat{v}(\hat{\ell})}{\mu^0\l(\tau,s\r)+\hat{v}(\hat{\ell})}
\end{align}
In what follows, we will omit the $0$ superscript and write $\hat{c}_i, \mu, \xi,$ and $E$ for $\hat{c}^0_i, \mu^0, \xi^0,$ and $E^0$. 

\subsection{Scale Invariance}

We assume, as in the text, that utility $u$ takes the isoelastic form (\ref{isoelastic form}). Then, using (\ref{ell in terms of 0}), if we write the equality measure as a function of the equivalent consumption profile $\hat{C}^{\hat{\ell}} = (\hat{c}_i^{\hat{\ell}})_{i \in \l[0,1\r]}$, we have:
\begin{equation*}
E^{\hat{\ell}}\l(\hat{C}^{\hat{\ell}}\r) = \frac{\l(\int \hat{c}_i^{1-\rho}\dd i\r)^{\frac{1}{1-\rho}} + \hat{v}(\hat{\ell})}{\int \hat{c}_i \dd i + \hat{v}(\hat{\ell})} = \frac{\l(\int \l(\hat{c}^{\hat{\ell}}_i -\hat{v}(\hat{\ell}) \r)^{1-\rho}\dd i\r)^{\frac{1}{1-\rho}} + \hat{v}(\hat{\ell})}{\int \hat{c}^{\hat{\ell}}_i \dd i}.
\end{equation*}
Say that $\hat{C}^{\hat{\ell}}$ is \textbf{permissible} when $\hat{c}_i^{\hat{\ell}} > \hat{v}(\hat{\ell}), \forall i$.  When $\hat{C}^{\hat{\ell}}$ is not permissible, $E^{\hat{\ell}}(\hat{C}^{\hat{\ell}})$ may not be well-defined.  Lemma \ref{increasing normal lemma} and Proposition \ref{prop: positivity} imply that any $\hat{C}^{\hat{\ell}}$ that is induced at $s$ by a tax rate admissible at $s$ is permissible.  Next, for any $\lambda > 0$, define $\lambda \hat{C}^{\hat{\ell}} =(\lambda \hat{c}_i^{\hat{\ell}})_{i \in \l[0,1\r]}$.  $E^{\hat{\ell}}$ is \textbf{scale invariant} if, for all permissible $\hat{C}^{\hat{\ell}}$, and for all $\lambda > 0$ such that $\lambda \hat{C}^{\hat{\ell}}$ is permissible,  $E^{\hat{\ell}}(\lambda \hat{C}^{\hat{\ell}}) = E^{\hat{\ell}}(\hat{C}^{\hat{\ell}})$.  Recall that reference labor supply satisfies $\hat{\ell} \geq 0$.

\begin{prop}\label{scale invariant proposition}
$E^{\hat{\ell}}$ is scale invariant if and only if $\hat{\ell} =0$.
\end{prop}
Proof.  It is clear that $E^0$ is scale invariant.

Next suppose that $\hat{\ell} \neq 0$.  Then $\hat{\ell} > 0$.  So also $\hat{v}(\hat{\ell}) > 0$. Choose permissible $\hat{C}^{\hat{\ell}}$  and $\lambda > 0$ such that (i) $\hat{c}_i$ is strictly increasing in $i$ and (ii) $\lambda \hat{c}_i^{\hat{\ell}}>\hat{v}\l(\hat{\ell}\r)$, for all $i$, so that $\lambda \hat{C}^{\hat{\ell}}$ is permissible.  We have:
\begin{equation*}
E^{\hat{\ell}}\l(\lambda \hat{C}^{\hat{\ell}}\r) = \frac{\l(\int \l(\lambda \hat{c}^{\hat{\ell}}_i -\hat{v}\l(\hat{\ell}\r) \r)^{1-\rho}\dd i\r)^{\frac{1}{1-\rho}} + \hat{v}\l(\hat{\ell}\r)}{\int \lambda  \hat{c}^{\hat{\ell}}_i \dd i}.
\end{equation*}
Letting $x_i:= \lambda \hat{c}^{\hat{\ell}}_i -\hat{v}\l(\hat{\ell}\r) $ and $X:= \int x_i^{1-\rho}\dd i$, we have 
\begin{equation}\label{E lambda 1}
\begin{split}
\dv{\lambda}E^{\hat{\ell}}\l(\lambda \hat{C}^{\hat{\ell}}\r) =& \frac{\lambda \mu^{\hat{\ell}} X^{\frac{\rho}{1-\rho}} \int x^{-\rho}_i\hat{c}^{\hat{\ell}}_i \dd i -\mu^{\hat{\ell}}\l(X^{\frac{1}{1-\rho}} + \hat{v}\l(\hat{\ell}\r)\r)}{\l(\lambda \mu^{\hat{\ell}}\r)^2}\\
=& \frac{ X^{\frac{\rho}{1-\rho}} \int x^{-\rho}_i\l(x_i + \hat{v}\l(\hat{\ell}\r)\r) \dd i -X^{\frac{1}{1-\rho}} - \hat{v}\l(\hat{\ell}\r)}{\lambda^2 \mu^{\hat{\ell}}}\\
=& \frac{\hat{v}\l(\hat{\ell}\r)\l[X^{\frac{\rho}{1-\rho}}\int x_i^{-\rho} \dd i-1\r] }{\lambda^2 \mu^{\hat{\ell}}}.
\end{split}
\end{equation}
where the second equality uses $ \lambda \hat{c}^{\hat{\ell}}_i =x_i + \hat{v}\l(\hat{\ell}\r)$.  Next observe that  
\begin{equation}\label{E lambda 2}
\int x_i^{-\rho} \dd i = \int \l(x_i^{1-\rho}\r)^{-\frac{\rho}{1-\rho}} \dd i > \l(\int x_i^{1-\rho}\dd i \r)^{-\frac{\rho}{1-\rho}} = X^{-\frac{\rho}{1-\rho}}, \end{equation}
where the inequality follows from the fact that, for $y> 0$,  $y \mapsto y^{-\frac{\rho}{1-\rho}}$ is strictly convex when $\rho > 0, \rho \neq 1$, that $x_i$ is strictly increasing in $i$, and Jensen's inequality. (When $\rho=1$, the analogous argument appeals to the strict convexity of $y \mapsto e^{-y}$, applied to $\int e^{-y_i} \dd i$, where $y_i=\log x_i$.)  It follows from (\ref{E lambda 1}) and (\ref{E lambda 2}), that, when $\hat{\ell} > 0$, and hence $\hat{v}(\hat{\ell}) > 0$, we have $\dv{\lambda}E^{\hat{\ell}}(\lambda \hat{C}^{\hat{\ell}}) > 0$. It follows that $E^{\hat{\ell}}$ is not scale invariant, as, scale invariance implies $\dv{\lambda}E^{\hat{\ell}}\l(\lambda \hat{C}^{\hat{\ell}}\r)=0$.  $\square$ 

Related to the last result, consider the effect of a proportional change in wages on equality, holding the tax rate $\tau$ and salience $s$ fixed.  In particular, suppose that all wages $w_i$ are multiplied by a constant $\gamma > 0,\gamma \neq 1$.  Then because of our isoelastic specification for the cost of earning income, it follows from (\ref{hat c i in terms of w i}) that zero-reference equivalent consumption is scaled by $\gamma^{1+\varepsilon}$.  Hence, under isoelastic outer utility $u$, $\xi$ and $\mu$ are scaled by $\gamma^{1+\varepsilon}$.  And equality becomes $\widetilde{E}^{\hat{\ell}} =\frac{\gamma^{1+\varepsilon} \xi + \hat{v}\l(\hat{\ell}\r)}{\gamma^{1+\varepsilon} \mu+ \hat{v}\l(\hat{\ell}\r)}$.  When $\hat{\ell} =0$, $\hat{v}(\hat{\ell})=0$ and hence scaling wages has no effect on the ratio.  But when $\hat{\ell}> 0$, $\hat{v}(\hat{\ell})>0$, and since $\xi < \mu$, scaling wages changes the ratio.  Hence a proportional change in wages leaves equality unchanged if and only if the reference level of labor $\hat{\ell}$ is equal to zero.     

\subsection{Partial effects}

Having established the scale-invariance justification for zero reference labor supply, we next examine how the comparative-static results depend on the reference level. Specifically, we generalize Proposition \ref{prop: partials} in the main text to the case where $\hat{\ell} > 0$. 

\begin{prop}\label{prop: partials ell}Assume that $\tau$ is admissible at $s$ and that $\tau >0$. Then, when equivalent consumption is defined relative to labor supply level $\hat{\ell}$, the signs of the partial derivatives for equality and efficiency are as follows: $ \pdv{\mu^{\hat{\ell}}}{s} < 0, \pdv{\mu^{\hat{\ell}}}{\tau} < 0, \pdv{E^{\hat{\ell}}}{s}<0$, and $\pdv{E^{\hat{\ell}}}{\tau} > 0$. 
\end{prop}
\begin{rem}
The proposition shows that, for arbitrary $\hat{\ell} >0$, the signs of partial derivatives for equality and efficiency remain the same as those we established in Proposition \ref{prop: partials} for the case where reference level for labor is $\hat{\ell}=0$. 
\end{rem}

\noindent \textit{Proof of Proposition \ref{prop: partials ell}.} The partial derivatives of $\mu^{\hat{\ell}}$ are unchanged when $\hat{\ell} > 0$ because, by (\ref{ell in terms of 0}), $\mu^{\hat{\ell}}$ is a translation of $\mu$, the signs (and even magnitudes) of its partial derivatives do not change.    

For equality, the following lemma is useful:
\begin{lem}\label{conditional inequalities lemma} Assume that $\tau > 0$ and $\tau$ is admissible at $s$. (1) If $\pdv{\mu}{s}-\pdv{\xi}{s}   \geq 0$, then $\pdv{E^{\hat{\ell}}}{s} < 0$ for all $\hat{\ell} \geq 0$. (2) If $\pdv{\mu}{\tau}-\pdv{\xi}{\tau}  \leq 0$, then $\pdv{E^{\hat{\ell}}}{\tau} > 0$ for all $\hat{\ell} \geq 0$. 
\end{lem}
\textit{Proof.} For $\theta\in\l\{s,\tau\r\}$, since $E=\frac{\xi}{\mu}$, we have: $\pdv{E}{\theta} =\frac{\pdv{\xi}{\theta}\mu-\xi\pdv{\mu}{\theta}}{\mu^2}$. Abbreviating $\hat{v}=\hat{v}\l(\hat{\ell}\r)$, and appealing to (\ref{ell in terms of 0}) and the expression for $\pdv{E}{\theta}$ just derived,
\begin{equation}\label{E ell and 0 relation}
\begin{split}
\pdv{E^{\hat{\ell}}}{\theta} &= \frac{ \pdv{\xi}{\theta}\l(\mu+\hat{v}\r) - \l(\xi+\hat{v}\r)\pdv{\mu}{\theta}}{\l(\mu+\hat{v}\r)^2}= \frac{\l(\pdv{\xi}{\theta}\mu -\xi\pdv{\mu}{\theta}\r)-\hat{v}\l(  \pdv{\mu}{\theta}-\pdv{\xi}{\theta}
\r)}{\l(\mu+\hat{v}\r)^2}= \l(\frac{\mu}{\mu + \hat{v}}\r)^2\pdv{E}{\theta} -  \frac{\hat{v}\l(  \pdv{\mu}{\theta}
-\pdv{\xi}{\theta}\r)}{\l(\mu+\hat{v}\r)^2}. 
\end{split}
\end{equation}
Since $\hat{v} \geq 0$, and, by Proposition \ref{prop: partials}, for $s,\tau$, such that $\tau > 0$ and $\tau$ is admissible at $s$, $\pdv{E}{s} < 0$ and $\pdv{E}{\tau} > 0$, the lemma now follows.  $\square$ 

We now proceed with the proof of Proposition \ref{prop: partials ell}, establishing first that $\pdv{E^{\hat{\ell}}}{s} < 0$ and then that $\pdv{E^{\hat{\ell}}}{\tau} > 0$.  Choose $s$ and $\tau >0$ so that $\tau$ is admissible at $s$.  
We use the facts that (1) $\xi=u^{-1}\l(\int u\l(\hat{c}_i\r) \dd i\r)$, (2), by Lemma \ref{increasing normal lemma}, $\hat{c}_i$ is strictly increasing in $i$, and so by the strict concavity of $u$, $\l[1-\frac{u'(\hat{c}_i)}{u'(\xi)}\r]$ is also strictly increasing in $i$ and (3), by (\ref{ci derivative s}), $\pdv{\hat{c}_i}{s}= \frac{\varepsilon\tau^2}{1-s\tau}\l[(1-s)z_i-\bar{z}\r]$ is nondecreasing in $i$. It follows from the Chebyshev integral inequality that 
\begin{equation}\label{xi-mu inequality}
\begin{split}
\mu_s-\xi_s = \int \pdv{\hat{c}_i}{s} \dd i - \frac{1}{u'\l(\xi\r)} \int u'\l(\hat{c}_i\r) \pdv{\hat{c}_i}{s} \dd i &\;\geq \l(\int \l[1-\frac{u'\l(\hat{c}_i\r) }{u'\l(\xi\r)}\r] \dd i\r)
\l(\int \pdv{\hat{c}_i}{s} \dd i\r)\\ &\;= \l(1-\frac{\int u'\l(\hat{c}_i\r) \dd i}{u'\l(\xi\r)}\r)\mu_s.
\end{split}
\end{equation} 
Next observe that:
\begin{align*}
\int u'\l(\hat{c}_i\r) \dd i = \int \hat{c}_i^{-\rho} \dd i = \int  \l(\hat{c}_i^{1-\rho} \r)^{\frac{-\rho}{1-\rho}} \dd i >&\; \l(\int  \hat{c}_i^{1-\rho}\dd i  \r)^{\frac{-\rho}{1-\rho}} =\l(\l[\l(1-\rho\r)\int  \frac{\hat{c}_i^{1-\rho}}{1-\rho}\dd i  \r]^{\frac{1}{1-\rho}}\r)^{-\rho} \\
=&\; u'\l(u^{-1}\l(W\r)\r) = u'\l(\xi\r), \end{align*}
where the strict inequality follows from the fact that $y \mapsto y^{-\frac{\rho}{1-\rho}}$ is convex when $\rho > 0, \rho\neq 1$, Jensen's inequality and the fact that $\hat{c}_i$ is strictly increasing in $i$. (When $\rho=1$, the analogous argument appeals to the strict convexity of $y \mapsto e^{-y}$, applied to $\int e^{-y_i} \dd i$, where $y_i=\log \hat{c}_i$.) It follows that $1-\frac{\int u'\l(\hat{c}_i\r) \dd i}{u'\l(\xi\r)} <0$.  From Proposition \ref{prop: partials}, it follows that $\mu_s < 0$.  It now follows from (\ref{xi-mu inequality}) that $\mu_s -\xi_s > 0$.  So, by Lemma \ref{conditional inequalities lemma}, $\pdv{E^{\hat{\ell}}}{s} < 0.$

To complete the proof, we need to show that $\pdv{E^{\hat{\ell}}}{\tau} > 0$.  Using the fact that $\xi = E \times \mu$, we have:
\begin{equation*}
\xi_\tau - \mu_\tau = \pdv{\l(E \times \mu\r)}{\tau}- \mu_\tau = \l(E_\tau \mu + E \mu_\tau \r)-\mu_\tau = \underbrace{E_\tau}_+ \underbrace{\mu}_+ +\underbrace{\l(E-1\r)}_-\underbrace{\mu_\tau}_- \geq 0,   
\end{equation*}
where the signs follow from Proposition \ref{prop: partials} and the fact that $\mu \geq 0$ and $E \leq 1$.  That $\pdv{E^{\hat{\ell}}}{\tau} > 0$ now follows from Lemma \ref{conditional inequalities lemma}. $\square$ 

\subsection{Extending Theorem \ref{equality theorem} to alternative reference labor levels}

In this section, we extend Theorem \ref{equality theorem} to alternative reference labor levels.  When $0< \rho \leq 1$, we can extend the theorem without further assumptions. When $\rho > 1$, we impose an assumption on the wage distribution.  Let $F_w$ and $f_w$ denote the cumulative distribution function and density of wages. In this case, we assume that the elasticity of the cumulative wage distribution $\frac{w f_w\l(w\r)}{F_w\l(w\r)}$ is nonincreasing in the wage $w$ wherever $F_w\l(w\r)>0$.  Equivalently,  the cumulative distribution function of log wages $y\mapsto F_w\l(e^y\r)$ is log-concave. The assumption is satisfied by a number of familiar distributions such as the uniform, lognormal, and Pareto distributions, the Pareto-lognormal specification used by \citeasnoun{heathcote2021optimal}, and the generalized beta distribution of the second kind (GB2). It is preserved under truncation of the wage distribution from below or above.  

Let $\xi^{\hat{\ell}}\l(s\r) = \xi^{\hat{\ell}}\l(\tau\l(s\r),s\r),  E^{\hat{\ell}}\l(s\r) = E^{\hat{\ell}}\l(\tau\l(s\r),s\r),$ and $  \mu^{\hat{\ell}}\l(s\r) = \mu^{\hat{\ell}}\l(\tau\l(s\r),s\r)$ be equally distributed equivalent, equality, and efficiency evaluated at state $s$ and the $s$-optimal tax rate $\tau\l(s\r)$.

By (\ref{ell in terms of 0}), $\mu^{\hat{\ell}}\l(s\r)=\mu\l(s\r)+\hat{v}\l(\hat{\ell}\r)$. So, changing the fixed reference labor level leaves the sign and magnitude of the effect of salience on optimal efficiency unchanged: 
\begin{align}\label{reference labor optimal salience}
\dv{s}\mu^{\hat{\ell}}\l(s\r)=\mu'\l(s\r),\;\;\; \forall \hat{\ell} \geq 0. 
\end{align}
Consequently, regardless of the reference level, the effect of salience on optimal efficiency remains ambiguous. Next, we turn to equality. 
\begin{prop}\label{equality theorem exogenous}
Suppose one of the following conditions holds:

(i) $0<\rho\leq1$.

(ii) $\rho>1$ and $y\mapsto F_w\l(e^y\r)$ is log-concave.

\noindent Then, for all $s$ and $\hat{\ell}\geq0$, $\dv{s} E^{\hat{\ell}}\l(s\r) < 0$.\footnote{The condition on the wage distribution in (ii) was suggested as a sufficient condition by ChatGPT. ChatGPT provided the initial proof of Proposition \ref{equality theorem exogenous}, which extends our result on equality to arbitrary reference levels $\hat{\ell} \geq 0$. The authors reorganized and refined the proof.} 
\end{prop}

\begin{rem}
Both under condition (i) and under condition (ii), equality falls with salience and the effect on efficiency remains ambiguous, regardless of the reference labor level $\hat{\ell} \geq 0$ used to define equivalent consumption.\footnote{\label{ambiguity reference wage footnote} For condition (i), the simulations in Appendix \ref{app:simulations} with $\rho=0.1$ illustrate that $\mu'\l(s\r)$ switches from negative to positive as salience increases. For condition (ii), the same income-distribution calibration with $\rho=1.5$ and $\varepsilon=2$ also yields such a sign change; the wage distribution is a truncated lognormal, for which the CDF of log wages is log-concave. By (\ref{reference labor optimal salience}), each sign change also occurs for $\dv{s}\mu^{\hat{\ell}}\l(s\r)$ for every reference level $\hat{\ell}\geq0$.}  So under each condition, we reproduce both the part of our result concerning equality and the part concerning efficiency. 
\end{rem}

\begin{rem}\label{endogenous condition remark}
For all $s$ and $\hat{\ell}\geq 0$, if $\tau'\l(s\r) \leq 0 $, then $\dv{s}E^{\hat{\ell}}\l(s\r) < 0$.\footnote{By  Theorem \ref{equality theorem}, $ E'\l(s\r)< 0$.  It follows from (\ref{extended ds E}) that to establish the claim in the first sentence of Remark \ref{endogenous condition remark}, it is sufficient to show that, if $\tau'\l(s\r) \leq 0$, then $\mu'\l(s\r)-\xi'\l(s\r)  \geq 0$.  We have
$\mu'\l(s\r) -\xi'\l(s\r)   = \l( \mu_\tau \tau'\l(s\r) + \mu_s  \r)-\l( \xi_\tau \tau'\l(s\r) + \xi_s  \r)=  \underbrace{\l( \mu_s-\xi_s\r)}_+ + \underbrace{\tau'\l(s\r)}_- \underbrace{\l(\mu_\tau-\xi_\tau\r)}_- \geq 0,$ where the signs follow from the assumption that $\tau'\l(s\r) \leq 0$,  and arguments made in the course of the proof of Proposition \ref{prop: partials ell}. 
}  So, when $\rho > 1$, $\tau'\l(s\r) \leq 0$ provides an alternative (endogenous) sufficient condition for extending  Theorem \ref{equality theorem} to alternative reference labor levels.\footnote{If $\tau'\l(s\r) \leq 0$ holds over some interval $I$ of salience levels, then the conclusion $\dv{s}E^{\hat{\ell}}\l(s\r) < 0$ holds over that interval, and if it holds over all salience levels, then $\dv{s}E^{\hat{\ell}}\l(s\r) < 0$ holds everywhere.}  When $\tau'\l(s\r) \leq 0$, the sign of $\dv{s}\mu^{\hat{\ell}}\l(s\r)$ remains ambiguous.\footnote{In the example for condition (ii) in footnote \ref{ambiguity reference wage footnote}, there is an interval in which the sign of $\dv{s}\mu^{\hat{\ell}}\l(s\r)$ changes and $\tau'\l(s\r) \leq 0$.} That the marginal tax rate falls with salience is the most natural assumption, but it is not in general implied by our model---see Section \ref{efficiency ambiguous} for a discussion. 
\end{rem}
Let $\hat{c}_i\l(s\r)=\hat{c}_i\l(\tau\l(s\r),s\r)$ be $i$'s equivalent consumption at the $s$-optimal tax rate.  A key to establishing Proposition \ref{equality theorem exogenous} is showing that, under the proposition's assumptions,    
\begin{align}\label{absolute increasing differences}
\dv{s}\l[\hat c_j\l(s\r)-\hat c_i\l(s\r)\r]>0 \textup{ for every } s\in\l(0,1\r] \textup{ and } i<j. 
\end{align}
 In other words, at the $s$-optimal tax rate, increasing salience moves the equivalent consumption of any pair of agents further apart in an absolute sense.  
 
Let $k_i = \frac{w_i^{1+\varepsilon}}{\int w_j^{1+\varepsilon} \dd j}$, $a\l(s\r)=  \frac{1-\tau\l(s\r)\l(1+\varepsilon\l(1-s\r)\r)}{1+\varepsilon}$, $A\l(s\r)=  a\l(s\r) \bar{z}\l(s\r), b\l(s\r) = \tau\l(s\r) \bar{z}\l(s\r)$, and recall that $\bar{z}\l(s\r) =  \l(1-s\tau\l(s\r)\r)^\varepsilon \int w_j^{1+\varepsilon} \dd j$.  

It follows from (\ref{hat c i in terms of w i}) that $\hat{c}_i\l(s\r) = A\l(s\r)k_i + b\l(s\r)$.  So $\hat{c}_j\l(s\r)-\hat{c}_i\l(s\r) =A\l(s\r) \l[k_j-k_i\r]$. It follows that (\ref{absolute increasing differences}) is equivalent to $A'\l(s\r) > 0$.  The following two lemmas imply Proposition \ref{equality theorem exogenous}.

\begin{lem}\label{A prime lemma}
If $A'\l(s\r) > 0$, then $\dv{s} E^{\hat{\ell}}\l(s\r) < 0$.
\end{lem}
\begin{lem}\label{distributional lemma} Suppose one of the following conditions holds: (i) $0<\rho\leq1$, or (ii) $\rho>1$ and $y\mapsto F_w\l(e^y\r)$ is log-concave. Then, for all $s$, $A'\l(s\r) > 0$. 
\end{lem}
\textit{Proof of Lemma \ref{A prime lemma}.} Throughout the proof, we use primes to denote the $s$ derivatives where the tax rate adjusts to remain optimal, so that, for example, $\xi' = \xi'\l(s\r)$.  

Arguing as in (\ref{E ell and 0 relation}),  we have:
\begin{equation}\label{extended ds E}
\dv{s} E^{\hat{\ell}}\l(s\r) = \l(\frac{\mu\l(s\r)}{\mu\l(s\r)+\hat{v}(\hat{\ell})}\r)^2E'\l(s\r) - \hat{v}(\hat{\ell})\frac{\l( \mu'\l(s\r)-\xi'\l(s\r)
\r)}{\l(\mu\l(s\r)+\hat{v}(\hat{\ell})\r)^2}.
\end{equation}
Assume that $A' > 0$.  By (\ref{extended ds E}), to establish the lemma, it is sufficient to show that $\mu' - \xi' \geq 0$.
 
We now develop convenient expressions for $\mu'$ and $\xi'$.  We have $\mu = \int \hat{c}_i \dd i = \int \l(A k_i + b\r) \dd i= A+b$.  Hence $\mu' = A' + b'$. We also have $\xi = \l(\int \hat{c}_i^{1-\rho} \dd i\r)^{\frac{1}{1-\rho}}$.  It follows that 
\begin{equation}\label{xi prime alternative expression}
\xi' = \l(\int \hat{c}_i^{1-\rho} \dd i\r)^{\frac{\rho}{1-\rho}} \int \hat{c}_i^{-\rho}\hat{c}'_i \dd i = \xi^{\rho} \int g_i \l(A' k_i + b'\r) \dd i = \xi^{\rho} \bar{g} \l(A' \bar{h} + b'\r),
\end{equation}  
where we have used the fact that $\bar{h} = \int \frac{g_i}{\bar{g}} \frac{z_i}{\bar{z}} \dd i = \int \frac{g_i}{\bar{g}} \frac{\l(1-s\tau\r)^\varepsilon w_i^{1+\varepsilon}}{\int \l(1-s\tau\r)^\varepsilon w_j^{1+\varepsilon} \dd j} \dd i  =   \frac{\int g_ik_i  \dd i}{\bar{g}} $.  
It follows that 
\begin{equation*}
\mu' - \xi' = \l(A' + b'\r) -\xi^{\rho} \bar{g} \l(A' \bar{h} + b'\r) =  \l(1 - \xi^{\rho} \bar{g} \bar{h}\r)A'  + \l(1-\xi^{\rho} \bar{g} \r) b'.  
\end{equation*}
Since the lemma assumes that $A' > 0$, to establish that $\mu' - \xi' \geq 0$, it is sufficient to show that: (I) $b' < 0$, (II) $\xi^{\rho} \bar{g} \geq 1$, and (III) $\xi^{\rho} \bar{g} \bar{h} \leq 1$.  Theorem \ref{salience vs welfare proposition} implies that $\xi' < 0$.  Hence, since $A' >0$, it follows from (\ref{xi prime alternative expression}) that $b' < 0$.   Next observe that, at the $s$-optimal tax policy, 
\begin{equation*}
\bar{g} = \int \hat{c}_i^{-\rho} \dd i = \int \l(\hat{c}_i^{1-\rho}\r)^{-\frac{\rho}{1-\rho}} \dd i \geq \l(\int \hat{c}_i^{1-\rho} \dd i \r)^{-\frac{\rho}{1-\rho}} = \xi^{-\rho},
\end{equation*}
where the inequality follows from Jensen's inequality and the fact that for $\rho > 0, \rho \neq 1$, the map $y \mapsto y^{-\frac{\rho}{1-\rho}}$ is convex. (When $\rho=1$, the analogous argument appeals to the convexity of $y \mapsto e^{-y}$, applied to $\int e^{-y_i} \dd i$, where $y_i=\log \hat{c}_i$.)  It follows that $\xi^{\rho} \bar{g} \geq 1$.  Next observe that, by definition, $\xi^{1-\rho} = \int \hat{c}_i^{1-\rho} \dd i$. Hence:
\begin{align*}
\xi =&\; \xi^{\rho}\int \hat{c}_i^{1-\rho} \dd i= \xi^{\rho}  \int \hat{c}_i^{-\rho} \hat{c}_i\dd i = \xi^{\rho}\int g_i \hat{c}_i \dd i\\ =&\; \xi^{\rho} \int g_i \l(A k_i + b\r) \dd i = \xi^{\rho}\bar{g} \int  \l(A \frac{g_ik_i}{\bar{g}} + b \frac{g_i}{\bar{g}}\r) \dd i = \xi^{\rho}\bar{g} \l(A \bar{h} + b\r),
\end{align*}
Since $\mu \geq \xi$, it now follows that $A + b=\mu \geq \xi= \xi^{\rho}\bar{g} \l(A \bar{h} + b\r)$. Since all the terms are positive and $\xi^{\rho}\bar{g} \geq 1$, it must be the case that $\xi^{\rho}\bar{g} \bar{h} \leq 1$.  We have now established the conditions (I), (II), and (III), which are jointly sufficient for $\mu'-\xi' \geq 0$, completing the proof of the lemma. $\square$

\bigskip

\noindent \textit{Proof of Lemma \ref{distributional lemma}.} Recall from Section \ref{relation c z subsection} that $a\l(\tau,s\r) =\frac{1-\tau\l(1+\varepsilon\l(1-s\r)\r)}{1+\varepsilon} $ and $b\l(\tau,s\r) = \tau \bar{z}\l(\tau,s\r)$.  Define $A\l(\tau,s\r) := a\l(\tau,s\r) \bar{z}\l(\tau,s\r)$.  Then $A\l(s\r)=A\l(\tau\l(s\r),s\r)$.  Our above decomposition of equivalent consumption  applies to all admissible $\l(\tau,s\r)$ with $\tau>0$: $\hat{c}_i\l(\tau,s\r)= A\l(\tau,s\r)k_i  +b\l(\tau,s\r)=A\l(\tau,s\r)\l[k_i + \frac{\tau}{a\l(\tau,s\r)}\r]$.  $Ak_i$ is the individual-specific part of equivalent consumption and $A \frac{\tau}{a}$ is the common part.  $\frac{\tau}{a}$ is then the ratio of common part of equivalent consumption to the scale factor $A$; recall that  $A$ determines the size of absolute differences in equivalent consumption, via $\hat{c}_j-\hat{c}_i=A\l(k_j-k_i\r)$.        

Next note that we can write $\hat{c}_i = A\l(k_i + e^t\r)$, where $t=\log\l(\frac{\tau}{a}\r)$.  Note that $a$ is positive and decreasing in $\tau$ and hence $t$ is increasing in $\tau$ for any fixed $s$.  This change of variables is useful because, in case (ii), it puts the problem in a form convenient for applying the assumption on the wage distribution.  Given the above, we may formally define $\tau^\circ\l(t,s\r)$, the value of $\tau$ that corresponds to $t$ at $s$,   implicitly by the equation $t = \log\l(\frac{\tau^\circ\l(t,s\r)}{a\l(\tau^\circ\l(t,s\r),s\r)}\r)$.

Next, define $A\l(t,s\r)=A\l(\tau^\circ\l(t,s\r),s\r)$ and $L\l(t,s\r)=\log A\l(t,s\r)$. We can also define $\xi\l(t,s\r)=\xi\l(\tau^\circ\l(t,s\r),s\r)$. Let $t\l(s\r)$ be the value of $t$ that maximizes $\xi\l(t,s\r)$.  Observe that $A\l(s\r) = A\l(t\l(s\r),s\r)$. Then we have:
\begin{equation}\label{A prime over A}
\frac{A'\l(s\r)}{A\l(s\r)}=L_t t'\l(s\r)+L_s.
\end{equation}
We want to show that $A'\l(s\r)>0$. Since $A\l(s\r)>0$, it is sufficient to show that $L_t t'\l(s\r)+L_s>0$.

Next we derive an expression for $t'\l(s\r)$.  Note that $t\l(s\r)$ maximizes $\log \xi\l(t,s\r)$.  We have $\log \xi\l(t,s\r) =\log  \l(\int \hat{c}_i^{1-\rho} \dd i\r)^{\frac{1}{1-\rho}}= L\l(t,s\r) + \Phi\l(t\r)$, where $\Phi\l(t\r)=\frac{1}{1-\rho} \log \l(\int \l(k_i + e^t\r)^{1-\rho} \dd i\r)$.  (Note that $\Phi\l(t\r)$ does not depend on $s$.) It follows from Proposition \ref{prop: uniqueness} that $t\l(s\r)$ is implicitly defined by the first order condition: $L_t + \Phi_t=0$.  It follows from Lemma \ref{strictly concave W} that $L_{tt}+\Phi_{tt} < 0$.\footnote{Lemma \ref{strictly concave W} implies that, at the optimum, $W_{\tau\tau} < 0$.  Also, since $W=u\l(\xi\r)$, and, at the optimum, $\xi_\tau=0$, at the optimum, we also have $\xi_{\tau\tau}=
\frac{1}{u'\l(\xi\r)}W_{\tau\tau}<0$. So, again appealing to $\xi_\tau=0$, at the optimum, 
$\xi_{tt}=\xi_{\tau\tau}\l(\tau^\circ_t\r)^2<0$, and, since $\xi_t=0$ and $\xi>0$, we also have $\l[\log\xi\r]_{tt}=\xi_{tt}/\xi<0$.}  It follows from the implicit function theorem applied to the first order condition that $t'\l(s\r) = -\frac{L_{st}}{L_{tt}+\Phi_{tt}}$, where we have used the fact that $\Phi$ does not depend on $s$.  

It now follows from (\ref{A prime over A}) that $\frac{A'}{A}= \frac{ L_sL_{tt}  -L_t L_{st}+L_s \Phi_{tt}}{L_{tt}+\Phi_{tt}}$.  Since the denominator is negative, to complete the proof, it is sufficient to show that $ L_sL_{tt}  -L_t L_{st}+L_s \Phi_{tt}< 0$.  

Next we develop convenient expression for $L$ that helps us to evaluate the above inequality.  Straightforward algebra starting from the definition of $\tau^\circ\l(t,s\r)$ implies that $\tau^\circ\l(t,s\r) = \frac{e^t}{Q}$ where $Q:=1+\varepsilon +\l[1+\varepsilon\l(1-s\r)\r]e^t$.  Moreover substituting $\tau^\circ=\frac{e^t}{Q}$ in the expression of $a$ and simplifying leads to $a\l(t,s\r) = a\l(\tau^\circ\l(t,s\r),s\r) = \frac{1}{Q}$.  We have $1-s\tau^\circ = \frac{\l(1+\varepsilon\r)P}{Q}$, where $P=1+\l(1-s\r)e^t$. Hence, $A= a\bar{z} = \frac{1}{Q} \l(1-s\tau\r)^\varepsilon\int w_i^{1+\varepsilon}\dd i= \int w_i^{1+\varepsilon}\dd i \l(1+\varepsilon\r)^\varepsilon \frac{P^\varepsilon}{Q^{1+\varepsilon}}$.  So $L = \log A= \log\l(\int w_i^{1+\varepsilon}\dd i\l(1+\varepsilon\r)^\varepsilon\r) + \varepsilon\log P - \l(1+\varepsilon\r)\log Q$, where the first term in the sum does not depend on $s$ or $t$.     

Next we develop an expression for $L_sL_{tt}  -L_t L_{st}$, which is a term in the desired inequality.  We have that $L_s=-\varepsilon \frac{e^t}{P}+\l(1+\varepsilon\r)\frac{\varepsilon e^t}{Q} = \varepsilon e^t\frac{\l(1+\varepsilon\r)P-Q}{PQ}= -\frac{\varepsilon s e^{2t}}{PQ}$, and $L_t=\varepsilon\frac{\l(1-s\r)e^t}{P}
-\l(1+\varepsilon\r)\frac{\l[1+\varepsilon\l(1-s\r)\r]e^t}{Q}
=-\frac{e^tD}{PQ}$, where $D:= 1+\varepsilon +\l(1-s\r)\l[1+\varepsilon\l(1-s\r)\r]e^t$.  Hence $\frac{L_s}{L_t} = \frac{\varepsilon s e^t}{D}$.  Next, observe that $L_sL_{tt}  -L_t L_{st} = -L_t^2 \pdv{\l(L_s/L_t\r)}{t}$.  Since $L_t^2 =\frac{e^{2t}D^2}{P^2Q^2}$ and  $\pdv{\l(L_s/L_t\r)}{t}= \frac{\varepsilon\l(1+\varepsilon\r) s e^t}{D^2}$, it follows that $L_s L_{tt}-L_t L_{st}= -\frac{\varepsilon \l(1+\varepsilon\r) s e^{3t}}{P^2 Q^2}$.  

It now follows that $L_sL_{tt}-L_tL_{st}+L_s\Phi_{tt} =
-\frac{\varepsilon s e^{2t}}{PQ}\l[\Phi_{tt}+\frac{\l(1+\varepsilon\r)e^t}{PQ}\r]$.  Since $-\frac{\varepsilon s e^{2t}}{PQ}<0$ and $\frac{\l(1+\varepsilon\r)e^t}{PQ}>0$, to complete the proof it is sufficient to establish
$\Phi_{tt}\geq 0, \forall t$. 

\medskip

\noindent \textit{Case (i): $0 < \rho \leq 1$.} We begin by assuming that $\rho < 1$ and then return to the case that $\rho=1$.  Consider $t_1, t_2$ and $\lambda \in\l[0,1\r]$.  Since
$t\mapsto \log\l(k_i+e^t\r)$ is convex,\footnote{This follows as a special case of the derivation in footnote \ref{convexity footnote} by fixing the first argument at $\log k_i$.}
\begin{align}\label{inequality from convexity}
\l(k_i+e^{\lambda t_1+\l(1-\lambda\r)t_2}\r)^{1-\rho} \leq\l(k_i+e^{t_1}\r)^{\l(1-\rho\r)\lambda}\l(k_i+e^{t_2}\r)^{\l(1-\rho\r)\l(1-\lambda\r)}.
\end{align}
Applying (\ref{inequality from convexity}) pointwise and then H\"older's inequality gives:
\begin{align*}
\exp\l(\l(1-\rho\r)\Phi\l(\lambda t_1+\l(1-\lambda\r)t_2\r)\r) &=\int_0^1
\l(k_i+e^{\lambda t_1+\l(1-\lambda\r)t_2}\r)^{1-\rho}
\dd i\\
&\leq \int_0^1
\l(k_i+e^{t_1}\r)^{\l(1-\rho\r)\lambda} \l(k_i+e^{t_2}\r)^{\l(1-\rho\r)\l(1-\lambda\r)} \dd i\\
&\leq \l[ \int_0^1 \l(k_i+e^{t_1}\r)^{\l(1-\rho\r)}\dd i \r]^\lambda \l[ \int_0^1 \l(k_i+e^{t_2}\r)^{\l(1-\rho\r)}\dd i \r]^{1-\lambda}\\
&=\exp\l(\l(1-\rho\r)\Phi\l(t_1\r)\r)^\lambda \exp\l(\l(1-\rho\r)\Phi\l(t_2\r)\r)^{1-\lambda}.
\end{align*}
Hence $\exp\l(\l(1-\rho\r)\Phi\l(t\r)\r)$ is log-convex.  So, since, in this case, $1-\rho>0$, $\Phi\l(t\r)$ is convex.  Hence $\Phi''\l(t\r) \geq 0, \forall t$.

When $\rho=1$, the analogous definition of $\Phi\l(t\r)$ is $\Phi(t)=\int_0^1\log\l(k_i+e^t\r)\dd i$.  In this case, the convexity of $\Phi\l(t\r)$ follows immediately from the convexity of $t\mapsto \log\l(k_i+e^t\r)$.

\medskip

\noindent \textit{Case (ii):  $\rho>1$ and $y\mapsto F_w\l(e^y\r)$ is log-concave.}  Since $\log k_i =\l(1+\varepsilon\r) \log w_i - \log \int w^{1+\varepsilon}_j \dd j$, log-concavity of the cumulative distribution of log wages implies that the cumulative distribution of $\log k$ is log-concave; in other words $y \mapsto F_k\l(e^y\r)$ is log-concave, where $F_k$ is the CDF of $k$. We have $\exp\l(\l(1-\rho\r)\Phi\l(t\r)\r) = \int_0^1 \l(k_i+e^t\r)^{1-\rho}\dd i$.  Since $\rho > 1$, for each $i$, $\l(k_i+e^t\r)^{1-\rho} =\l(\rho-1\r)\int_{k_i}^{\infty}\l(x+e^t\r)^{-\rho}\dd x.$  Hence, interchanging the order of integration,
\begin{align*}
\exp\l(\l(1-\rho\r)\Phi\l(t\r)\r)&=\l(\rho-1\r)\int_0^1\int_{k_i}^{\infty}\l(x+e^t\r)^{-\rho}
\dd x\dd i  = \l(\rho-1\r) \int_0^\infty F_k\l(x\r)
\l(x+e^t\r)^{-\rho} \dd x\\
&= \l(\rho-1\r)\int_{-\infty}^\infty F_k\l(e^y\r) \l(e^y +e^t\r)^{-\rho} e^y \dd y.
\end{align*}
The integrand is $I\l(y,t\r)=F_k\l(e^y\r) \l(e^y +e^t\r)^{-\rho} e^y$.  So $\log I \l(y,t\r) = \log F_k\l(e^y\r)-\rho \log\l(e^y+e^t\r)+y$.  By our assumptions $ \log F_k\l(e^y\r)$ is concave in $y$; and $\log\l(e^y+e^t\r)$ is convex in $\l(y,t\r)$,\footnote{\label{convexity footnote}Consider $\l(y_1,t_1\r)$ and $\l(y_2,t_2\r)$ and $\lambda\in[0,1]$.  By H\"older's inequality, $e^{\lambda y_1+\l(1-\lambda\r)y_2}+e^{\lambda t_1+\l(1-\lambda\r)t_2}\leq \l(e^{y_1}+e^{t_1}\r)^\lambda\l(e^{y_2}+e^{t_2}\r)^{1-\lambda}$.  Taking logs gives $\log\l(e^{\lambda y_1+\l(1-\lambda\r)y_2}+e^{\lambda t_1+\l(1-\lambda\r)t_2}\r) \leq \lambda \log\l(e^{y_1}+e^{t_1}\r)+ \l(1-\lambda\r)\log\l(e^{y_2}+e^{t_2}\r)$, which implies convexity in $\l(y,t\r)$.} and hence $-\rho \log\l(e^y+e^t\r)$ is concave in $\l(y,t\r)$. It follows that $I\l(y,t\r)$ is log-concave. Pr\'ekopa's marginalization theorem says that if $I\l(y,t\r) \geq 0$ is jointly log-concave in $\l(y,t\r)$, then $M\l(t\r)=\int_{-\infty}^{\infty} I\l(y,t\r) \dd y$ is log-concave in $t$, provided that the integral is finite \cite[Theorem~6, p.~342]{prekopa1973logarithmic}. It follows that $\exp\l(\l(1-\rho\r)\Phi\l(t\r)\r)$ is log-concave in $t$.  Hence, since $1-\rho< 0$, $\Phi\l(t\r)$ is convex in $t$.  So for this case as well, $\Phi''\l(t\r) \geq 0, \forall t$. Having established this in both cases, the proof is complete. $\square$

\section{Illustrative numerical examples}\label{app:simulations}

We present here a simple economy to build intuition and illustrate possible outcomes and comparative statics. One important result is that increasing tax salience has an ambiguous impact on efficiency at the optimum, in sharp contrast to the robust trade-off between honesty and equality established by Theorem \ref{equality theorem}.

We assume that, at $s=1$ and $\tau = 0.25$, income is log-normally distributed with mean $\$114,500$ and mean log deviation of income of $0.616$, truncated below at \$1,000 and above at \$2,000,000 to match our assumption of a bounded wage distribution. While purely illustrative, the mean and mean log deviation were chosen to match the United States household income distribution in 2023 \cite{guzman2024income}. As in our model setup (Sections \ref{sec: model} and \ref{sec: decomposition derivation}), utility is of the form:
\vspace{-0.3em}\begin{equation*}\vspace{-0.3em}
    U_i\l(c_i,z_i\r) = u\l(c_i - \frac{\l(\frac{z_i}{w_i}\r)^{1+\frac{1}{\varepsilon}}}{1+\frac{1}{\varepsilon}}\r),
\end{equation*}
where:
  \begin{equation*}u\l(\hat{c}\r)= \frac{\hat{c}^{1-\rho}}{1-\rho},\;\;\; \rho > 0.\end{equation*}
We consider $\varepsilon =0.25$ and $\varepsilon = 0.5$, which are in line with the range of estimates summarized by \citeasnoun{chetty2012bounds}. Using different values of $\varepsilon$ yields similar figures.  We back out a wage distribution using the assumed income distribution, utility function, $\varepsilon$, and initial values of $s$ and $\tau$. Wages are held fixed throughout the simulation, while the income distribution varies as $s$ and $\tau$ change.

We define a grid of values of $s$ and $\rho$. For each combination of $s$ and $\rho$, we calculate the optimal tax, $\tau$. This is the tax that maximizes utilitarian welfare, holding $s$ constant. Mechanically we repeatedly use the first order condition for the optimal tax to define a direction to move the tax rate that raises social welfare. We iterate this process until convergence, and check that the second order condition holds for the utilitarian optimum.

Once this tax rate has been obtained, the value of $s$ directly implies the perceived tax rate, $\tilde{\tau}=s\tau$. It is also straightforward to calculate statistics such as inequality and efficiency as defined in the paper. In every case, the fundamentals such as the utility function and wage distribution are held constant across all combinations of $s$ and $\rho$.

\subsection{Simulation results}\label{app:simresults}

We present the results of these simulations below. Figure \ref{fig:efficiency_equality} shows how equality and efficiency change at the optimal tax rate as $s$ varies. Figure \ref{fig:actual_perceived} shows the actual and perceived tax rates at the optimum. In both cases, we present results for low inequality aversion ($\rho=0.1$) and  high inequality aversion ($\rho=3$), with two different values of the elasticity, $\varepsilon$.

The maximum level of efficiency is attained when the tax rate is set to zero, and depends on the elasticity ($\varepsilon$) because this changes each individual's labor effort costs. With $\varepsilon = 0.25$, the maximum efficiency is approximately \$96,000, and, with $\varepsilon = 0.5$, it is around \$86,000. Our calibration holds the income distribution fixed at a specific value of $\l(\tau,s\r)$ independently of $\varepsilon$ and backs out the wage distribution from (\ref{zi expression}). It follows that, in our calibration, equality at $\tau=0$ depends only on $\rho$ but not $\varepsilon$.\footnote{To see why equality at $\tau=0$ does not depend on $\varepsilon$, observe that (\ref{zi expression}) and (\ref{simplification of an expression s}) together imply that $\frac{\hat{c}_i\l(0,s'\r)}{\hat{c}_j\l(0,s'\r)} = \frac{z_i\l(\tau,s\r)}{z_j\l(\tau,s\r)}$ for all agents $i$ and $j$, salience levels $s, s'$, and tax rates $\tau$. By scale invariance of the equality measure, it follows that, for any fixed $\rho$, the level of equality of equivalent consumption when $\tau=0$ is pinned down by the income distribution at \emph{any} value of $\l(\tau,s\r)$.  In our calibration, the income distribution at $\l(\tau,s\r) =\l(0.25,1\r)$ is held fixed, independently of $\varepsilon$. The level of equality of equivalent consumption when $\tau =0$ is therefore pinned down, regardless of $\varepsilon$.  Note that it is not generally the case that  $\frac{\hat{c}_i\l(\tau',s'\r)}{\hat{c}_j\l(\tau',s'\r)} = \frac{z_i\l(\tau,s\r)}{z_j\l(\tau,s\r)}$ for $\tau' \neq 0$, so equality away from $\tau=0$ is not held fixed as $\varepsilon$ varies.} In our case, setting $\tau=0$ yields $E = 0.94$ when $\rho = 0.1$ and $E = 0.17$ when $\rho = 3$. In this sense, there is much more at stake from a utilitarian perspective when $\rho$ is higher, which leads to more variation in efficiency and equality at the optimum as $s$ changes.

In line with Theorem \ref{equality theorem}, equality falls as $s$ increases. However, efficiency may rise or fall, depending on the values of the parameters. In the case where $\rho=0.1$ and $\varepsilon =0.25$, efficiency falls with $s$ until $s\approx 0.55$ and then rises as $s$ increases and approaches one. In contrast, when $\rho=3$, efficiency falls uniformly with $s$. Since efficiency depends only on the perceived tax rate (Section \ref{efficiency ambiguous}), the minimum level of efficiency and the maximum perceived tax rate must occur at the same value of $s$. Figure \ref{fig:actual_perceived} illustrates this: both occur at $s \approx 0.55$ when $\rho = 0.1$ and $\varepsilon = 0.25$, and similarly coincide when $\varepsilon = 0.5$.  Proposition \ref{rho less than 2 proposition} establishes that $\dv{\tau}{s} < 0$ at $s=1$ when $\rho \leq 2$. Away from $s=1$, however, the sign can reverse: in Figure~\ref{fig:actual_perceived}, panels~(c) and (d), $\dv{\tau}{s} > 0$ for low values of $s$, and we have found similar reversals for smaller values of $\rho$ (e.g., $\rho=0.5$ or $\rho=1$).

\begin{landscape}
\begin{figure}[p]
  \centering

  \begin{subfigure}{0.48\linewidth}
    \caption{Low preference for redistribution ($\rho=0.1$, $\varepsilon=0.25$)}
    \label{fig:muxi_low_left}
    \centering
    \includegraphics[width=\linewidth, trim=1pt 1pt 1pt 1pt, clip]{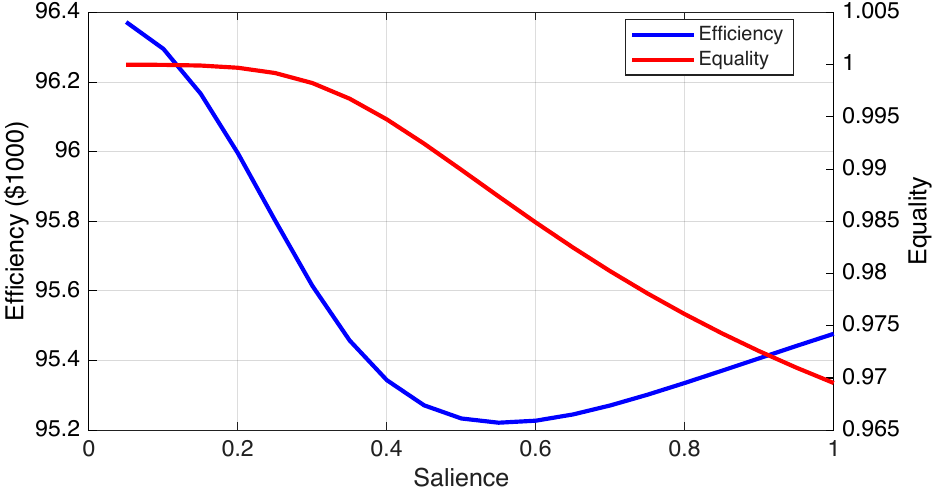}
  \end{subfigure}
  \hfill
  \begin{subfigure}{0.48\linewidth}
    \caption{Low preference for redistribution ($\rho=0.1$, $\varepsilon=0.5$)}
    \label{fig:muxi_low_right}
    \centering
    \includegraphics[width=\linewidth, trim=1pt 1pt 1pt 1pt, clip]{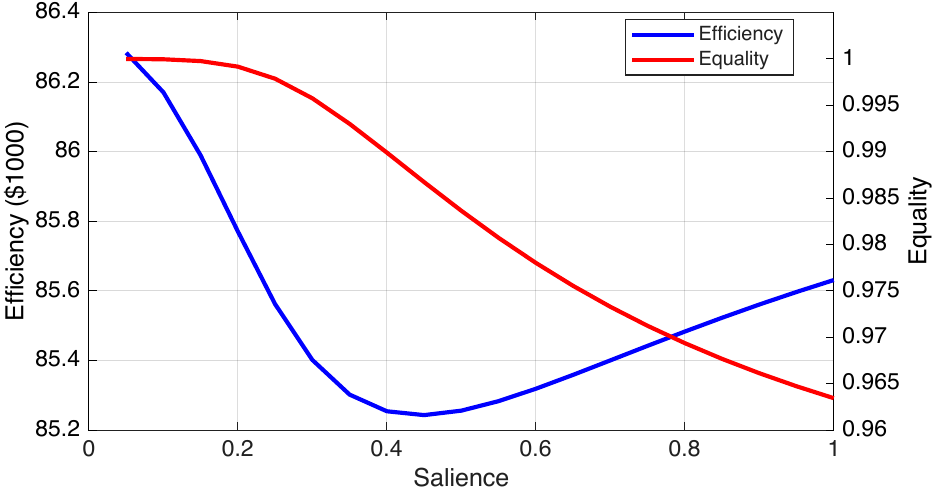}
  \end{subfigure}

  \vspace{0.8em}

  \begin{subfigure}{0.48\linewidth}
    \caption{High preference for redistribution ($\rho=3$, $\varepsilon=0.25$)}
    \label{fig:muxi_high_left}
    \centering
    \includegraphics[width=\linewidth, trim=1pt 1pt 1pt 1pt, clip]{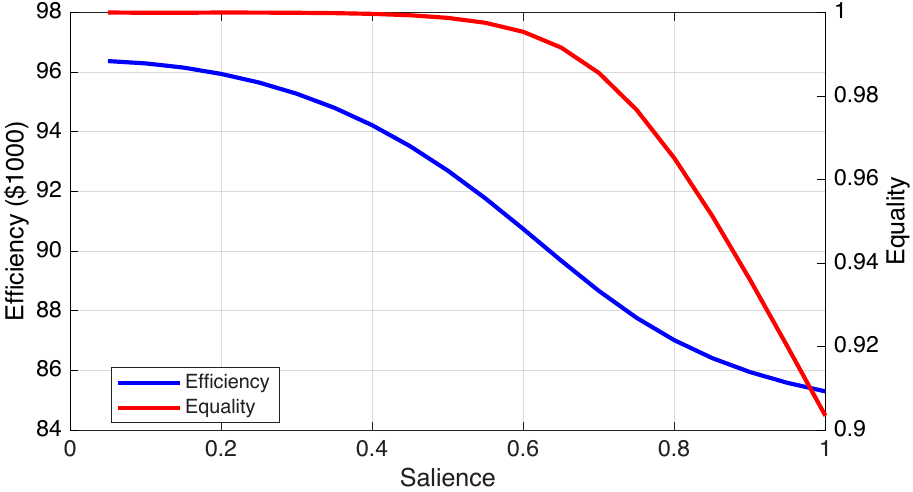}
  \end{subfigure}
  \hfill
  \begin{subfigure}{0.48\linewidth}
    \caption{High preference for redistribution ($\rho=3$, $\varepsilon=0.5$)}
    \label{fig:muxi_high_right}
    \centering
    \includegraphics[width=\linewidth, trim=1pt 1pt 1pt 1pt, clip]{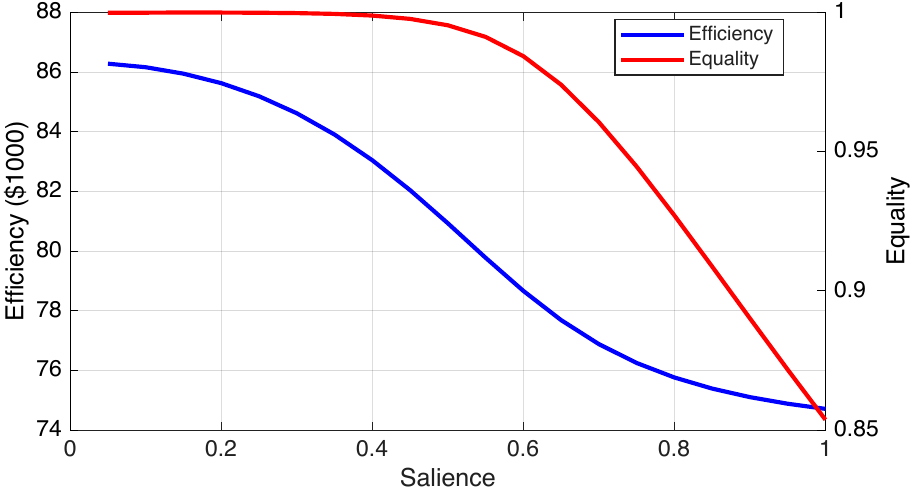}
  \end{subfigure}

  \caption{Efficiency and equality at the optimum}
  \label{fig:efficiency_equality}

  \vspace{0.5em}
  \begin{minipage}{0.95\linewidth}
    \footnotesize \textbf{Note:}
    This figure plots efficiency and equality at the optimum as a function of $s$.
    Panels (a) and (b) show the results for a utilitarian objective with low preference for redistribution ($\rho=0.1$).
    Panels (c) and (d) increase inequality aversion ($\rho=3$).
    As $s$ rises, inequality always rises (Theorem \ref{equality theorem}).
    Efficiency may rise or fall, depending on whether the income effect or substitution effect dominates
    (see Section \ref{efficiency ambiguous}).
  \end{minipage}

\end{figure}
\end{landscape}

\begin{landscape}
\begin{figure}[p]
  \centering

  \begin{subfigure}{0.48\linewidth}
    \caption{Low preference for redistribution ($\rho=0.1$, $\varepsilon=0.25$)}
    \label{fig:actual_perceived_low_left}
    \centering
    \includegraphics[width=\linewidth, trim=1pt 1pt 1pt 1pt, clip]{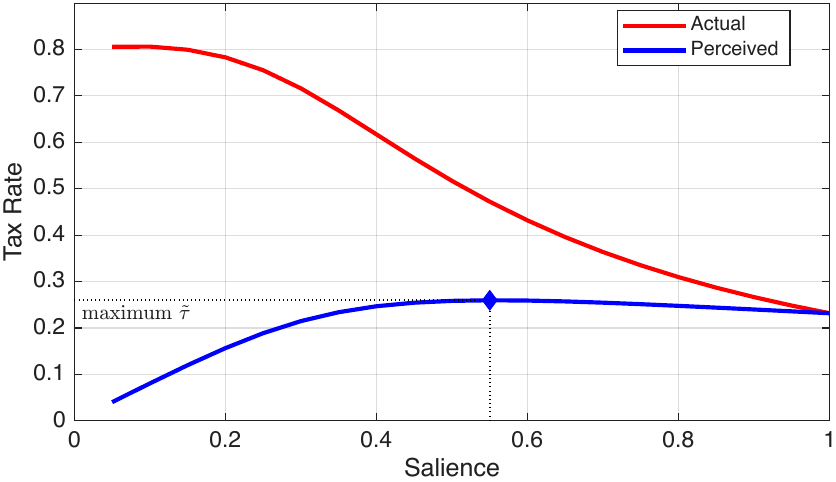}
  \end{subfigure}
  \hfill
  \begin{subfigure}{0.48\linewidth}
    \caption{Low preference for redistribution ($\rho=0.1$, $\varepsilon=0.5$)}
    \label{fig:actual_perceived_low_right}
    \centering
    \includegraphics[width=\linewidth, trim=1pt 1pt 1pt 1pt, clip]{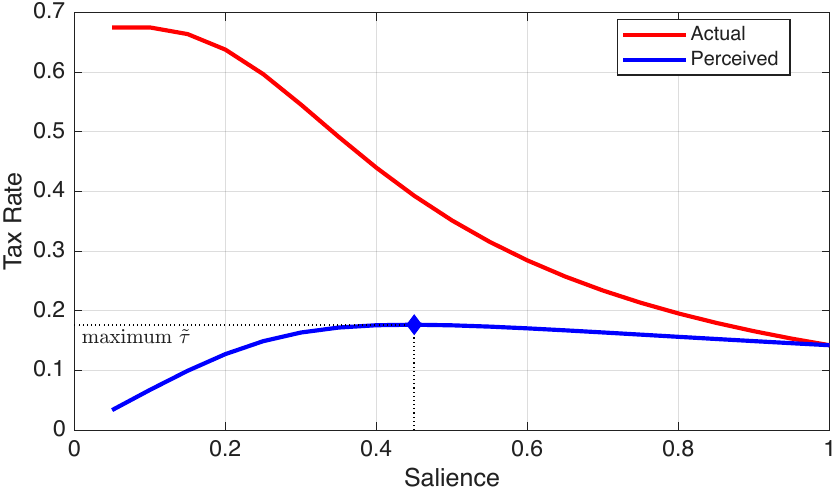}
  \end{subfigure}

  \vspace{0.8em}

  \begin{subfigure}{0.48\linewidth}
    \caption{High preference for redistribution ($\rho=3$, $\varepsilon=0.25$)}
    \label{fig:actual_perceived_high_left}
    \centering
    \includegraphics[width=\linewidth, trim=1pt 1pt 1pt 1pt, clip]{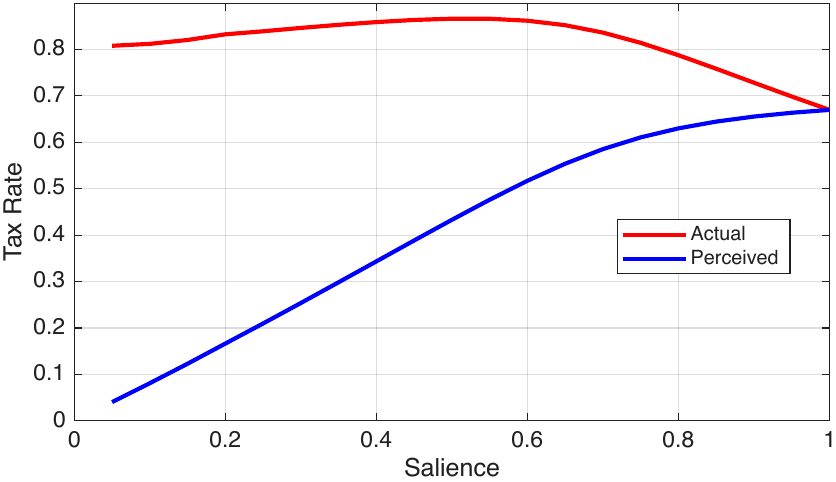}
  \end{subfigure}
  \hfill
  \begin{subfigure}{0.48\linewidth}
    \caption{High preference for redistribution ($\rho=3$, $\varepsilon=0.5$)}
    \label{fig:actual_perceived_high_right}
    \centering
    \includegraphics[width=\linewidth, trim=1pt 1pt 1pt 1pt, clip]{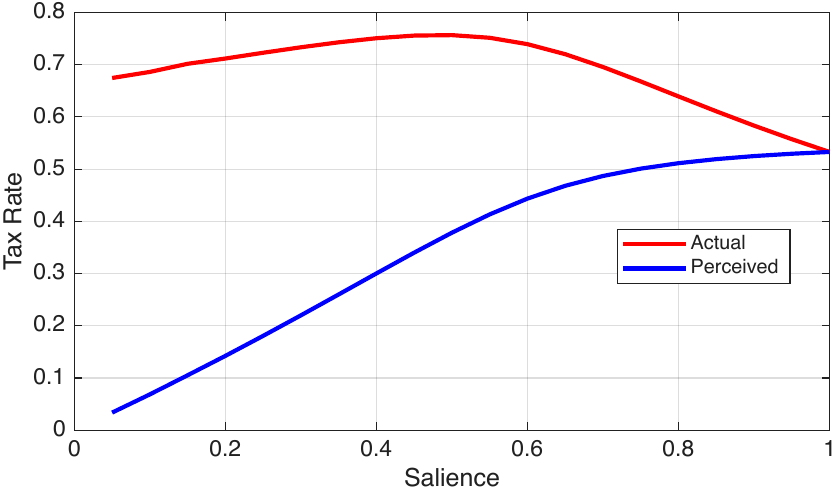}
  \end{subfigure}

  \caption{Actual and perceived tax rates at the optimum}
  \label{fig:actual_perceived}

  \vspace{0.5em}
  \begin{minipage}{0.95\linewidth}
    \footnotesize \textbf{Note:}
    This figure plots the actual and perceived optimal tax rates as a function of $s$.
    Panels (a) and (b) show the results for a utilitarian objective with low inequality aversion ($\rho=0.1$).
    Panels (c) and (d) show the case of high inequality aversion ($\rho=3$).
    Both the actual and perceived tax rates can increase or decrease as $s$ rises.
  \end{minipage}
\end{figure}
\end{landscape}

\section{Alternative interpretation of the model}\label{app: alternative interpretation section}

As outlined in Section \ref{alternative interpretation section}, one interpretation of our model is that it represents a choice over a mix between a salient income tax and a less-salient consumption tax.

As in Section \ref{sec: model}, agents are uniformly distributed on the interval $\l[0,1\r]$.  To formally map the two-tax setup to our model, we suppose that the agent pays a fraction $\tau^{L}$ of their income to the government as tax.  The agent also  pays a tax of $\tau^{C}$ on anything she consumes. Finally, she receives a transfer from the government, $G$. So, the agent faces a budget constraint:
\begin{align}\label{agent budget constraint}
\l(1+\tau^C\r)c_i = \l(1-\tau^L\r) z_i +G,
\end{align}
and, as in Section \ref{sec: model}, the agent's utility is:
\begin{align}\label{agent utility}
U_i = u\l(c_i-v_i\l(z_i\r)\r).
\end{align}

The government uses its tax revenue to cover the universal transfer $G$
and any exogenous expenditures. However, as in our original model of Section \ref{sec: model}, we assume for simplicity that all revenue is rebated lump-sum, and 
the government budget constraint is:
\begin{align}\label{government budget constraint} 
 \tau^L \bar{z} + \tau^C \bar{c} = G,
 \end{align}
where $\bar{c}=\int c_i \dd i$ is aggregate consumption. 

Aggregating individual budget constraints (\ref{agent budget constraint}), we have $\l(1+\tau^C\r)\bar{c} = \l(1-\tau^L\r)\bar{z} +G$. Then eliminating $G$ from this equation via the government budget constraint (\ref{government budget constraint}), it follows that:
\begin{align}\label{c equals z}
\bar{c} = \bar{z}.
\end{align}
This is just resource feasibility: with no outside expenditure, aggregate consumption must equal aggregate output.  

Define $\tau$ by:
\begin{align}\label{retention rate LC}
1-\tau = \frac{1-\tau^L}{1+\tau^C}, 
\end{align}
so that $1-\tau$ is the effective retention rate, in light of both taxes. Equivalently, the combined taxes correspond to facing a marginal tax rate on income of:
\begin{align}\label{tau tauL tauC}
\tau=\frac{\tau^L+\tau^C}{1+\tau^C}
\end{align} 
without an additional consumption tax.
 
Appealing to (\ref{government budget constraint})-(\ref{tau tauL tauC}), it follows that
\begin{align}\label{ci simplification}
\begin{split}
c_i =& \frac{\l(1-\tau^L\r)}{1+\tau^C} z_i + \frac{G}{1+\tau^C} = \frac{\l(1-\tau^L\r)}{1+\tau^C} z_i + \frac{\tau^L \bar{z} +\tau^C \bar{c}}{1+\tau^C} = \frac{\l(1-\tau^L\r)}{1+\tau^C} z_i + \frac{\tau^L +\tau^C}{1+\tau^C} \bar{z} \\ =&\l(1-\tau\r) z_i + \tau \bar{z}.
\end{split}
\end{align}
Plugging (\ref{ci simplification}) into (\ref{agent utility}), it follows that:
\begin{align*}
U_i=u\l(\l(1-\tau\r) z_i  + \tau \bar{z} - v_i\l(z_i\r)\r),
\end{align*} 
which has the same form as utility in our original model, without a consumption tax.

Note that the two taxes $\tau^L$ and $\tau^C$ are equivalent instruments if both are perfectly salient.  But suppose the taxes are not equally salient. Specifically, suppose agents fully perceive the income tax $\tau^L$ but only partially perceive the consumption tax, treating it as if its rate were only $s^C \tau^C$, for some fixed salience level $s^C \in \l(0,1\r)$. Then the perceived effective retention rate is $\frac{1-\tau^{L}}{1+s^C\tau^{C}}$ rather than (\ref{retention rate LC}) because agents misperceive the consumption tax.   Now define $s$ by:
\begin{align*}
1-s\tau =  
 \frac{1-\tau^{L}}{1+s^C\tau^{C}}.
\end{align*} 
That is, $s$ is the salience level in our original model that produces the same perceived incentives as with both a consumption and income tax, where the consumption tax has salience $s^C$.  Since the perceived effective retention rate exceeds the actual effective retention rate, $s \in \l(0,1\r)$.    

We can also go in the other direction.  Consider a consumption tax salience level $s^C \in \l(0,1\r)$ and imagine that this is fixed and cannot be altered.     
Now choose a marginal tax rate $\tau$ in our original model  and a salience level $s$.  Then it is possible to create a situation that is equivalent to one with marginal tax rate $\tau$ and a salience level $s$ in our original model by choosing $\tau^L$ and $\tau^C$ in the model of the current section according to the following equations:
\begin{align*}
\tau^L =&\;\tau \l[1-\frac{\l(1-\tau\r) \l(1-s\r)}{\l(1-\tau\r)-s^C\l(1-s \tau\r)}\r]&
\tau^C =&\; \frac{\tau-s\tau}{\l(1-\tau\r)-s^C\l(1-s \tau\r)}.
\end{align*}
For a fixed effective marginal tax rate $\tau$, for both $\tau^L$ and $\tau^C$ to be nonnegative requires that
\begin{align}\label{salience inequality}
s \geq \frac{s^C}{1-\tau+s^C \tau}.
\end{align}
The binding constraint is that $\tau^L \geq 0$ because to minimize salience subject to a fixed effective tax rate $\tau$, one wants to set $\tau^L =0$.  Note that $\frac{s^C}{1-\tau+s^C \tau} = \frac{s^C\l(1-\tau\r) + s^C \tau}{1-\tau+s^C \tau}< 1$, implying that a non-trivial range of salience levels $s$ is achievable.
  
Thus, we can think of the government as achieving a given pair $\l(\tau,s\r)$ in our model satisfying (\ref{salience inequality}) by choosing a mix of a salient income tax and a less salient consumption tax, where the salience of the consumption tax $s^C$ is fixed.

\clearpage
\singlespacing
\bibliographystyle{aer}
\bibliography{salience}

\end{document}